\documentclass[10pt]{article}

\usepackage[utf8]{inputenc}

\usepackage[T1]{fontenc} 
\usepackage{mathtools,amsmath,amssymb,amsfonts,amsthm,siunitx,bm,bbm}
\usepackage{moreverb,rotating,graphics,esint}
\usepackage{mathrsfs}      
\usepackage{enumitem} 
\usepackage{titlesec}
\usepackage{subfig}

\usepackage[a4paper, margin=2.5cm]{geometry}

\usepackage{dsfont}

\numberwithin{equation}{section}

\allowdisplaybreaks

\usepackage[colorlinks=true,linkcolor=blue,citecolor=blue,urlcolor=blue,breaklinks]{hyperref}

\usepackage{graphicx} 
\usepackage{mathtools}
\usepackage{mathabx}
\usepackage{bbm}
\usepackage{bigints}
\usepackage[dvipsnames]{xcolor}
\usepackage[toc,page]{appendix}
\usepackage{comment} 

\usepackage{multirow}

\usepackage{tikz}
\usetikzlibrary{decorations.pathreplacing}

\numberwithin{equation}{section}

\newtheorem{theorem}{Theorem}[section]
\newtheorem{proposition}[theorem]{Proposition}
\newtheorem{remark}[theorem]{Remark}
\newtheorem{lemma}[theorem]{Lemma}
\newtheorem{corollary}[theorem]{Corollary}

\newtheorem{assumption}[theorem]{Assumption}

\usepackage{bigints}
\DeclareMathOperator*{\fiint}{\ensuremath{\iint\text{\kern-1.36em{\raisebox{5.87pt}{\rotatebox{-93}{$\setminus$}}}}}}

\def\N{{\mathbb N}}
\def\Z{{\mathbb Z}}
\def\R{{\mathbb R}}
\def\C{{\mathbb C}}

\def\eps{\varepsilon}

\def\P{{\mathbb P}}
\def\1{{\mathds{1}}}

\def\cH{{\mathcal H}}

\def\cN{{\mathcal N}}

\def\cS{{\mathcal S}}
\def\cT{{\mathcal T}}

\def\cB{{\mathcal B}}

\def\cU{{\mathcal U}}

\def\cO{{\mathcal O}}

\def\ba{{\mathbf a}}

\def\bk{{\mathbf k}}

\def\bx{{\mathbf x}}

\def\bA{{\mathbf A}}

\def\bG{{\mathbf G}}
\def\bK{{\mathbf K}}

\def\bX{{\mathbf X}}

\def\j{{j_0}} 

\def\Op{\mathrm{Op}_\eps}

\newcommand{\Supp}{{{\rm Supp}}}

   \DeclarePairedDelimiterX\Set[1]\{\}{%
      
      #1
}

\usepackage[style=numeric-comp,maxbibnames=9,natbib=true]{biblatex}
\AtEveryBibitem{%
  \clearfield{issn} 
  \clearfield{doi} 

  \ifentrytype{online}{}{
    \clearfield{url}
  }
}

\usepackage{hyperref}
\title{Approximate eigenfunctions for some aperiodic crystals}
\author{
  Long Meng\footnote{\textsc{%
      Long Meng,
      Center for Interdisciplinary Applied Mathematics \& Institute of Fundamental and Transdisciplinary Research,
      Zhejiang University,
      China
    }%
    (\texttt{\href{mailto:longmeng@zju.edu.cn}{longmeng@zju.edu.cn}})}
}
\date{}

\begin{document}

\maketitle

\begin{abstract}
    In this paper, we consider a broad class of continuum two-scale Hamiltonians
\begin{align*}
    H_\eps:=T(-i\nabla_x+\bA(x,\eps x))+V(x,\eps x),\qquad x\in \R^d
\end{align*}
where $T$ represents either a Dirac operator or a Schr\"odinger operator, and $x\mapsto \bA(x,X)$ and $x\mapsto V(x,X)$ are $\mathbb L$-periodic with respect to some lattice $\mathbb L\subset\R^d$. No periodicity assumption is imposed on the second variable $X$.

Let 
\begin{align*}
   \R^d\times \R^d\ni   (k,X) \mapsto  h(k,X):=T(-i\nabla_x+k+\bA(x,X))+V(x,X)
\end{align*}
be a family of operators acting on $L^2(\mathbb{R}^d/\mathbb{L})$ with periodic boundary conditions. We assume only local spectral information near one point $(k_0,X_0)$: an
isolated $J$-dimensional Bloch bands at energy $e_0$, possibly with internal crossings, and a Hermitian matrix-valued homogeneous function of degree $m\in\{1,2\}$. From these data we derive an effective Hamiltonian $\mathfrak{h}$. Then every sufficiently localized eigenpair $(\vec v,\mu)$ of $\mathfrak{h}$ gives an approximate eigenfunction $\Theta_\eps$ such that for $\eps$ small enough,
\[
\|\Theta_\eps\|_{L^2(\R^d)}=|\Omega|^{-1/2}+O(\eps^{1/2}),\qquad
 \|(H_\eps-e_0-\eps^{m/2}\mu)\Theta_\eps\|_{L^2(\R^d)}
 =O(\eps^{m/2+1/4}).
\]
In particular, for $m=2$, the effective Hamiltonian contains an explicit matrix correction
generated by coupling to the complementary spectrum of $h(k,X)$ away from $e_0$. We verify the abstract
hypotheses in several distinct settings and obtain, at the two-scale aperiodic
level, harmonic-oscillator, Landau--Schr\"odinger, and Landau--Dirac effective
Hamiltonians.  We also prove the existence of almost-flat Bloch bands for commensurate supercell approximation for some quasiperiodic models.

\end{abstract}
\section{Introduction}

Aperiodic solids, including quasicrystals, strained crystals, and incommensurate multilayers, may retain long-range order while lacking the translation symmetry of an ordinary crystal.  The loss of a common lattice is not merely a geometric inconvenience.  In a periodic crystal, translation symmetry decomposes the
Hamiltonian into independent Bloch fibers. This Bloch decomposition converts an infinite-space spectral problem into a family of eigenvalue problems posed on a single periodic unit cell, making the electronic band structure accessible to both rigorous analysis and efficient numerical computation. The absence of a common translation lattice prevents the global application of Bloch theory and therefore eliminates its usual simplification of the spectral problem. Consequently, the electronic states cannot be reconstructed by solving a fixed-size eigenvalue problem; both their
analytical description and their numerical computation become substantially more difficult.

In this paper, we consider the following family of continuum
two-scale self-adjoint operators on
$L^2(\mathbb R^d;\mathbb C^n)$ for some $n\in \N^+$:
\begin{equation}\label{eq:intro-H}
    H_\eps
    :=
    T\bigl(-i\nabla_x+\bA(x,\eps x)\bigr)
    +V(x,\eps x),
    \qquad
    x\in\mathbb R^d,
    \qquad
    0<\eps\ll1.
\end{equation}
Here,
\begin{itemize}

\item
The external potentials satisfy
\[
    V(x,X)\in C^\infty(\mathbb R^d\times\mathbb R^d;\mathbb R),
    \qquad
    \bA(x,X):=(A_1(x,X),\ldots,A_d(x,X))^{\rm T}
    \in C^\infty(\mathbb R^d\times\mathbb R^d;\mathbb R^d).
\]
They are $\mathbb L$-periodic in the first variable; namely,
\[
    \bA(x+R,X)=\bA(x,X),
    \qquad
    V(x+R,X)=V(x,X)
\]
for all $R\in\mathbb L$ and
$(x,X)\in\mathbb R^d\times\mathbb R^d$, where
$\mathbb L\subset\mathbb R^d$ is the lattice defined in
\eqref{def:lattice} below. Here $x$ is the atomic-scale variable and
$X$ is the mesoscopic-scale variable. No periodicity assumption is
imposed on $X$.

\item
The kinetic-energy operator is either of magnetic Schr\"odinger type,
\[
    T\bigl(-i\nabla_x+\bA(x,\eps x)\bigr)
    :=(-i\nabla_x+\bA(x,\eps x))^2=
    \sum_{j=1}^d
    \bigl(-i\partial_{x_j}+A_j(x,\eps x)\bigr)^2,
\]
or of magnetic Dirac type. In the Dirac case, we restrict to
$d=2$ and $n=2$, and set
\[
    T\bigl(-i\nabla_x+\bA(x,\eps x)\bigr)
    :=
    \pmb{\sigma}\cdot
    \bigl(-i\nabla_x+\bA(x,\eps x)\bigr)
    +M\sigma_3,
    \qquad M\in\mathbb R,
\]
where $\pmb{\sigma}:=(\sigma_1,\sigma_2)$ and
\[
    \sigma_1:=
    \begin{pmatrix}
        0&1\\
        1&0
    \end{pmatrix},
    \qquad
    \sigma_2:=
    \begin{pmatrix}
        0&-i\\
        i&0
    \end{pmatrix},
    \qquad
    \sigma_3:=
    \begin{pmatrix}
        1&0\\
        0&-1
    \end{pmatrix}.
\]
The cases $M=0$ and $M\neq0$ correspond to the massless and massive
Dirac operators, respectively.
\end{itemize}

This formulation covers a periodic medium subject to a slowly varying electric or magnetic field, as well as doped or defective crystals whose impurities or defects vary on a mesoscopic scale. It also
includes continuum models in which the slow modulation is induced by strain, relative displacement, or a second lattice.  In the last case the slow coefficients can themselves be periodic, but with a lattice incommensurate with $\mathbb L$, see e.g.,
\cite{bistritzer2011moire,cances2023simple,cances2023semiclassical,
sinner2023strain,cances2025numerical,garrigue2026effectiveoperatorsgraphenemonolayer}.

Although $H_\eps$ is not globally periodic,  near the mesoscopic position $x=\eps^{-1} X$ this operator can be approximately described by the periodic Hamiltonian
\begin{align*}
    T(-i\nabla_x+\bA(x,X))+V(x,X)
\end{align*}
which can be decomposed by the Bloch transform (see \eqref{eq:Bloch_1}) into the family of Bloch fiber Hamiltonians
\begin{equation}\label{eq:intro-h}
 h(k,X):=T(-i\nabla_x+k+\bA(x,X))+V(x,X)
\end{equation}
acts on $L^2_{\rm per}:=L^2_{\rm per}(\Omega;\C^n)$, defined by \eqref{eq:L2-per}, with periodic boundary conditions. Indeed,  the Bloch transform represents $H_\eps$ as the Weyl quantization of the operator-valued symbol $h(k,X)$ (see \eqref{op-h}). Unlike the usual Bloch decomposition of a periodic Hamiltonian, however, this representation couples the variables $k$ and $X$ and therefore does not simplify $H_\eps$.

In this paper, our aim is to address the following question: when $H_\eps$ has an approximate eigenfunction $\Theta_\eps\in L^2(\R^d;\C^n)$ in the following sense
\begin{align}\label{eq:1.2}
    \frac{\|(H_\eps-e_0-\mu_\eps)\Theta_\eps\|
    _{L^2(\R^d;\C^n)}}{\|\Theta_\eps\|_{L^2(\R^d;\C^n)}}\approx 0,\qquad \eps\ll 1.
\end{align}
More precisely, we try to give a new framework to study this type of problem, and we try to establish the connection between $H_\eps$ and some effective Hamiltonian frequently used in physics. In line with the methodological nature of this work,  we only focus on some simple cases (e.g., homogeneous symbols of degree $m=1,2$ in Assumption \ref{ass:bandstructure-nondegenerate}) relevant to some well known effective Hamiltonians used in physics such as Landau-Dirac/Schr\"odinger effective Hamiltonian quantum harmonic oscillator effective Hamiltonian. Then we also prove the almost flat-band theory for some supercell approximations of quasi-periodic crystals.

Finally, we shall point out that as noted in Remark \ref{rem:extension}  the framework is flexible and can be extended to more general effective symbols and models, and we will use this framework to study more complicated problems later.

\subsection{Local spectral hypothesis and the main results}

We now describe the main results. Assume that $e_0$ belongs to an
isolated $J$-dimensional eigenvalue of $h(k_0,X_0)$. The local
spectral hypothesis, Assumption
\ref{ass:bandstructure-nondegenerate}, requires an orthonormal basis $\vec{w}$ of $ {\rm Ker}(h(k_0,X_0)-e_0)$ and a $J\times J$ Hermitian homogeneous matrix-valued function $f_m^{\rm eff}(k,X)$ of degree $m\in \N^+$ with $J=\dim({\rm Ker}(h(k_0,X_0)-e_0))$.

The central statement is a general novel operator-valued reduction, Theorem
\ref{th:quadratic}.  If $\vec v_\eps$ is a localized function satisfying Assumption \ref{ass:localization} and $\vec u$
is a $\mathbb L$-periodic function, then
\begin{equation}\label{eq:intro-reduction}
 \big\|(H_\eps-e_0)\Phi_\eps(\vec u\otimes\vec v_\eps)
 -\Phi_\eps\big(\mathfrak h_\eps^{\rm eff}
   (\vec u\otimes\vec v_\eps)\big)\big\|_{L^2(\R^d;\C^n)}
 =O(\eps^{m/2+1/4})
\end{equation}
where $\Phi_\eps(\bullet)$,  defined by \eqref{eq:Phi-eps}, is a linear mapping; the operator $\mathfrak h_\eps^{\rm eff}$, defined by \eqref{h-eff-total}, is entirely constructed from the local information of $h(k,X)$ around $(k_0,X_0)$. 

Theorems \ref{th:m} turn this reduction into localized approximate
eigenfunctions.  If $(\vec v_*,\mu_*)$ is a sufficiently localized eigenpair of
the quantized effective symbol $\mathfrak{h}$ (i.e. Assumption \ref{ass:m}), then the constructed state, centered near
$x=\eps^{-1}X_0$ and localized on the scale $\eps^{-1/2}$, satisfies
\begin{equation}\label{eq:intro-quasimode}
 \|\Theta_\eps\|_{L^2}=|\Omega|^{-1/2}+O(\eps^{1/2}),\qquad
 \big\|(H_\eps-e_0-\eps^{m/2}\mu_*)\Theta_\eps\big\|_{L^2}
 =O(\eps^{m/2+1/4}).
\end{equation}
where
\begin{align*}
    \Theta_\eps(x):=\sum_{j=0}^m \eps^{\frac{j}{2}} \Phi_\eps(U_\eps^{(j)}(\vec{w}\otimes \vec{v}_*))
\end{align*}
and $U_\eps^{(0)}$, $U_\eps^{(1)}$ and $U_\eps^{(2)}$ are defined by \eqref{eq:U0-eps}, \eqref{eq:U1-eps} and \eqref{eq:U2-epsilon} respectively. Depending on the choice of $m\in \{1,2\}$, the effective Hamiltonian is
\begin{align*}
    \mathfrak{h}=\begin{cases}
        \mathcal{F}^{-1}{\rm Op}_1(f^{\rm eff}_1)\mathcal{F},& m=1\\
        \mathcal{F}^{-1}{\rm Op}_1(f^{\rm eff}_2)\mathcal{F}+\widetilde{\mathcal{M}},& m=2
    \end{cases}
\end{align*}
where the meaning of notation $ \mathcal{F}^{-1}{\rm Op}_1(f)\mathcal{F}$ is explained in \eqref{eq:F-Opf-F}; and the $J\times J$ matrix $\widetilde{\mathcal M}$, defined in \eqref{eq:M-tilde}, represents the coupling of the Bloch bands near $e_0$ and the ones away from $e_0$. For $J=1$, $\widetilde{\mathcal M}$ is a scalar energy shift; for a massive Landau--Dirac model it agrees with the corresponding Zeeman effect, see Section \ref{sec:shift-of-energy} below. To the best of our knowledge, this correction has not previously been studied in this explicit form for the present two-scale problem.

\subsection{Applications and almost-flat bands}

Even though we only focus on some simple cases in this paper, our result can already be used to explain several different phenomena in physics:
\begin{itemize}
    \item In Section \ref{sec:quantum-harmonic-oscillator}, we further study the relationship between some aperiodic crystals and quantum harmonic oscillators effective Hamiltonian $\mathfrak{h}$, see Theorem \ref{th:harmonic-oscillator}. In this case, Assumption \ref{ass:bandstructure-nondegenerate} is easily verified with $J=1$ and $m=2$. 
\item Our next example is relevant to the Landau-Schr\"odinger effective Hamiltonian used in quantum hall effect. Within a very weak magnetic field, Theorem \ref{th:quantum-hall} states that one can obtain quantized Landau level generated by Landau-Schr\"odinger operator. This is relevant to recent experimental discovery of emergent quantum hall effects below $50 {\rm mT}$ in a two-dimensional topological insulator \cite{week-magnetic-field}. 
\item The last example is the emergence of the Landau-Dirac effective operator in unconventional quantum hall effect for honeycomb materials. In this case, Assumption \ref{ass:bandstructure-nondegenerate} is satisfied with $m=1$ and $J=2$, due to the conic band-structure in graphene \cite{Fefferman-Weinstein-Honeycomb-2012}. Theorem \ref{th:quantum-spin-hall} then justifies the Landau-Dirac operator used in quantum hall effect within the weak magnetic field case.
\end{itemize}
We also emphasize that Theorem \ref{th:m} is not restricted to the quantum harmonic oscillator or Landau-Schr\"odinger/Dirac operator. Whenever the local Bloch cluster admits an effective Hamiltonians $\mathfrak{h}$ with a sufficiently localized eigenfunction, the same procedure produces a physical-space quasimode of the original two-scale Hamiltonian. This yields further effective Hamiltonians, which we do not develop here in order to keep the applications concise.

In addition to above applications, we also use Theorem \ref{th:m} to study the almost flat-band properties in incommensurate crystals: We assume the potential $\bA$ and $V$ are also $\mathbb L$-periodic w.r.t. $X$ and $\eps=p/q\in\mathbb Q$. In physics, this periodic Hamiltonian is used to approximate quasi-periodic crystals, see e.g., \cite{carr2020duality,naik2018ultraflatbands,kariyado2019flat}. In this case, $H_\eps$ is $q\mathbb L$-periodic, and $H_\eps$ can be decomposed into a family of operators $h_q(k)$ defined by \eqref{eq:h-q} by Bloch transform \eqref{eq:Bloch-transform-q} on $q\mathbb L$. Then in Section \ref{sec:almost-flat-band}, Theorem \ref{th:almost-flat} gives
    \begin{align}\label{eq:1.4}
        \sup_{k\in\Omega_q^*}
 {\rm dist}\big(e_0+\eps^{m/2}\mu_*,\sigma(h_q(k))\big)
 =O(\eps^{m/2+1/4})
    \end{align}
where $\Omega^*_q$, defined by \eqref{eq:unit-cell-q}, is the first Brillouin zone for the lattice $q\mathbb L$.

Recently, (almost) flat-band property has been studied for Hamiltonians of the form
\begin{align*}
  T(-i\nabla+\bA(\eps x))+V(\eps x)
\end{align*}
or equivalent Hamiltonians after a scaling $\eps x\mapsto x$, see \cite{becker2024semiclassical,Mathmagicangles}. Compared with existing results, our result is somewhat less sharp.  This is due to Assumption \ref{ass:bandstructure-nondegenerate} and the complexity of our two-scale operator $H_\eps$. In Assumption \ref{ass:bandstructure-nondegenerate}, we restrict ourselves to the leading order asymptotic behavior of the eigenvalues and eigenfunctions of $h(k,X)$ near $e_0$ and $(k_0,X_0)$. To obtain a sharper estimate, one would need, if possible, a higher-order asymptotic expansion of the eigenvalues and eigenfunctions. In this paper, we focus on establishing a general framework for the Hamiltonian $H_\eps$, the study of the flat-band property and other spectrum properties of specific Hamiltonian will be addressed later in future work. 

The oscillation of the density of state in \cite{cances2025numerical} is relevant to \eqref{eq:1.4} and the case $m=2$ with $\mathfrak{h}$ being a quantum harmonic oscillator operator given in Section \ref{sec:quantum-harmonic-oscillator}. Taking the energy shift into account, one can get a better approximation of the density of state:
\begin{figure}[htbp]
  \centering
   \subfloat[]{\includegraphics[scale=0.08]{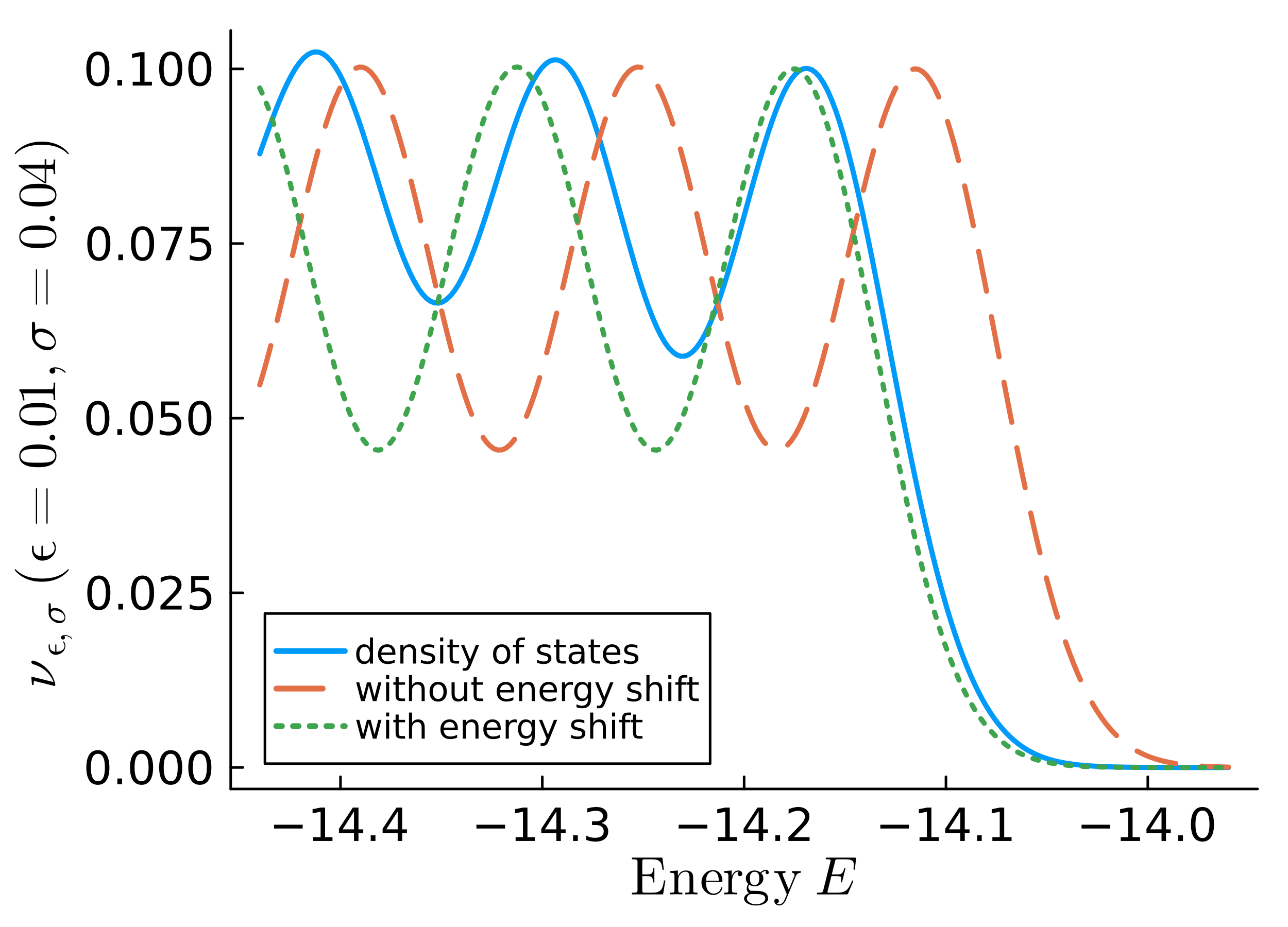}}
    \subfloat[]{\includegraphics[scale=0.08]{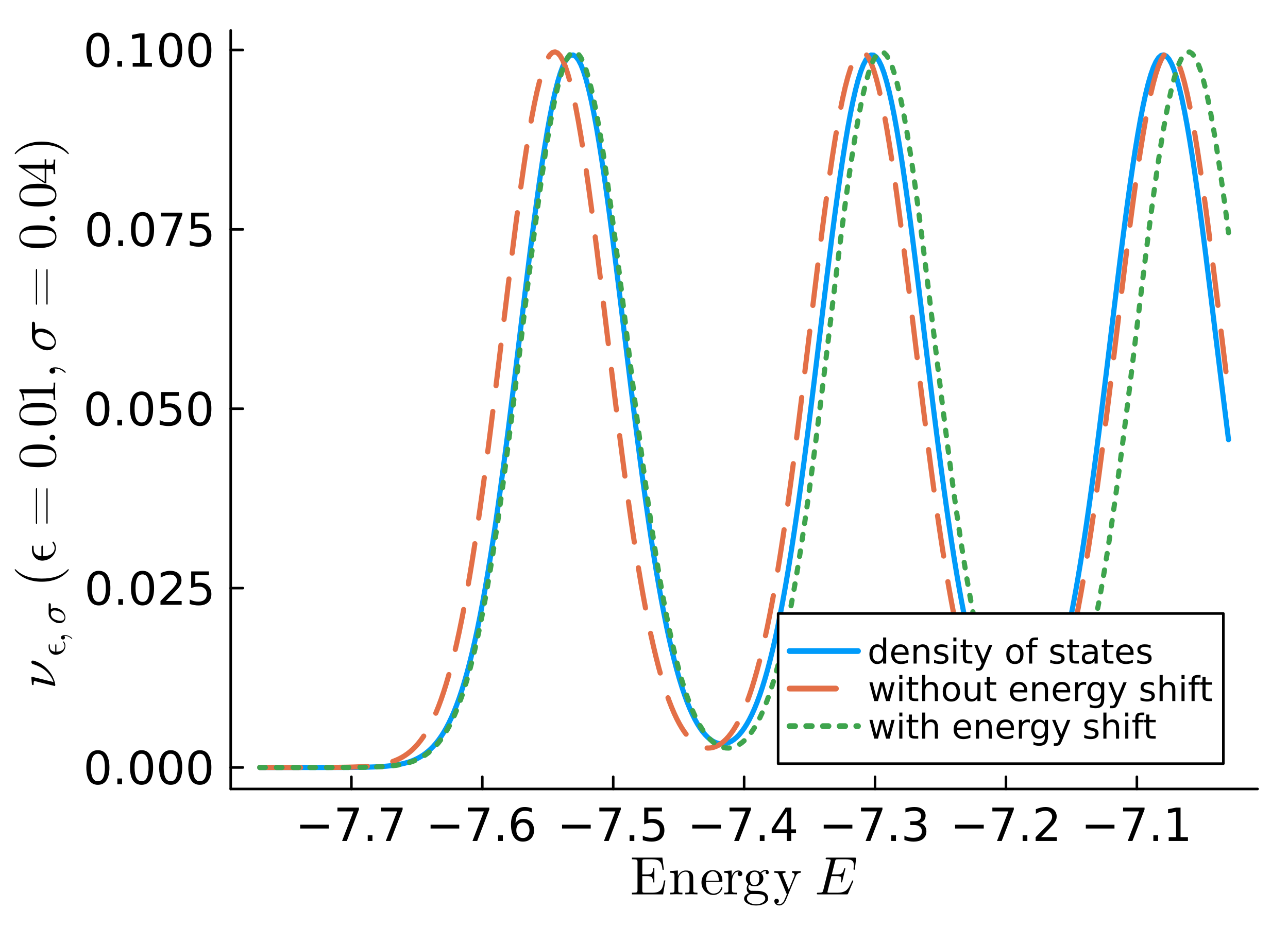}}
  \caption{{Approximations of oscillation of density of state.} The solid blue line represents the exact density of states of a Hamiltonian $H_\eps$ (defined in \cite[Section 3]{cances2025numerical}.) with $\eps=0.01$ by using the method \cite{wang2025convergence}, while the dashed red line corresponds to the quantum harmonic oscillator approximation without energy shift used in \cite{cances2025numerical}. Taking the energy shift $\widetilde{\mathcal{M}}$ into account yields the dashed green line \cite{xue}, which offers a more accurate result for oscillations close to the edge $e_0\approx -14.046$ (left) or $e_0\approx -7.662$ (right).}\label{fig:dos}
\end{figure}

\subsection{Related literature}
The Hamiltonian \eqref{eq:intro-H} lies at the intersection of semiclassical analysis, homogenization, and the spectral theory of aperiodic operators.

In mathematics, the operator $H_\eps$ has been extensively studied via semiclassical methods.
Effective Hamiltonians for slowly perturbed periodic operators, adiabatic theory, and semiclassical calculus for operator-valued symbols have produced Peierls--Onsager reductions, effective band dynamics, and effective-mass equations \cite{gerard1991mathematical,dimassi1993developpements,
panati2003effective,teufel2003adiabatic,stiepan2013semiclassical,
cornean2026fresh,gerard1997homogenization,chabu2022effective,cances2023semiclassical}. Although the proof of this paper uses the common language of operator-valued semiclassical calculus, it does not proceed through a Grushin parametrix or other arguments as above references. Instead, it develops a direct local reconstruction argument based on finite local Bloch spectrum. Our theorem extracts a different and previously unstated consequence from the two-scale operator: a local construction that turns spectral and projector data of $h(k,X)$ near one point $(k_0,X_0)$ into an $L^2(\R^d)$-localized approximate eigenfunctions of the full aperiodic operator $H_\eps$, even when no global band reduction is available.

The connection with homogenization becomes apparent after a dilation $x\mapsto \eps^{-1}x$. This form contains rapidly oscillating coefficients and can therefore be
compared with homogenization problems, especially the periodic one
\cite{Gregoire-homogenization,Gregoire-Andrey,Gregoire-Carlos,
Lin-Shen-criticalsets,Kenig-Lin-Shen-HomogenizationNeumann,
Kenig-Lin-Shen-periodichomogenization,zhang-homogenization,
cances-garrigue-david-homogenization}. Classical homogenization seeks a global macroscopic equation by averaging cell problems. Here, instead, the state is concentrated near $x=\eps^{-1}X_0$ on the intermediate scale $\eps^{-1/2}$, and the effective Hamiltonian is obtained by quantizing the local spectral geometry of $(h(k,X))$. The $\sqrt{\eps}$-scaling is a consequence of the homogeneity hypothesis in Assumption \ref{ass:bandstructure-nondegenerate} (see Lemma~\ref{lem:scaling}), and would change for other cases. 

When $\eps=0$, the problem reduces to ordinary Bloch theory; if $\bA(x,X)=\bA(X)$ and $V(x,X)=V(X)$, it reduces to a standard semiclassical problem
\cite{Fefferman-Weinstein-Honeycomb-2012,
Fefferman-Weinstein-Wavepackets-2014,fefferman2018honeycomb,
shapiro2022tight,Barry-1983-semiclassical,simon1984semiclassical,helfferanalyseI,helfferanalyseII,helfferanalyseIII,helffer1984puits,dencker2003pseudospectra}. The two-scale Hamiltonian $H_\eps$ combines these regimes, especially \cite{Barry-1983-semiclassical,Fefferman-Weinstein-Wavepackets-2014}, but cannot be treated by either theory separately. Bloch transform represents $H_\eps$ as the Weyl quantization of the infinite-dimensional operator-valued symbol $(h(k,X))$ in which the Bloch momentum $k$ and the slow position $X$ are coupled. Consequently, one must simultaneously control the fast oscillations of the Bloch eigenfunctions, the localization of the slowly varying part, the variation of the spectral projector in both $k$ and $X$, and transitions between the selected Bloch bands and their infinite-dimensional complement.

The almost Mathieu operator and related quasiperiodic lattice models admit deep global results on Cantor spectrum, spectral type, and localization
\cite{simon1982almost,jitomirskaya1999metal,avila2009ten,
bourgain2002anderson,klein2007mott,goldstein2001holder,
jitomirskaya2018universal}.The operator $H_\eps$ presents different analytical difficulties: it is an unbounded multidimensional continuum operator, whose fast periodic structure generates infinitely many Bloch bands that are coupled through the slow variable $X$.  A comparably explicit global classification of its spectrum is therefore unlikely to be available at this level of generality. Instead, the present work develops a paradigm that uses only local Bloch spectral data to construct an effective Hamiltonian governing the relevant localized states.

Finally, the direct computation of $H_\eps$ is notoriously difficult in physics for the same structural complexity. Projection, plane-wave, and momentum-space methods nevertheless provide the density of states and approximate states for important classes of incommensurate systems
\cite{jiang2025projection,wang2025convergence,cances2025numerical}. Near a prescribed energy, our reduction requires only local Bloch data from a microscopic unit cell and the solution of a $J$-component effective Hamiltonian. It therefore remains applicable without requiring global symmetry or periodicity of the slow coefficients, a regime that lies outside the scope of existing numerical methods. Figure~\ref{fig:dos} illustrates its consistency with existing numerical computations; a systematic comparison of the approximate eigenfunctions with the projection method of \cite{jiang2025projection} is in preparation.

\medskip

{\bf Organization of the paper.} This paper is organized as follows. In section \ref{sec:2}, we first introduce the Bloch transformation and some notation used throughout the paper. In Section \ref{sec:3}, we introduce our main assumption \ref{ass:bandstructure-nondegenerate}, and states our main results. Then in Section \ref{sec:effective-hamiltonian} we introduce and study the main effective Hamiltonian $\mathfrak{h}_{\eps}^{\rm eff}$; then we use it and Theorem \ref{th:quadratic} to prove our main result, i.e., Theorem \ref{th:m}; the proof of Theorem \ref{th:quadratic} is given in Section \ref{sec:proofofquadratic}. In Sections \ref{sec:quantum-harmonic-oscillator}-\ref{sec:quantum-spin-hall}, we apply Theorem \ref{th:m} to aperiodic crystals relevant to quantum harmonic oscillator effective Hamiltonian, Landau-Schr\"odinger effective Hamiltonian, and Landau-Dirac effective Hamiltonian. Finally we study the almost flat-band properties of some periodic crystals in Section \ref{sec:almost-flat-band}.

\section{Bloch transform and Notation}\label{sec:2}

In this section, we introduce the Bloch transform and some notation that will be used throughout the paper.
\subsection{Bloch transform}
We first recall the lattice and dual-lattice notation, then compare the direct integral decomposition of the periodic operator $H_0$ with the operator-valued Weyl representation of $H_\eps$.

\medskip

\noindent{\bf Lattice and dual lattice.} A lattice $\mathbb L\subset\R^d$ is the set of integer linear combinations of $d$ linearly independent vectors $\ba_1,\ldots,\ba_d\in\R^d$:
\begin{align}\label{def:lattice}
    \mathbb L=\sum_{j=1}^d \ba_j \Z=\{R\in \R^d; \; R=\sum_{j=1}^d \ba_j \gamma_j,\; \gamma_j\in \Z\}.
\end{align}
The dual lattice $\mathbb L^*\subset \R^d$ of $\mathbb L$ is defined by
\begin{align}
    \mathbb L^*:=\{G\in \R^d;\;\forall R\in \mathbb{L},\, e^{iG\cdot R}=1\}=\{G\in \R^d;\;\forall R\in \mathbb{L},\, G\cdot R\in 2\pi \Z\}.
\end{align}
The corresponding Wigner-Seitz cell $\Omega$ and the first Brillouin zone $\Omega^*$ can be identified with the tori
\begin{align}\label{eq:Omega}
\Omega := \mathbb{R}^d/\mathbb{L},\qquad \mbox{with}\quad \overline{\Omega}=\{x\in \R^d; |x|\leq {\rm dist}(x,\mathbb L\setminus \{0\})\}, 
\end{align}
and
\begin{align}\label{eq:Omega*}
   \Omega^* := \mathbb{R}^d/\mathbb{L}^*,\qquad \mbox{with}\quad \overline{\Omega^*}=\{k\in \R^d; |k|\leq {\rm dist}(k,\mathbb L^*\setminus \{0\})\}
\end{align}
respectively. Here and below, ``${\rm dist}$'' denotes the distance between a point $x\in \R^d$ and a set $A\subset \R^d$:
\begin{align*}
    {\rm dist}(x,A):=\inf_{y\in A}|x-y|.
\end{align*}

\medskip

\noindent{\bf Bloch transform.} For the periodic operator $H_0$, the Bloch
transform (also called the Zak transform) is the unitary map
\begin{equation}\label{eq:Bloch_1}
\cU : L^2(\R^d;\mathbb{C}^n) \to \cH:=L^2_{\rm qp}(\Omega^*;L^2_{\rm per})
\end{equation}
such that
\begin{align}
&\forall u \in C^\infty_{\rm c}(\R^d;\C^n), \quad (\cU u)_k(x) = \sum_{R \in \mathbb{L}} u(x+R) e^{-ik \cdot (x+R)},  \\
& L^2_{\rm per}(\Omega;\C^n):=\{f \in L^2_{\rm loc}(\mathbb{R}^d;\C^n) \; |  \; \forall R\in\mathbb{L}, 
\; f(x-R)= f(x) \mbox{ for a.a. } x \in \R^d\}, \label{eq:L2-per}\\
&L^2_{\rm qp}(\Omega^*;L^2_{\rm per}):= \{ u_\bullet \in L^2_{\rm loc}(\R^d;L^2_{\rm per}) \; | \; \forall G\in\mathbb{L}^*,  \; u_{k-G} = \tau_G u_k \mbox{ for a.a. } k \in \R^d \},
\end{align}
where $\tau_G$ is the unitary operator on $L^2_{\rm per}$ defined by $(\tau_Gf)(x)= e^{iG\cdot x}f(x)$.

The space $L^2_{\rm per}(\Omega;\C^n)$ is
endowed with the inner product
\begin{align*}
    \langle u,v \rangle_{L^2_{\rm per}} :=\int_{\Omega} u^*(x)v(x) \, dx 
\end{align*}
and the space $\cH$ with the inner product
$$
\langle u_\bullet ,v_\bullet \rangle_\cH :=\fint_{\Omega^*} \langle u_k, v_k \rangle_{L^2_{\rm per}}  \, dk.
$$
We also set, for all $s\in \N$,
$$
H^s_{\rm per}:=\{ u \in L^2_{\rm per} \; | \; \partial_x^\alpha  u \in L^2_{\rm per}, \, \forall \alpha\in \N^d \mbox{ s.t. } |\alpha|\leq s \}
$$
endowed with its natural inner product. 

Since $H_0$ is $\mathbb L$-periodic, it can be decomposed by the Bloch transform $\cU$: 
\begin{align}\label{op-h0}
    H_0= \cU^{-1} \left( \fint_{\Omega^*}^\oplus  h(k,0) \, dk \right) \cU,
\end{align}
where $h(k,0)$ is the operator on $L^2_{\rm per}$ given by
$$
h(k,0):= T(-i\nabla_x+k+\bA(x,0))+V(x,0).
$$

\medskip

\noindent{\bf Weyl quantization on Bloch transform.} When
$\eps\neq0$, the operator $H_\eps$ need not be $\mathbb L$-periodic and hence does not admit a direct integral decomposition. It can nevertheless be represented after the Bloch transform as an operator-valued pseudodifferential operator; see \cite{panati2003effective,cances2023semiclassical}:
\begin{align}\label{op-h}
    H_{\eps} = \cU^{-1} h(k,i\eps\nabla_k) \cU
\end{align}
where
\begin{itemize}
\item $h$ is the operator-valued symbol on $\R^d\times \R^d$ such that for all $(k,X) \in \R^d \times \R^d$,  $h(k,X)$ is the self-adjoint operator on $L^2_{\rm per}$ with domain $H^{r}_{\rm per}$ for some $r\in\{1,2\}$: for all $u \in H^r_{\rm per}$ 
\begin{equation}\label{eq:def_hdeps}
[h(k,X) u](x) := T \left( -i \nabla_x+k +\bA(x,X)\right) u(x)+ V(x,X) u(x);
\end{equation}
\item The pseudo-differential operator $h(k,i\eps\nabla_k )$ is defined by Weyl quantization:
\begin{align}
    h(k,i\eps\nabla_k ):={\rm Op}_\epsilon(h)(k,i\eps\nabla_k );
\end{align}
\item  For a suitable operator-valued symbol $a(k,X)$, ${\rm Op}_\epsilon(a)$ is defined by the following Weyl quantization rule:
\begin{equation}\label{eq:Weyl_quantization}
[{\rm Op}_\epsilon(a) \phi]_k (r) = \frac{1}{(2\pi\epsilon)^d} \int_{\R^d \times \R^d} \left[a\left( \frac{k+k'}2,X \right) \phi_{k'} \right](r) \; e^{-i \frac{(k-k')\cdot X}\epsilon} \, dk' \, dX.
\end{equation}
\end{itemize}
Unlike the periodic decomposition \eqref{op-h0}, the Weyl operator \eqref{op-h} is not diagonal in $k$: the Bloch electrons interact with each other through a pseudo-differential operator in the aperiodic crystals $H_\eps$, and the localized approximate eigenfunctions of $H_\eps$ are thus a consequence of the collective behavior of Bloch electrons.

\subsection{Notation}
For brevity, throughout the paper we write $L^2_{\rm per}:=L^2_{\rm per}(\Omega;\C^n)$ and $L^2(\R^d):=L^2(\R^d;\C)$.

\medskip

\noindent{\bf Vector-valued functions.} For $J\in\N^+$, and for
$u_j\in L^2_{\rm per}$ and $v_j\in L^2(\R^d)$, $j=1,\cdots, J$, we write
\begin{align*}
    \vec{u}:=(u_1,\ldots,u_J)^T\in L^2_{\rm per}(\R^d;\C^J\otimes \C^n),
\end{align*}
and
\begin{align*}
    \vec{v}:=(v_1,\ldots,v_J)^T \in L^2(\R^d;\C^J).
\end{align*}
Their product is
\begin{align}\label{eq:vector-product}
    \vec{u}^T(x) \vec{v}(y)=\sum_{j=1}^J u_j(x)v_j(y)\in  \C^n.
\end{align}
We use the following componentwise conventions:
\begin{itemize}
\item we set $\vec{u}\in L^2_{\rm per}$ if $u_j\in L^2_{\rm per}$ for any $1\leq j\leq J$, and set $\vec{v}\in L^2(\R^d)$ if $v_j\in L^2(\R^d;\C)$ for any $1\leq j\leq J$;
\item for any $\vec{w}, \vec{u}\in L^2_{\rm per}$,  we set
\begin{align*}
    \left<\vec{w}, \vec{u}\right>_{L^2_{\rm per}}=\sum_{j=1}^J \left<w_j,u_j\right>_{L^2_{\rm per}};
\end{align*}
    \item for any operator $A: H^{t_1}_{\rm per}\to L^2_{\rm per}$ and $B: H^{t_2}(\R^d;\C)\mapsto L^2(\R^d;\C)$ with some $t_1,t_2\in \R$, we define
\begin{align*}
    A\vec{u}:=(Au_1,\cdots, Au_J)^T
\end{align*}
and
\begin{align*}
    B\vec{v}:=(Bv_1,\cdots, Bv_J)^T.
\end{align*}

\end{itemize}

\medskip

\noindent{\bf Multi-index.} For
$\gamma=(\gamma_1,\ldots,\gamma_d)$ and
$\beta=(\beta_1,\ldots,\beta_d)$ in $\N^d$, and for $k,X\in\R^d$, set
\begin{align*}
    k^\beta=\prod_{j=1}^dk_j^{\beta_j},\qquad  X^\gamma=\prod_{j=1}^dX_j^{\gamma_j}
\end{align*}
and
\begin{align*}
    \partial_k^\beta=\prod_{j=1}^d\partial_{k_j}^{\beta_j},\qquad  \partial_X^\gamma=\prod_{j=1}^d\partial_{X_j}^{\gamma_j}.
\end{align*}
In addition, we set
\begin{align*}
    |\gamma|_1=\sum_{j=1}^d |\gamma_j|
\end{align*}
and
\begin{align*}
    \gamma!=\prod_{j=1}^d\gamma_j!
\end{align*}

\medskip

\noindent{\bf Cut-off function and ball.} We define a smooth cut-off function $\chi\in C^\infty(\R^d;[0,1])$ as follows
\begin{align*}
 \chi(x)= \begin{cases}
      1\qquad \mbox{for } |x|\leq 1;\\
      0\qquad \mbox{for } |x|\geq 2.
  \end{cases} 
\end{align*}
Let $R>0$ and $a\in \R^d$, define 
\begin{align}
    B_R(a):=\{x; |x-a|< R\}
\end{align}
as the ball of radius $R$ centered at $a$.

\medskip

\noindent{\bf Fourier Transform.} Let $\mathcal{F}$ be the Fourier transform: for any $g\in L^2(\R^d)$,
\begin{align}\label{eq:Fourier-transform}
    \mathcal{F}(g)(k):=\int_{\R^d}e^{-i k\cdot x}g(x)dx
\end{align}
with the inverse Fourier transform
\begin{align*}
     \mathcal{F}^{-1}(g)(x):=\frac{1}{(2\pi)^d}\int_{\R^d}e^{i k\cdot x}g(k)dk.
\end{align*}

For the effective Hamiltonians it is convenient to conjugate the Weyl quantization by the Fourier transform.  If $f(k,X)$ is a $J\times J$ matrix-valued polynomial and $\vec v\in\mathcal S(\R^d;\C^J)$, we set
\begin{align}\label{eq:F-Opf-F}
    \mathcal{F}^{-1}\Op(f)\mathcal{F} \vec v (x) = \mathcal{F}^{-1}\Big(\Op(f)(k,i\eps\nabla_k)\mathcal{F}(\vec v)(k)\Big)(x).
\end{align}

\medskip

\noindent{\bf Eigenpairs.}  We introduce some additional notation associated with the family of operators $h(k,X)$.  Let $e_0\in\mathbb R$, as specified in
Assumption \ref{ass:bandstructure-nondegenerate} below, be an
eigenvalue of $h(k_0,X_0)$ of multiplicity $J\in \N^+$ for some point $(k_0,X_0)\in \R^d\times \R^d$. We focus only on the spectrum near $e_0$ and on parameters near $(k_0,X_0)$.

Let $e_1,e_2\in\mathbb R$ be chosen such that
\[
    e_1<e_0<e_2,\qquad
    \sigma\bigl(h(k_0,X_0)\bigr)\cap[e_1,e_2]=\{e_0\}.
\]
In particular,
\[
    e_1,e_2\in\rho\bigl(h(k_0,X_0)\bigr),
\]
where $\rho(h(k,X))$ denotes the resolvent set of $h(k,X)$.

By the continuity of $(k,X)\mapsto h(k,X)$,
there exists a sufficiently small neighborhood
\[
    \mathcal N\subset\mathbb R^d\times\mathbb R^d
\]
of $(k_0,X_0)$ such that
\[
    e_1,e_2\in\rho\bigl(h(k,X)\bigr)
    \qquad
    \text{for every }(k,X)\in\mathcal N,
\]
and the spectral projector
\[
    P(k,X)
    :=
    \1_{(e_1,e_2)}\bigl(h(k,X)\bigr)
\]
has constant rank $J$ on $\mathcal N$.

For each $(k,X)\in\mathcal N$, let
\[
    E_1(k,X)\leq E_2(k,X)\leq\cdots\leq E_J(k,X)
\]
be the eigenvalues of the finite-dimensional self-adjoint operator
\[
    h(k,X)\big|_{\operatorname{Ran}P(k,X)},
\]
counted with multiplicity. Thus,
\[
    e_1<E_1(k,X)\leq\cdots\leq E_J(k,X)<e_2,
    \qquad
    (k,X)\in\mathcal N.
\]
This defines the local eigenvalue functions
\[
    E_j(k,X):\mathcal N\longrightarrow\mathbb R,
    \qquad j=1,\ldots,J.
\]

For each fixed $(k,X)\in\mathcal N$, choose an orthonormal eigenbasis
\[
    \{\phi_1(k,X),\ldots,\phi_J(k,X)\}
\]
of $\operatorname{Ran}P(k,X)$ such that
\[
    h(k,X)\phi_j(k,X)
    =
    E_j(k,X)\phi_j(k,X),
    \qquad
    \|\phi_j(k,X)\|_{L^2_{\rm per}}=1.
\]

\section{Asymptotic behavior of eigenmodes of \texorpdfstring{$h(k,X)$}{} and main results}\label{sec:3}
In this section, we summarize our main results. To do so, we first need to characterize the asymptotic behavior of eigenmodes  of $h(k,X)$ w.r.t. $(k,X)$ and then we use this asymptotic behavior to construct the approximate eigenfunctions of $H_\eps$.

\subsection{ Asymptotic behavior of the eigenvalues.}\label{sec:asymptotic-eigenvalues}
The following hypothesis describes an isolated spectral structure of $h(k,X)$ near the energy $e_0$ and a point $(k_0,X_0)$. Sections
\ref{sec:quantum-harmonic-oscillator}--\ref{sec:quantum-spin-hall} verify it in several distinct settings.
\begin{assumption}[Behavior of eigenvalues of $h$]\label{ass:bandstructure-nondegenerate}
Let $ m\in \N^+$.  Let 
\begin{align*}
   s_1:=\frac{1}{2}+\frac{1}{2\mathfrak{n} d \left(m+1\right)},\qquad    s_2:=\frac{1}{2}-\frac{1}{2\mathfrak{n} d (m+1)},\quad \mbox{and}\quad \mathfrak{n}:=8+\frac{16}{ m+1}.
\end{align*}
We assume that the operator $H_\eps$ or equivalently the family of operators $(h(k,X))_{k,X}$ satisfies the following properties:
\begin{enumerate}
    \item {\bf Eigenvalues.} The eigenvalue $e_0$ is an eigenvalue of $h(k_0,X_0)$ with  multiplicity $J:=\dim {\rm Ker}(h(k_0,X_0)-e_0)\in \N^+$, i.e.,
    \begin{align*}
     E_1(k_0,X_0)=\cdots=E_{J}(k_0,X_0)=e_0.
    \end{align*}
    \item  {\bf Asymptotic behavior of eigenvalues.} There exist functions $(\lambda_j^{\rm eff}(k,X))_{1\leq j\leq J}$ in $\R$  such that for any $1\leq j\leq J$,
\begin{align}\label{eq:ass1-1}
    E_{j}(k,X)=e_0+\lambda_j^{\rm eff}(k,X) +\mathcal{O}\left(|k-k_0|^{m+1}+|X-X_0|^{m+1}\right);
\end{align}
    \item {\bf Conditions on $\lambda^{\rm eff}_j(k,X)$.} There exist a Hermitian matrix-valued polynomial function $h^{\rm eff}: \R^{2d}\to \mathbb M_{J\times J}(\C)$ and an orthonormal basis $(\alpha_j(k,X))_{1\leq j\leq J}$ on $\C^J$ such that
\begin{align}\label{eq:ass1-3}
     h^{\rm eff}(k,X)\vec{\alpha}_j(k,X)= \lambda^{\rm eff}_j(k,X)\vec{\alpha}_j(k,X)
\end{align}
and there exists a homogeneous matrix-valued function $f_m^{\rm eff}$ of degree $m$ such that
\begin{align*}
    h^{\rm eff}(k,X):=f_m^{\rm eff}(k-k_0,X-X_0);
\end{align*}
\item {\bf Asymptotic behavior of eigenfunctions.} Let 
\begin{align*}
    w_{j}=\phi_{j}(k_0,X_0),\qquad 1\leq j\leq J,
\end{align*}
and
\begin{align*}
    \vec{w}:=(w_1,\cdots,w_J)^T.
\end{align*}
The orthonormal basis $(\vec{\alpha}_j(k,X))_{1\leq j\leq J}$ satisfies that for any $1\leq j\leq J$,
\begin{align}\label{eq:ass1-2}
    \|\phi_{j}(k,X)-\vec{w}^T\, \vec{\alpha}_j(k,X)\|_{L^2_{\rm per}}=\cO(|k-k_0|+|X-X_0|).
\end{align}
\end{enumerate}
\end{assumption}

The above assumption describes the asymptotic behavior of the eigenvalues $E_{j}(k,X)$ and the eigenfunctions $\phi_{j}(k,X)$ for $j=1,\cdots,J$ around $(k_0,X_0)$. It is a verifiable spectral condition, not an assumption on the unknown eigenpairs. For a simple isolated eigenvalue
it follows from standard analytic perturbation theory; conical crossings can
be treated by degenerate perturbation theory \cite{Fefferman-Weinstein-Honeycomb-2012,Reed-Simon-1978-operator}. Sections \ref{sec:quantum-harmonic-oscillator}--\ref{sec:quantum-spin-hall} verify the hypothesis for the effective models used later.

\subsection{Notation on \texorpdfstring{$h(k,X)$}{}}
Before stating the main results, we need to introduce some notation for $h(k,X)$ under Assumption \ref{ass:bandstructure-nondegenerate}.

Let
\begin{align*}
    h_e(k,X):=h(k,X)-e_0,\qquad  h_{e,0}:=h_e(k_0,X_0).
\end{align*}

\medskip

\noindent{\bf Projectors.} For any $g\in L^2_{\rm per}$, we define the operator $\left|g\right>\left<g\right|$ as follows: for any $f\in L^2_{\rm per}$,
\begin{align*}
    \left(\left|g\right>\left<g\right| f\right)(x) =\left<g,f\right>_{L^2_{\rm per}} g(x).
\end{align*}
We need a projector associated with the eigenfunctions of $h(k,X)$ in Assumption \ref{ass:bandstructure-nondegenerate}:
\begin{align}\label{def:P-parallel}
    P^\parallel(k,X):=\sum_{j=1}^J\left|\phi_{j}(k,X)\right>\left<\phi_{j}(k,X)\right|.
\end{align}
Its orthogonal projector is
\begin{align}\label{def:P-bot}
    P^\bot(k,X)=1-P^\parallel(k,X).
\end{align}
We also set
\begin{align}\label{def:P-k0X0}
    P^\parallel_0:=P^\parallel(k_0,X_0)=\1(h(k_0,X_0)=e_0),\qquad P_0^\bot:=P^\bot(k_0,X_0),
\end{align}
and
\begin{align}
    h_e^\bot(k,X)=h_e(k,X)P^\bot(k,X)=P^\bot(k,X)h_e(k,X)P^\bot(k,X).
\end{align}

\medskip

\noindent{\bf Derivatives of $h_e(k,X)$ and $P^\bot(k,X)$.} Note that under Assumption \ref{ass:bandstructure-nondegenerate}, $e_0$ is an eigenvalue of $h(k_0,X_0)$ of multiplicity $J$. According to the continuity of $(k,X)\mapsto h(k,X)$, there is a gap in the spectrum of $h(k,X)$ in the following sense: for any $(k,X)\in \cN$,
\begin{align}\label{eq:gap-distance}
\inf_{\ell=1,\cdots,J}  {\rm dist}\Big( E_{\ell}(k,X), \sigma(h(k,X))\setminus\{E_{j}(k,X),\;j=1,\cdots,J\}\Big)>0.
\end{align}
Thus $(k,X)\mapsto P^\bot(k,X)$ is smooth in this neighborhood. We define these derivatives at $(k_0,X_0)$: for any $\gamma,\beta\in \N^d$,
\begin{align}
    P^\bot_{\beta,\gamma,0}:=(\partial_k^\beta\partial_X^\gamma P^\bot)(k_0,X_0),\\
    h_{e,\beta,\gamma,0}:=(\partial_k^\beta\partial_X^\gamma h_e)(k_0,X_0)
\end{align}
and
\begin{align}\label{h-e-beta-alpha-0}
    h_{e,\beta,\gamma,0}^\bot:=(\partial_k^\beta\partial_X^\gamma h_e^\bot)(k_0,X_0).
\end{align}
As a result, we have the Taylor expansion of the operator $(k,X)\mapsto h_{e}^\bot(k,X)$ at $(k_0,X_0)$, i.e.,
\begin{align}\label{eq:taylor expansion}
    h_{e}^\bot(k,X)=\sum_{\substack{\beta,\gamma\in \N^d\\ |\beta+\gamma|_1\leq m}}\!\!\!\! \frac{1}{\beta!\gamma !}h^\bot_{e,\beta,\gamma,0} 
(k-k_0)^\beta(X-X_0)^\gamma +\cO_{\cB(H^{2}_{\rm per};L^2_{\rm per})}(|k-k_0|^{m+1}+|X-X_0|^{m+1}).
\end{align}

\subsection{Main results}
Based on Assumption \ref{ass:bandstructure-nondegenerate}, we can now construct approximate eigenfunctions of $H_\eps$. The construction of the approximate eigenfunctions and its spectral scale depend on the degree $m$ in Assumption \ref{ass:bandstructure-nondegenerate}.  We treat $m=1$ and $m=2$, which cover the conical and quadratic situations used in the applications.  Higher-degree or non-homogeneous cases would require a corresponding change of scaling and additional correctors and are not pursued here.

In our main results, the approximate eigenfunction $\Phi_\eps$ can be regarded as a composite of periodic wavefunctions in $L^2_{\rm per}$ and localized wavefunctions in $L^2(\R^d)$. Using this observation, we define the wavefunction $\Phi_\eps$ as a linear operator from $L^2_{\rm per}\otimes \big(L^2(\R^d)\cap W^{d+1,1}(\R^d)\big)$ to $L^2(\R^d;\C^n)$ by
\begin{align}\label{eq:Phi-eps}
    \Phi_{\eps}(\vec a\otimes\vec b)(x)=\chi(\eps^{s_1} x-\eps^{s_1-1}X_0)  \vec{a}^T(x)\,\vec{b}(x).
\end{align}
Here $\Phi_{\eps}(\vec a\otimes\vec b)\in L^2(\R^d;\C^n)$ is a consequence of Lemma \ref{lem:Uf-cH} when $\vec{b}\in  W^{d+1,1}(\R^d)$.

 We need the following assumption.
\begin{assumption}\label{ass:m}
    Let $m\in \{1,2\}$ and let $f_m^{\rm eff}$ be the homogeneous function of degree $m$ defined in Assumption \ref{ass:bandstructure-nondegenerate}. We assume that 
    \begin{itemize}
        \item $(\vec{v}_*,\mu_*)\in L^2(\R^d)\times \R$ is a normalized eigenpair of the operator 
    \begin{align*}
        \mathfrak{h}=\begin{cases}
        \mathcal{F}^{-1}{\rm Op}_1(f^{\rm eff}_1)\mathcal{F},& m=1\\
        \mathcal{F}^{-1}{\rm Op}_1(f^{\rm eff}_2)\mathcal{F}+\widetilde{\mathcal{M}},& m=2
    \end{cases}
    \end{align*}
where $\widetilde{\mathcal{M}}$ is a $J\times J$ matrix with elements
\begin{align}\label{eq:M-tilde}
(\widetilde{\mathcal{M}})_{jk}:=   \sum_{\substack{\beta,\gamma\in \N^d,\\\beta=0,\;|\gamma|_1=1}} \left<w_j, {\rm Im}(P_{\beta,\gamma,0}^\bot h_{e,0}^\bot P_{\gamma,\beta,0}^\bot) w_k\right>_{L^2_{\rm per}}
\end{align}
with the imaginary part of $ P_{\beta,\gamma,0}^\bot h_{e,0}^\bot P_{\gamma,\beta,0}^\bot$ being defined by
\begin{align*}
    {\rm Im}(P_{\beta,\gamma,0}^\bot h_{e,0}^\bot P_{\gamma,\beta,0}^\bot):=\frac{P_{\beta,\gamma,0}^\bot h_{e,0}^\bot P_{\gamma,\beta,0}^\bot-P_{\gamma,\beta,0}^\bot h_{e,0}^\bot P_{\beta,\gamma,0}^\bot}{2i};
\end{align*}
\item $\vec{v}_*$  satisfies that for any $\gamma\in \N^d$
\begin{align}\label{eq:v2-m}
    y^\gamma\vec{v}_*,\; \partial^\gamma\vec{v}_*\in \cS(\R^d)
\end{align}
and
 \begin{align}\label{eq:v3-m}
 \sum_{\gamma\in\N^d,\;|\gamma|_1\leq m} \|y^\gamma\vec{v}_*\|_{W^{d+2+2m,1}(\R^d\setminus B_{\eps^{1/2-s_1}}(0))}=\cO(\eps^{\frac12(m+1)}).
\end{align}
    \end{itemize}
\end{assumption}

Before going further, for any function $\vec{v}\in \mathcal{S}(\R^d)$, we define the following unitary scaling-translation operator $T_\eps$ by
\begin{align}\label{eq:scaling-translation}
    T_\eps(\vec{v})(x)= \eps^{\frac{d}{4}}  e^{ik_0\cdot (x-\eps^{-1}X_0)}\vec{v}(\sqrt\eps(x-\eps^{-1}X_0)).
\end{align}
Based on Assumption \ref{ass:m} with eigenpair $(\vec{v}_*,\mu_*)$, we can now construct the approximate eigenfunction of $H_\eps$: 
\begin{itemize}
    \item the leading order term of the approximate eigenfunction is $\Phi_\eps(U_\eps^{(0)}(\vec{w}\otimes \vec{v}_*))$ with
\begin{align}\label{eq:U0-eps}
     U_{\eps}^{(0)}(\vec{u}\otimes \vec{v}):=(\1_{L^2_{\rm per}}\otimes T_\eps) U^{(0)}(\vec{u}\otimes \vec{v})=\vec{u}\otimes T_\eps(\vec{v})
\end{align}
and
\begin{align}\label{eq:U0}
    U^{(0)}(\vec{u}\otimes \vec{v}):= \vec{u}\otimes \vec{v}\in     L^2_{\rm per}\otimes L^2(\R^d);
\end{align}
\item the first order correction term of the eigenfunction is $\sqrt\eps \Phi_\eps(U_\eps^{(1)}(\vec{w}\otimes \vec{v}_*))$ with
\begin{align}\label{eq:U1-eps}
     U_{\eps}^{(1)}(\vec{u}\otimes \vec{v}):= (\1_{L^2_{\rm per}}\otimes  T_\eps) U^{(1)}(\vec{u}\otimes \vec{v})
\end{align}
and
\begin{align}\label{eq:U1}
 U^{(1)}(\vec{u}\otimes \vec{v})(x,y):&= -  (P^\bot_0\otimes \1_{L^2(\R^d)}) \Big(h_{e,0}\otimes\1_{L^2(\R^d)}\Big)^{-1}\mathfrak{h}_{12}^{(1)}  U^{(0)}(\vec{u}\otimes \vec{v})(x,y)\notag\\
     &=-\sum_{\substack{\beta,\gamma\in \N^d,\\ \gamma=0,\;|\beta|_1=1}} \big[P^\bot_0 P_{\beta,\gamma,0}^\bot  \vec{ u}\big](x)\otimes ((-i\partial)^\beta \vec v)(y)\notag\\
    &\quad -\sum_{\substack{\beta,\gamma\in \N^d,\\ \beta=0,\;|\gamma|_1=1}} \big[P^\bot_0P_{\beta,\gamma,0}^\bot  \vec{ u}\big](x)\otimes (y^\gamma \vec v)(y)
\end{align}
where
\begin{align}\label{eq:h12-1}
    \mathfrak{h}_{12}^{(1)}(\vec u\otimes\vec v)(x,y):&=\sum_{\substack{\beta,\gamma\in \N^d,\\ \gamma=0,\;|\beta|_1=1}} \big[h_{e,0} P_{\beta,\gamma,0}^\bot  \vec{ u}\big](x)\otimes((-i\partial)^\beta \vec v(y))\notag\\
    &\quad +\sum_{\substack{\beta,\gamma\in \N^d,\\ \beta=0,\;|\gamma|_1=1}} \big[h_{e,0} P_{\beta,\gamma,0}^\bot  \vec{ u}\big](x)\otimes(y^\gamma \vec v(y)).
\end{align}
Here the operator $(P^\bot_0\otimes \1_{L^2(\R^d)}) \Big(h_{e,0}\otimes\1_{L^2(\R^d)}\Big)^{-1}$ should be understood as a bounded operator on $L^2_{\rm per}\otimes L^2(\R^d)$ defined as
\begin{align*}
    (P^\bot_0\otimes \1_{L^2(\R^d)}) \Big(h_{e,0}\otimes\1_{L^2(\R^d)}\Big)^{-1}=(P^\bot_0 h_{e,0}^{-1})\otimes\1_{L^2(\R^d)}=(P^\bot_0 h_{e,0}^{-1} P^\bot_0)\otimes\1_{L^2(\R^d)};
\end{align*}
\item the second order correction term is $\eps\Phi_\eps(U_\eps^{(2)}(\vec{w}\otimes \vec{v}_*))$ with
\begin{align}\label{eq:U2-epsilon}
     U_{\eps}^{(2)}(\vec{u}\otimes \vec{v}):= (\1_{L^2_{\rm per}}\otimes  T_\eps) U^{(2)}(\vec{u}\otimes \vec{v})
\end{align}
and
\begin{align}\label{eq:U2}
     U^{(2)}(\vec{u}\otimes \vec{v}):&=- (P^\bot_0\otimes \1_{L^2(\R^d)} )\Big(h_{e,0}\otimes\1_{L^2(\R^d)}\Big)^{-1}\Big[\mathfrak{h}^{(2)}U^{(0)} (\vec{u}\otimes \vec{v})+\mathfrak{h}_{22}^{(1)} U^{(1)}(\vec{u}\otimes \vec{v})\Big]
\end{align}
where
\begin{align}\label{eq:h22-1}
    \mathfrak{h}_{22}^{(1)}(\vec u\otimes\vec v)(x,y):&=\sum_{\substack{\beta,\gamma\in \N^d,\\\ \gamma=0,\; |\beta|_1=1}} \big[h_{e,\beta,\gamma,0}\vec u\big](x)\otimes ((-i\partial_{y})^\beta\vec{v})(y)\notag\\
    &\quad+\sum_{\substack{\beta,\gamma\in \N^d,\\\ \beta=0,\; |\gamma|_1=1}} \big[h_{e,\beta,\gamma,0}\vec u\big](x)\otimes(y^\gamma\vec{v})(y)
\end{align}
and
\begin{align}\label{eq:h-2}
    \mathfrak{h}^{(2)}(\vec{u}\otimes\vec{v})(x,y):=\frac{1}{2}\sum_{\substack{\beta,\gamma\in \N^d\\ |\beta+\gamma|_1=2}}\frac{1}{\beta!\gamma!}\left(h^\bot_{e,\beta,\gamma,0} \vec u\right)(x)\otimes \Big((-i\partial_y)^\beta y^\gamma+y^\gamma (-i\partial_y)^\beta\Big) \vec{v}(y).
\end{align}
\end{itemize}

Then,
\begin{theorem}\label{th:m}
We assume that the operator $H_\eps$ satisfies Assumption \ref{ass:bandstructure-nondegenerate} with $m\in\{1,2\}$. Let $(\vec{v}_*,\mu_*)$ be an eigenpair satisfying Assumption \ref{ass:m}. Then,  for $\eps$ small enough,
\begin{align}\label{eq:m}
          \left\| (H_\eps  -e_0- \eps^{\frac{m}{2}}\mu_*)\sum_{j=0}^m \eps^{\frac{j}{2}} \Phi_\eps\Big(U_\eps^{(j)}(\vec{w}\otimes \vec{v}_*)\Big)\right\|_{L^2(\R^d;\C^n)}=\cO( \eps^{\frac{m}{2}+\frac{1}{4}}),
\end{align}
with
\begin{align*}
    \left\|\sum_{j=0}^m \eps^{\frac{j}{2}}\Phi_\eps\Big(U_\eps^{(j)}(\vec{w}\otimes \vec{v}_*)\Big)\right\|_{L^2(\R^d;\C^n)}=\frac{1}{|\Omega|^{1/2}}+\cO(\sqrt{\eps}).
\end{align*}
\end{theorem}

\begin{remark}[Error estimate]
Under Assumption \ref{ass:periodicpotentials} in Section \ref{sec:almost-flat-band} later, $H_\eps$ can be a periodic Hamiltonian for some suitable $\eps$. According to \cite{thomas1973time}, if $\bA=0$, then $H_\eps$ has only absolutely continuous spectrum. Thus, there are no localized eigenfunctions of $H_\eps$. In this sense, an error estimate in \eqref{eq:m} is necessary. 

Better remainder estimates may be possible under stronger
quantitative estimates (i.e., \eqref{eq:ass1-1} and \eqref{eq:ass1-3}) and localization assumptions, together with a
refined choice of cutoffs and, if necessary, additional correctors.
Such improvements are not pursued here.
\end{remark}

\begin{remark}[Appearance of the $\widetilde{\mathcal{M}}$ term]\label{rem:zeeman}
  From the proof of Theorem \ref{th:m} in Section \ref{sec:proof-m=2}, we observe that $\widetilde{\mathcal{M}}$ arises from the term
    \begin{align*}
        \mathfrak{h}^{(2)}-\big(\mathfrak{h}_{12}^{(1)}\big)^* \Big(h_{e,0}\otimes\1_{L^2(\R^d)}\Big)^{-1}\mathfrak{h}_{12}^{(1)}.
    \end{align*}
    It takes a form analogous to the Zeeman effect \eqref{eq:zeeman'} in the Landau-Dirac operator, where $\mathfrak{h}^{(2)}$ plays the role of Landau-Schr\"odinger operator $H^S$, $\mathfrak{h}_{12}^{(1)}$ plays the role of $\mathcal{D}^*$, and $\Big(h_{e,0}\otimes\1_{L^2(\R^d)}\Big)^{-1}$ plays the role of $\frac{1}{2}$, with notation defined in Section \ref{sec:shift-of-energy}.
\end{remark}

\begin{remark}[Possible extension]\label{rem:extension}
  In this paper, our aim is to provide a new method to study aperiodic crystals of the form $H_\eps$, and we only considered some simple cases, i.e., $f_m^{\rm eff}$ is a homogeneous function with $m\in\{1,2\}$. But this does not mean that we cannot consider more general $f^{\rm eff}$. 
  
  For example, for $J=1$, one can easily extend the present results to the double well case \cite{simon1984semiclassical}, for example
  \begin{align*}
      \lambda^{\rm eff}_j(k,X)=|k|^2+X^2(X-1)^2+\cO(|k|^5+|X|^5).
  \end{align*}
  with
  \begin{align*}
      f^{\rm eff}=|k|^2+X^2(X-1)^2.
  \end{align*} 
One can also consider the case
\begin{align*}
    f^{\rm eff}=  |k|^2+|X|^4.
\end{align*}
In the latter, the $\sqrt{\eps}$ scaling will be changed, and as a result, we need to modify the definition of $s_1$ and $s_2$ in Assumption \ref{ass:bandstructure-nondegenerate}.

We also want to study the twisted-bilayer graphene problem. For this problem, we refer to \cite{cances2023simple} for the formal relationship between $h(k,X)$ and the symbol $f^{\rm eff}$ of Bistritzer-MacDonald effective Hamiltonian. Thus the main problem becomes to find approximate eigenfunction of Bistritzer-MacDonald operator $\mathcal{F}^{-1}{\rm Op}_\eps(f^{\rm eff})\mathcal{F}$ satisfying some assumptions analogous to Assumption \ref{ass:m}. But this is still nontrivial. It is also possible to study the edge state of honeycomb materials in incommensurate case. We refer to \cite{Drouor-Weinstein-Edgestates,Ying-Zhu-Edgestates} for the commensurate case.
\end{remark}

\section{Effective Hamiltonian \texorpdfstring{$\mathfrak{h}_{\eps}^{\rm eff}$}{}}\label{sec:effective-hamiltonian}
To prove Theorem \ref{th:m} and study $H_\eps$, we need to construct an effective Hamiltonian $ \mathfrak{h}_{\eps}^{\rm eff}$ that describes approximate eigenfunction $\Phi_\eps(\bullet)$. In this section, we define this key operator and study its properties. 

\subsection{Construction of the effective Hamiltonian \texorpdfstring{$\mathfrak{h}_{\eps}^{\rm eff}$}{}}
Concerning the Hamiltonian $H_\eps$, we mainly focus on its behavior around the energy level $e_0$. Thus characterized by the energy level $e_0$, the effective Hamiltonian naturally splits into two parts:
\begin{align}\label{h-eff-total}
     \mathfrak{h}_{\eps}^{\rm eff}(\vec u\otimes \vec v):= \mathfrak{h}_{1,\eps}^{\rm eff}(\vec u\otimes \vec v)+ \mathfrak{h}_{2,\eps}^{\rm eff}(\vec u\otimes \vec v)
\end{align}
where
\begin{itemize}
    \item The first part, $\mathfrak{h}_{1,\eps}^{\mathrm{eff}}$, quantizes the local spectral structure described in Assumption~\ref{ass:bandstructure-nondegenerate}:
    \begin{align}
    \mathfrak{h}_{1,\eps}^{\rm eff}(\vec u\otimes\vec v):= \vec{w}\otimes\Big(\mathcal{F}^{-1}\Op(h^{\rm eff}) \mathcal{M}_{\vec u} \mathcal{F} \vec{v}\Big)
\end{align}
where the $J\times J$ matrix-valued functional $\vec u\mapsto \mathcal{M}_{\vec u}\in \mathbb M_{J\times J}$ is defined by
\begin{align}\label{eq:M-u}
(\mathcal{M}_{\vec{u}})_{ij}=\left<w_i,u_j\right>_{L^2_{\rm per}}
\end{align}
In ideal cases such as $m=1$, this term is the leading contribution to $\mathfrak{h}_{\eps}^{\mathrm{eff}}$.

\item The second part, $\mathfrak{h}_{2,\eps}^{\mathrm{eff}}$, accounts for effects arising from the loss of translation invariance in $H_\eps$. This correction arises from the Taylor expansion \eqref{eq:taylor expansion} and is represented by
    \begin{align}
 \MoveEqLeft  \mathfrak{h}_{2,\eps}^{\rm eff}(\vec u\otimes\vec v):=\!\!\!\!\!\!\sum_{\substack{\beta,\gamma\in \N^d\\ |\beta+\gamma|_1\leq m}}\!\!\!\! \frac{1}{\beta!\gamma!}\left(h^\bot_{e,\beta,\gamma,0} \vec u\right)\otimes \Big(\mathcal{F}^{-1}\Op((k-k_0)^\beta(X-X_0)^\gamma ) \mathcal{F}\vec v \Big)
\end{align}
where $h^\bot_{e,\beta,\gamma}$ is defined by \eqref{h-e-beta-alpha-0}. In particular, if $\beta=\gamma=0$, 
\begin{align*}
    h^\bot_{e,\beta,\gamma,0}=h_{e,0}^\bot=h_{e,0}.
\end{align*}
\end{itemize}
\subsection{Study of \texorpdfstring{$\mathfrak{h}_{1,\eps}^{\rm eff}$}{}}
We now apply the scaling–translation operator $T_\eps$ defined in \eqref{eq:scaling-translation} to $\mathfrak{h}_{1,\eps}^{\mathrm{eff}}$. By Lemma \ref{lem:scaling},
\begin{align*}
   \mathfrak{h}_{1,\eps}^{\rm eff} (m) =\eps^{\frac{m}{2}} (\1_{L^2_{\rm per}}\otimes T_\eps)\big(\mathfrak{h}_{\rm main}^{\rm eff}(m)\big)(\1_{L^2_{\rm per}}\otimes T_\eps)^{-1}
\end{align*}
where
\begin{align*}
    \mathfrak{h}_{\rm main}^{\rm eff}(m)(\vec u\otimes \vec v):&=\vec{w}\otimes \mathcal{F}^{-1}{\rm Op}_1(f_m^{\rm eff}) \mathcal{M}_{\vec u} \mathcal{F}(\vec v).
\end{align*}
Using the definition of $\vec{w}$ and $\mathcal{M}_{\bullet}$,we further obtain
\begin{align}\label{eq:h1-eff-para-decomp}
 \mathfrak{h}_{1,\eps}^{\rm eff}(m)&= \eps^{\frac{m}{2}}  \begin{pmatrix}
       P^\parallel_0 \otimes T_\eps\\ P^\bot_0\otimes T_\eps
     \end{pmatrix}^T \begin{pmatrix}
         \mathfrak{h}_{\rm main}^{\rm eff}(m) &0\\0& 0
     \end{pmatrix}\begin{pmatrix}
         P^\parallel_0\otimes T_\eps^{-1}\\
         P^\bot_0\otimes T_\eps^{-1}
     \end{pmatrix}.\notag\\
  &= \eps^{\frac{m}{2}}  \begin{pmatrix}
       P^\parallel_0 \otimes T_\eps\\ P^\bot_0\otimes T_\eps
     \end{pmatrix}^T \begin{pmatrix}
         \vec{w}\otimes\Big(\mathcal{F}^{-1}{\rm Op}_1(f_m^{\rm eff}) \mathcal{M}_{\vec u} \mathcal{F} \vec{v}\Big) &0\\0& 0
     \end{pmatrix}\begin{pmatrix}
         P^\parallel_0\otimes T_\eps^{-1}\\
         P^\bot_0\otimes T_\eps^{-1}
     \end{pmatrix}
\end{align}
where the projectors $P_0^\parallel$ and $P_0^\bot$ are defined in \eqref{def:P-k0X0}

The structure of $\mathfrak{h}_{1,\eps}^{\mathrm{eff}}$ is rather involved. However, since we are primarily interested in states near $e_0$, $k_0$, and $X_0$, the case $\vec u=\vec w$ is of particular importance. In this case,
\begin{align*}
    \mathcal{M}_{\vec u}=\1_{J\times J},\qquad h_{e,0}\vec{u}(x)=0.
\end{align*}
Thus
\begin{align*}
     \mathfrak{h}_{1,\eps}^{\rm eff}(m)(\vec u\otimes \vec v)= \eps^{\frac{m}{2}} \vec{w}\otimes \Big(T_\eps\mathcal{F}^{-1}{\rm Op}_1(f_m^{\rm eff}) \mathcal{F}T_\eps^{-1}\vec{v}\Big).
\end{align*}
It therefore suffices to study the operator $\mathcal{F}^{-1}{\rm Op}_1(f_m^{\rm eff}) \mathcal{F}$.

\subsection{Study of \texorpdfstring{$\mathfrak{h}_{2,\eps}^{\rm eff}$}{}}
The effective Hamiltonian $\mathfrak{h}_{2,\eps}^{\rm eff}$ is also quite complicated. In this paper, we restrict to  $m\leq 2$. Analogously to the previous subsection, we apply $T_\eps$ and use Lemma~\ref{lem:scaling}: for $m=1$,
 \begin{align*} 
 \MoveEqLeft  \mathfrak{h}_{2,\eps}^{\rm eff}(m=1):=h_{e,0}\otimes\1_{L^2(\R^d)}+\sqrt{\eps}\!\!\!\!\sum_{\substack{\beta,\gamma\in \N^d\\ |\beta+\gamma|_1= 1}}\!\!\! h^\bot_{e,\beta,\gamma,0} \otimes \Big(T_\eps\mathcal{F}^{-1}{\rm Op}_1(k^\beta X^\gamma ) \mathcal{F} T_\eps^{-1}\Big),
\end{align*}
and for $m=2$
\begin{align*}
     \MoveEqLeft  \mathfrak{h}_{2,\eps}^{\rm eff}(m=2):=h_{e,0}\otimes\1_{L^2(\R^d)}\\
 &+\sqrt{\eps}\!\!\!\!\sum_{\substack{\beta,\gamma\in \N^d\\ |\beta+\gamma|_1= 1}}\!\!\! h^\bot_{e,\beta,\gamma,0} \otimes \Big(T_\eps\mathcal{F}^{-1}{\rm Op}_1(k^\beta X^\gamma ) \mathcal{F} T_\eps^{-1}\Big)+ \eps\!\!\!\!\sum_{\substack{\beta,\gamma\in \N^d\\ |\beta+\gamma|_1=2}}\!\!\! \frac{1}{\beta!\gamma!} h^\bot_{e,\beta,\gamma,0} \otimes \Big(T_\eps\mathcal{F}^{-1}{\rm Op}_1(k^\beta X^\gamma ) \mathcal{F}T_\eps^{-1}\Big).
\end{align*}

We now decompose $\mathfrak{h}_{2,\eps}^{\mathrm{eff}}$ using the projectors $P_0^\bot$ and $P_0^\parallel$, as in \eqref{eq:h1-eff-para-decomp}. To this end, we claim that
\begin{proposition}\label{prop:h-e-bot-derivatives}
 For $\beta,\gamma\in \N^d$ and $|\beta+\gamma|_1=1$,
    \begin{align*}
        h^\bot_{e,\beta,\gamma,0}=\begin{pmatrix}
            P_0^\parallel\\P_0^\bot
        \end{pmatrix}^T \begin{pmatrix}
            0&  P^\bot_{\beta,\gamma,0} h_{e,0}\\ h_{e,0} P^\bot_{\beta,\gamma,0}&h_{e,\beta,\gamma,0}
        \end{pmatrix}\begin{pmatrix}
            P_0^\parallel\\P_0^\bot
        \end{pmatrix}.
    \end{align*}
\end{proposition}
\begin{proof}
For simplicity, we treat the case $|\beta|_1=1$, $\gamma=0$; the case $\beta=0$, $|\gamma|_1=1$ is analogous. Since $P^\bot(k,X)$ is a projector,
    \begin{align*}
    P^\bot_{\beta,\gamma,0}  =(\partial_k^\beta  P^\bot)(k_0,X_0)= (\partial_k^\beta  (P^\bot)^2)(k_0,X_0)= P^\bot_{\beta,\gamma,0} P^\bot_0+P^\bot_0 P^\bot_{\beta,\gamma,0}.
    \end{align*}
Thus,
\begin{align}
   P^\parallel_0  P^\bot_{\beta,\gamma,0}P^\parallel_0 =0,
\end{align}
and
\begin{align}
     P^\bot_0  P^\bot_{\beta,\gamma,0}P^\bot_0=0
\end{align}
where the second identity follows from the equality 
\begin{align*}
   P^\bot_0  P^\bot_{\beta,\gamma,0}P^\bot_0 =2 P^\bot_0  P^\bot_{\beta,\gamma,0}P^\bot_0,
\end{align*}
which forces the term to vanish. We also refer to e.g., \cite{MR4695788} for the use of this type of formula. As a result,
\begin{align}\label{eq:P-P'-P}
    P^\bot_{\beta,\gamma,0}=P^\bot_0 P^\bot_{\beta,\gamma,0} P^\parallel_0 +P^\parallel_0 P^\bot_{\beta,\gamma,0}P^\bot_0 .
\end{align}
Using these two equation and the fact that
\begin{align*}
    h_e^\bot(k,X)=P^\bot(k,X)h_e(k,X)P^\bot(k,X),
\end{align*}
for $|\beta|_1=1$ and $\gamma=0$,
\begin{align*}
    h^\bot_{e,\beta,\gamma,0}&= P^\bot_0 h^\bot_{e,\beta,\gamma,0}P^\bot_0+P^\bot_{\beta,\gamma,0}h_{e,0}P^\bot_0+P^\bot_0h_{e,0}P^\bot_{\beta,\gamma,0}\\
    &=P^\bot_0 h_{e,\beta,\gamma,0}P^\bot_0+P^\parallel_0 P^\bot_{\beta,\gamma,0}h_{e,0}+h_{e,0}P^\bot_{\beta,\gamma,0}P^\parallel_0
\end{align*}
where in the last equation, we used the fact that
\begin{align*}
    P^\bot_0h_{e,0}=h_{e,0}=h_{e,0}P^\bot_0.
\end{align*}
Thus,
 \begin{align*}
        h^\bot_{e,\beta,\gamma,0}=\begin{pmatrix}
            P_0^\parallel\\P_0^\bot
        \end{pmatrix}^T \begin{pmatrix}
            0&  P^\bot_{\beta,\gamma,0} h_{e,0}\\ h_{e,0} P^\bot_{\beta,\gamma,0}&h_{e,\beta,\gamma,0}
        \end{pmatrix}\begin{pmatrix}
            P_0^\parallel\\P_0^\bot
        \end{pmatrix}
    \end{align*}
This ends the proof.
\end{proof}

As a result of Proposition \ref{prop:h-e-bot-derivatives}, we can now reformulate the effective Hamiltonian $\mathfrak{h}_{2,\eps}^{\rm eff}$ as follows:
\begin{align}
\mathfrak{h}_{2,\eps}^{\rm eff}(m=1)= &\begin{pmatrix}
            P_0^\parallel\otimes T_\eps\\P_0^\bot\otimes T_\eps
        \end{pmatrix}^T \begin{pmatrix}
           0 &  0\\  0& h_{e,0}\otimes \1_{L^2(\R^d)}
        \end{pmatrix}\begin{pmatrix}
            P_0^\parallel\otimes T_\eps^{-1}\\P_0^\bot\otimes T_\eps^{-1}
        \end{pmatrix}\notag\\
    &+\sqrt{\eps} \begin{pmatrix}
            P_0^\parallel\otimes T_\eps\\P_0^\bot\otimes T_\eps
        \end{pmatrix}^T \begin{pmatrix}
            0&    (\mathfrak{h}_{12}^{(1)})^*\\   \mathfrak{h}_{12}^{(1)}&  \mathfrak{h}_{22}^{(1)}
        \end{pmatrix}\begin{pmatrix}
            P_0^\parallel\otimes T_\eps^{-1}\\P_0^\bot\otimes T_\eps^{-1}
        \end{pmatrix},
\end{align}
and
\begin{align}
  \mathfrak{h}_{2,\eps}^{\rm eff}(m=2)=& \begin{pmatrix}
            P_0^\parallel\otimes T_\eps\\P_0^\bot\otimes T_\eps
        \end{pmatrix}^T \begin{pmatrix}
            0&  0\\  0& h_{e,0}\otimes \1_{L^2(\R^d)}
        \end{pmatrix}\begin{pmatrix}
            P_0^\parallel\otimes T_\eps^{-1}\\P_0^\bot\otimes T_\eps^{-1}
        \end{pmatrix}\notag\\
    &+\sqrt{\eps} \begin{pmatrix}
            P_0^\parallel\otimes T_\eps\\P_0^\bot\otimes T_\eps
        \end{pmatrix}^T \begin{pmatrix}
            0&    (\mathfrak{h}_{12}^{(1)})^*\\   \mathfrak{h}_{12}^{(1)}&  \mathfrak{h}_{22}^{(1)}
        \end{pmatrix} \begin{pmatrix}
            P_0^\parallel\otimes T_\eps^{-1}\\P_0^\bot\otimes T_\eps^{-1}
        \end{pmatrix}\notag\\
        &+\eps \begin{pmatrix}
            P_0^\parallel\otimes T_\eps\\P_0^\bot\otimes T_\eps
        \end{pmatrix}^T \begin{pmatrix}
             \mathfrak{h}^{(2)}&   \mathfrak{h}^{(2)}\\ \mathfrak{h}^{(2)}& \mathfrak{h}^{(2)}
        \end{pmatrix}\begin{pmatrix}
            P_0^\parallel\otimes T_\eps^{-1}\\P_0^\bot\otimes T_\eps^{-1}
        \end{pmatrix}
\end{align}
where the Hamiltonian $\mathfrak{h}_{12}^{(1)}$, $\mathfrak{h}_{22}^{(1)}$ and $\mathfrak{h}^{(2)}$ are defined by \eqref{eq:h12-1}, \eqref{eq:h22-1} and \eqref{eq:h-2} respectively.

\subsection{The Hamiltonian \texorpdfstring{$\mathfrak{h}_{\eps}^{\rm eff}$ for $m=1,2$}{}}
Combining the above, we obtain the following explicit forms of the effective Hamiltonian for $m=1$ and $m=2$:
\begin{align}\label{eq:h-scaling-1}
\mathfrak{h}_{\eps}^{\rm eff}(m=1)=& \begin{pmatrix}
            P_0^\parallel\otimes T_\eps\\P_0^\bot\otimes T_\eps
        \end{pmatrix}^T \begin{pmatrix}
            \sqrt\eps \mathfrak{h}_{\rm main}^{\rm eff}(m=1) &  0\\  0& h_{e,0}\otimes \1_{L^2(\R^d)}
        \end{pmatrix}\begin{pmatrix}
            P_0^\parallel\otimes T_\eps^{-1}\\P_0^\bot\otimes T_\eps^{-1}
        \end{pmatrix}\notag\\
    &+\sqrt{\eps} \begin{pmatrix}
            P_0^\parallel\otimes T_\eps\\P_0^\bot\otimes T_\eps
        \end{pmatrix}^T \begin{pmatrix}
            0&    (\mathfrak{h}_{12}^{(1)})^*\\   \mathfrak{h}_{12}^{(1)}&  \mathfrak{h}_{22}^{(1)}
        \end{pmatrix}\begin{pmatrix}
            P_0^\parallel\otimes T_\eps^{-1}\\P_0^\bot\otimes T_\eps^{-1}
        \end{pmatrix},
\end{align}
and
\begin{align}\label{eq:h-scaling-2}
\mathfrak{h}_{\eps}^{\rm eff}(m=2)= &\begin{pmatrix}
            P_0^\parallel\otimes T_\eps\\P_0^\bot\otimes T_\eps
        \end{pmatrix}^T \begin{pmatrix}
             \eps\mathfrak{h}_{\rm main}^{\rm eff}(m=2)&  0\\  0& h_{e,0}\otimes \1_{L^2(\R^d)}
        \end{pmatrix}\begin{pmatrix}
            P_0^\parallel\otimes T_\eps^{-1}\\P_0^\bot\otimes T_\eps^{-1}
        \end{pmatrix}\notag\\
    &+\sqrt{\eps} \begin{pmatrix}
            P_0^\parallel\otimes T_\eps\\P_0^\bot\otimes T_\eps
        \end{pmatrix}^T \begin{pmatrix}
            0&    (\mathfrak{h}_{12}^{(1)})^*\\   \mathfrak{h}_{12}^{(1)}&  \mathfrak{h}_{22}^{(1)}
        \end{pmatrix} \begin{pmatrix}
            P_0^\parallel\otimes T_\eps^{-1}\\P_0^\bot\otimes T_\eps^{-1}
        \end{pmatrix}\notag\\
        &+\eps \begin{pmatrix}
            P_0^\parallel\otimes T_\eps\\P_0^\bot\otimes T_\eps
        \end{pmatrix}^T \begin{pmatrix}
             \mathfrak{h}^{(2)}&   \mathfrak{h}^{(2)}\\ \mathfrak{h}^{(2)}& \mathfrak{h}^{(2)}
        \end{pmatrix}\begin{pmatrix}
            P_0^\parallel\otimes T_\eps^{-1}\\P_0^\bot\otimes T_\eps^{-1}
        \end{pmatrix}.
\end{align}

\section{From the effective reduction to approximate eigenfunctions}\label{sec:WKB-approximation}
In this section, we prove Theorem \ref{th:m}. To do so, we need the following Theorem \ref{th:quadratic} which shows that $H_\eps$ can be approximated by the effective Hamiltonian $\mathfrak{h}_\eps^{\rm eff}$. The proof of Theorem \ref{th:m} then reduces to a WKB approximation for the eigenvalue problem of $\mathfrak{h}_\eps^{\rm eff}$ up to an error of order $\cO(\eps^{\frac{m}{2}+\frac{1}{4}})$. The case $m=2$ requires a second corrector and the identification of $\widetilde{\mathcal M}$.

\medskip

Before going further, we need the following assumption.
\begin{assumption}[Localization]\label{ass:localization}
We assume that $\vec{v}_\eps:=(v_{1,\eps},\cdots, v_{J,\eps})^T$ is in $W^{d+m+2,1}(\R^d)$ such that 
    \begin{align*}
    \|\vec{v}_\eps\|_{W^{d+m+2,1}(\R^d\setminus B_{\eps^{-s_1}}(\eps^{-1}X_0))}=\cO(\eps^{\frac12(m+1)}),
\end{align*}
and for any $M\in \N$,
\begin{align*}
  \epsilon^{\frac{d }{4}}\|\Delta^{M}e^{-ik_0\cdot x} \vec{v}_\eps\|_{L^1(\R^d)}=\cO(\eps^{M}).
\end{align*}

\end{assumption}
This assumption is satisfied, for example,  when $\vec{v}_\eps(x)=T_\eps(\vec{v})$ with $T_\eps$ being defined by \eqref{eq:scaling-translation} and $\vec{v}\in C^\infty(\R^d)$ having exponential decay.

The main operator-valued reduction is the following.
\begin{theorem}[Reduction of the Hamiltonian]\label{th:quadratic}
Let $\vec{u}=(u_1,\cdots,u_J)^T$ be in $H^2_{\rm per}$ and let $\vec{v}_\eps$ satisfy Assumption \ref{ass:localization}. 
Then under Assumption \ref{ass:bandstructure-nondegenerate}, for $\eps$ small enough
\begin{align}\label{eq:Hepsilon-Phi-epsilon}
\| (H_\eps  -e_0)\Phi_\eps(\vec u\otimes \vec v_\eps) -\Phi_\eps(\mathfrak{h}_\eps^{\rm eff}(m)(\vec u\otimes \vec v_\eps)) \|_{L^2(\R^d;\C^n)}=\cO( \eps^{\frac{1}{2}m+\frac{1}{4}}).
\end{align}
\end{theorem}
The proof of Theorem \ref{th:quadratic} will be provided in Section \ref{sec:proofofquadratic} later.

\begin{remark}
  The effective Hamiltonian $\mathfrak{h}_\eps^{\rm eff}$ is constructed primarily from Assumption \ref{ass:bandstructure-nondegenerate}. The reduction result can be extended if different asymptotic behaviors of the mapping $(k,X)\mapsto E_j(k,X)$ and $(k,X)\mapsto \phi_j(k,X)$ are available.
\end{remark}

As a consequence of Theorem \ref{th:quadratic}, we have the following.
\begin{corollary}[WKB approximation]\label{th:main}
Let $N\in \N^+$.  For any $0\leq j\leq N$, let $\vec{u}_j=(u_{j,1},\cdots,u_{j,J})^T$ be in $H^2_{\rm per}$ and let $\vec{v}_{j,\eps}$ satisfy Assumption \ref{ass:localization}. If in addition a sequence $(\vec u_j\otimes \vec{v}_{j,\eps})_{0\leq j\leq N}$ is a WKB approximate eigenfunction of $\mathfrak{h}_\eps^{\rm eff}(m)$ with approximated eigenvalue $\mu_\eps\in \R$ in the following sense
    \begin{align}\label{eq:eigen-problem-h-eff}
        \left\|\sum_{j=0}^N\Phi_\eps\left(\mathfrak{h}_\eps^{\rm eff}(m)(\vec u_j\otimes\vec v_{j,\eps})\right)- \mu_\eps \sum_{j=0}^N\Phi_\eps\left(\vec u_j \otimes \vec v_{j,\eps}\right) \right\|_{ L^2(\R^d)}=\cO(\eps^{\frac{1}{2}m+\frac{1}{4}}), 
    \end{align}
    then
    \begin{align}\label{eq:eigen-main}
         \left\| (H_\eps  -e_0-\mu_\eps)\sum_{j=0}^N\Phi_\eps(\vec u_j\otimes\vec v_{j,\eps})\right\|_{L^2(\R^d)}=\cO(\eps^{\frac{1}{2}m+\frac{1}{4}}).
    \end{align}
\end{corollary}
\begin{proof}
    This is a direct consequence of Theorem \ref{th:quadratic}. Indeed, we have
    \begin{align*}
   \MoveEqLeft   \left\| (H_\eps  -e_0-\mu_\eps)\sum_{j=0}^N\Phi_\eps(\vec u_j\otimes\vec v_{j,\eps})\right\|_{L^2(\R^d)}\\
   &\leq \left\| (H_\eps  -e_0)\sum_{j=0}^N\Phi_\eps(\vec u_j\otimes\vec v_{j,\eps})-\sum_{j=0}^N\Phi_\eps\left(\mathfrak{h}_\eps^{\rm eff}(\vec u_j\otimes \vec v_{j,\eps})\right)\right\|_{L^2(\R^d)}\\
   &\quad+\left\|\sum_{j=0}^N\Phi_\eps\left(\mathfrak{h}_\eps^{\rm eff}(\vec u_j\otimes\vec v_{j,\eps})\right)- \mu_\eps \sum_{j=0}^N\Phi_\eps\left(\vec u_j \otimes \vec v_{j,\eps}\right) \right\|_{ L^2(\R^d)}=\cO(\eps^{\frac{m}{2}+\frac{1}{4}}).
    \end{align*}
    This ends the proof.
\end{proof}

We now apply Corollary \ref{th:main} to prove Theorems \ref{th:m}. We begin with the simpler case $m=1$.
\subsection{Proof of Theorem \ref{th:m}, case \texorpdfstring{$m=1$}{}}
Before going further, for simplicity, in this case, we set
\begin{align*}
   \mathfrak{h}_\eps^{\rm eff}= \mathfrak{h}_\eps^{\rm eff}(m=1),\qquad\mbox{and}\quad  \mathfrak{h}_{\rm main}^{\rm eff}= \mathfrak{h}_{\rm main}^{\rm eff}(m=1).
\end{align*}
We rewrite $U^{(0)}$ and $U^{(1)}$ by using a sequence of functions $\vec{u}_j$ in $L^2_{\rm per}$ and $\vec{v}_j$ in $L^2(\R^d)$ for $0\leq j\leq 2d$
\begin{align*}
    \vec{u}_0\otimes \vec{v}_0=\vec{w}\otimes \vec{v}_*=U^{(0)}(\vec{w}\otimes \vec{v}_*)\qquad  \vec{u}_0\otimes T_\eps(\vec{v}_0)=U_\eps^{(0)}(\vec{w}\otimes \vec{v}_*)
\end{align*}
and
\begin{align*}
   \sum_{j=1}^{2d} \vec{u}_j\otimes \vec{v}_j=U^{(1)}(\vec{w}\otimes \vec{v}_*),\qquad    \sum_{j=1}^{2d}\vec{u}_j\otimes T_\eps(\vec{v}_j)=U_\eps^{(1)}(\vec{w}\otimes \vec{v}_*).
\end{align*}
According to the definition of $U^{(1)}(\bullet)$, for $1\leq j\leq 2d$ and for some $\beta,\gamma\in \N^d$ with $|\beta|_1=|\gamma|_1=1$
\begin{align*}
    \vec{v}_j\in\{y^\gamma \vec v_*,\; \partial^\beta \vec v_*\}.
\end{align*}
Under Assumption \ref{ass:m} on $\vec{v}_*$, it is straightforward to check that $(T_\eps(\vec{v}_j))_{0\leq j\leq 2d}$ satisfies Assumption \ref{ass:localization}.

To finish the proof of Theorem \ref{th:m}, it remains to verify \eqref{eq:eigen-problem-h-eff}. From the definition of $U_\eps^{(0)}$ and $U_\eps^{(1)}$,
\begin{align*}
    \mathfrak{h}_{\rm main}^{\rm eff} U^{(0)}(\vec{w}\otimes \vec{v}_*)=\mu_*\vec{w}\otimes \vec{v}_*=\mu_* U^{(0)}(\vec{w}\otimes \vec{v}_*)
\end{align*}
and
\begin{align*}
    \mathfrak{h}_{12}^{(1)} U^{(0)}(\vec{w}\otimes \vec{v}_*)+ h_{e,0}^\bot\otimes \1_{L^2(\R^d)} U^{(1)}(\vec{w}\otimes \vec{v}_*)=0
\end{align*}
where we used again the fact that $h_{e,0}=P^\bot_0h_{e,0}$. Moreover, by \eqref{eq:U0-eps} and \eqref{eq:U1-eps},
\begin{align*}
    U^{(0)}_\eps(\vec{w}\otimes \vec{v}_*)\in (P^\parallel_0 L^2_{\rm per})\otimes L^2(\R^d)
\end{align*}
and
\begin{align*}
   U^{(1)}_\eps(\vec{w}\otimes \vec{v}_*) \in (P^\bot_0 L^2_{\rm per})\otimes L^2(\R^d).
\end{align*}
Using \eqref{eq:h-scaling-1},
\begin{align*}
 \MoveEqLeft    (\mathfrak{h}_\eps^{\rm eff}-\sqrt{\eps}\mu_*)(U_\eps^{(0)}+\sqrt{\eps} U_\eps^{(1)})(\vec{w}\otimes \vec{v}_*)\\
 &=\eps (\1_{L^2_{\rm per}}\otimes T_\eps) \Big((\mathfrak{h}_{12}^{(1)})^* +\mathfrak{h}_{22}^{(1)}-\mu_* \Big)U^{(1)}(\vec{w}\otimes \vec{v}_*).
\end{align*}
Thus,
\begin{align*}
  \MoveEqLeft  \Phi_\eps\Big( (\mathfrak{h}_\eps^{\rm eff}-\sqrt{\eps}\mu_*)(U_\eps^{(0)}+\sqrt{\eps} U_\eps^{(1)})(\vec{w}\otimes \vec{v}_*)\Big)\\
    &=\eps \Phi_\eps\Big( (\1_{L^2_{\rm per}}\otimes T_\eps) \Big((\mathfrak{h}_{12}^{(1)})^* +\mathfrak{h}_{22}^{(1)}-\mu_* \Big)U^{(1)}(\vec{w}\otimes \vec{v}_*)\Big)
\end{align*}
Using Lemma \ref{lem:Uf-cH} and the condition \eqref{eq:v2-m} in Assumption \ref{ass:m}, we obtain
\begin{align*}
  \MoveEqLeft  \left\|\Phi_\eps\Big( (\1_{L^2_{\rm per}}\otimes T_\eps) \Big((\mathfrak{h}_{12}^{(1)})^* +\mathfrak{h}_{22}^{(1)}-\mu_* \Big)U^{(1)}(\vec{w}\otimes \vec{v}_*)\Big)\right\|_{L^2(\R^d;\C^n)}\\
  &\lesssim \|\vec{w}\|_{H^{4}_{\rm per}}\sum_{\gamma\in \N^d,\; |\gamma|_1\leq 2m}\left(\|y^\gamma\vec{v}_*\|_{W^{d+2,1}(\R^d)}+\|\partial^\gamma\vec{v}_*\|_{W^{d+2,1}(\R^d)}\right)<\infty
\end{align*}
where we used the fact that $\vec{w}\in C^\infty(\R^d)$ since $x\mapsto \bA(x,X_0)$ and $x\mapsto V(x,X_0)$ are smooth. Hence,
\begin{align*}
    \left\| \Phi_\eps\Big( (\mathfrak{h}_\eps^{\rm eff}-\sqrt{\eps}\mu_*)(U_\eps^{(0)}(\vec{w}\otimes \vec{v}_*)+\sqrt{\eps} U_\eps^{(1)}(\vec{w}\otimes \vec{v}_*))\Big)\right\|_{L^2(\R^d;\C^n)}=\cO(\eps).
\end{align*}
This estimate and Corollary \ref{th:main}  yield \eqref{eq:m}. 

It remains to prove the normalization estimate: 
\begin{align*}
    \left\|\Phi_\eps\Big((U_\eps^{(0)}+\sqrt\eps U_\eps^{(1)}\big)(\vec{w}\otimes \vec{v}_*)\Big)\right\|_{L^2(\R^d;\C^n)}=\frac{1}{|\Omega|^{1/2}}+\cO(\eps^{1/2}).
\end{align*}
By Assumption \ref{ass:m}, Lemma \ref{lem:Uf-cH} and the fact that operators $P_{\beta,\gamma,0}^\bot$ are bounded on $L^2_{\rm per}$  for $|\beta+\gamma|_1\leq 1$,
\begin{align*}
  \MoveEqLeft  \|\Phi_\eps(\sqrt{\eps} U_\eps^{(1)})(\vec{w}\otimes \vec{v}_*)\|_{L^2(\R^d;\C^n)}=\|\cU(\Phi_\eps(\sqrt{\eps} U_\eps^{(1)})(\vec{w}\otimes \vec{v}_*))\|_\cH\\
  \lesssim& \sqrt{\eps}\sum_{\substack{\beta,\gamma\in \N^d,\\ \gamma=0,\;|\beta|_1=1}} \|P^\bot_0P_{\beta,\gamma,0}^\bot  \vec{ w}\|_{L^2_{\rm per}}\left(\|\partial_y^\beta \vec v_*\|_{L^2(\R^d)}+\|\partial_y^\beta \vec v_*\|_{W^{d+1,1}(\R^d)}\right)\\
    &+\sqrt{\eps}\sum_{\substack{\beta,\gamma\in \N^d,\\ \beta=0,\;|\gamma|_1=1}} \|P^\bot_0P_{\beta,\gamma,0}^\bot  \vec{ w}\|_{L^2_{\rm per}}\left(\|y^\gamma \vec v_*\|_{L^2(\R^d)}+\|y^\gamma \vec v_*\|_{W^{d+1,1}(\R^d)}\right) =\cO(\sqrt{\eps})
\end{align*}
where we used the fact that $\chi(\eps^{s_1} x-\eps^{s_1-1}X_0)$ and its derivatives are uniformly bounded independently of $\eps$, and the fact that for any function $\vec f\in L^2(\R^d)$,
\begin{align*}
   \MoveEqLeft \fint_{\Omega^*}\left|\mathcal{F}\Big(\chi(\eps^{s_1} \cdot-\eps^{s_1-1}X_0) \vec{f}\Big)\right|^2(k)dk\\
   &\lesssim \int_{\R^d}\left|\mathcal{F}\Big(\chi(\eps^{s_1} \cdot-\eps^{s_1-1}X_0) \vec{f}\Big)\right|^2(k)dk\\
   &\lesssim \int_{\R^d}\left|\chi(\eps^{s_1} \cdot-\eps^{s_1-1}X_0) \vec{f}\right|^2(x)dx\lesssim \|\vec{f}\|_{L^2(\R^d)}^2.
\end{align*}
Thus,
\begin{align*}
    \left\|\Phi_\eps\Big((U_\eps^{(0)}+\sqrt\eps U_\eps^{(1)}\big)(\vec{w}\otimes \vec{v}_*)\Big)\right\|_{L^2(\R^d;\C^n)}= \|\Phi_\eps(U_\eps^{(0)})(\vec{w}\otimes \vec{v}_*)\|_{L^2(\R^d;\C^n)}+\cO(\sqrt{\eps}).
\end{align*}
Recall that
\begin{align*}
    \vec{w}=(w_1,\cdots, w_J)^T,\qquad \vec{v}_*=(v_1,\cdots,v_J)^T.
\end{align*}
It remains to study
\begin{align*}
 \Phi_\eps(U_\eps^{(0)})(\vec{w}\otimes \vec{v}_*)(x)&=\chi(\eps^{s_1} x-\eps^{s_1-1}X_0) \vec{w}^T(x) T_\eps(\vec v_*)(x)\\
 &=\eps^{\frac{d}{4}}\sum_{1\leq j\leq J}\chi(\eps^{s_1} x-\eps^{s_1-1}X_0)  e^{ik_0\cdot (x-\eps^{-1}X_0)}w_j(x)v_j(\sqrt\eps(x-\eps^{-1}X_0)).
\end{align*}
Let
\begin{align*}
    \widetilde{v}_{\eps,j}=\eps^{\frac{d}{4}}\chi(\eps^{s_1} x-\eps^{s_1-1}X_0) v_j(\sqrt\eps(x-\eps^{-1}X_0)).
\end{align*}
Thus,
\begin{align*}
    \|\Phi_\eps(U_\eps^{(0)})(\vec{w}\otimes \vec{v}_*)\|_{L^2(\R^d;\C^n)}&=\left\|\sum_{1\leq j\leq J}e^{ik_0\cdot (x-\eps^{-1}X_0)} w_j  \widetilde{v}_{\eps,j} \right\|_{L^2(\R^d;\C^n)}=\left\|\sum_{1\leq j\leq J} w_j  \widetilde{v}_{\eps,j} \right\|_{L^2(\R^d;\C^n)}.
\end{align*}
Using Bloch transform and Lemma \ref{lem:Uf-cH},
\begin{align*}
  \MoveEqLeft  \left\|\sum_{1\leq j\leq J} w_j  \widetilde{v}_{\eps,j} \right\|_{L^2(\R^d;\C^n)}=\left\|\sum_{1\leq j\leq J} \cU(w_j  \widetilde{v}_{\eps,j}) \right\|_{\cH}\\
    &=|\Omega|^{-1}\left(\fint_{\Omega^*}\left\|\sum_{1\leq j\leq J} w_j(x)  \mathcal{F}(\widetilde{v}_{\eps,j})(k) \right\|_{L^2_{\rm per}}^2dk \right)^{1/2}+\cO(\eps)\\
    &=\frac{1}{(2\pi)^{d/2}|\Omega|^{1/2}}\left(\sum_{1\leq j\leq J}\int_{\Omega^*}\left|  \mathcal{F}(\widetilde{v}_{\eps,j})(k) \right|^2dk \right)^{1/2}+\cO(\eps)
\end{align*}
where in the last one, we used the fact that $ \left<w_j,w_l\right>_{L^2_{\rm per}}=\delta_{jl}$ according to Assumption \ref{ass:bandstructure-nondegenerate}. Since ${\rm dist}(0, \R^d\setminus \Omega^*)>0$,
\begin{align*}
  \MoveEqLeft  \int_{\Omega^*}\left|  \mathcal{F}(\widetilde{v}_{\eps,j})(k) \right|^2dk\\
    &=  \int_{\R^d}\left|  \mathcal{F}(\widetilde{v}_{\eps,j})(k) \right|^2dk-  \int_{\R^d\setminus \Omega^*}\left|  \mathcal{F}(\widetilde{v}_{\eps,j})(k) \right|^2dk\\
    &={(2\pi)^d}\int_{\R^d}\left|  \widetilde{v}_{\eps,j} \right|^2dx-  \int_{\R^d\setminus \Omega^*} \frac{1}{|k|^{4d}}\left|  \mathcal{F}\Big((-\Delta)^d(\widetilde{v}_{\eps,j})\Big) \right|^2dk \\
    &=\eps^{\frac{d}{2}}(2\pi)^d\int_{\R^d}\left|  \chi(\eps^{s_1} x-\eps^{s_1-1}X_0) v_j(\sqrt\eps(x-\eps^{-1}X_0)) \right|^2dx +\cO(\eps^{d})\\
    &={(2\pi)^d}\int_{\R^d}|v_j|^2dx+\cO(\sqrt\eps)
\end{align*}
where in the last step we used \eqref{eq:v3-m}.

Therefore,
\begin{align*}
      \|\Phi_\eps(U_\eps^{(0)})(\vec{w}\otimes \vec{v}_*)\|_{L^2(\R^d;\C^n)}=\frac{1}{|\Omega|^{1/2}}\|\vec{v}_*\|_{L^2(\R^d)}+\cO(\sqrt\eps)= \frac{1}{|\Omega|^{1/2}}+\cO(\sqrt\eps)
\end{align*}
and
\begin{align*}
    \left\|\Phi_\eps\Big((U_\eps^{(0)}+\sqrt\eps U_\eps^{(1)}\big)(\vec{w}\otimes \vec{v}_*)\Big)\right\|_{L^2(\R^d;\C^n)}=\frac{1}{|\Omega|^{1/2}}+\cO(\sqrt\eps).
\end{align*}
This ends the proof of Theorem \ref{th:m}.

\subsection{Proof of Theorem \ref{th:m}, case \texorpdfstring{$m=2$}{}}\label{sec:proof-m=2}
For simplicity, in this case we set
\begin{align*}
   \mathfrak{h}_\eps^{\rm eff}= \mathfrak{h}_\eps^{\rm eff}(m=2),\qquad\mbox{and}\quad  \mathfrak{h}_{\rm main}^{\rm eff}= \mathfrak{h}_{\rm main}^{\rm eff}(m=2).
\end{align*}
Let $(\vec{v}_*,\mu_*)$ be the eigenpair from Assumption \ref{ass:m}. Analogous to the proof of Theorem \ref{th:m}, it is easy to see that the $\vec{v}$ part of the state $U^{(0)}_\eps(\vec{w}\otimes \vec{v}_*)$, $U^{(1)}_\eps(\vec{w}\otimes \vec{v}_*)$ and $U^{(2)}_\eps(\vec{w}\otimes \vec{v}_*)$ satisfy Assumption \ref{ass:localization}. From the definition of $U^{(0)}_\eps$, $U^{(1)}_\eps$ and $U^{(2)}_\eps$, we have
\begin{align*}
    U^{(0)}_\eps(\vec{w}\otimes \vec{v}_*)\in (P^\parallel_0 L^2_{\rm per})\otimes L^2(\R^d),
\end{align*}
\begin{align*}
  U^{(1)}_\eps(\vec{w}\otimes \vec{v}_*),\qquad   U^{(2)}_\eps(\vec{w}\otimes \vec{v}_*)\in (P^\bot_0 L^2_{\rm per})\otimes L^2(\R^d)
\end{align*}
and
\begin{align*}
  \MoveEqLeft   (\mathfrak{h}_\eps^{\rm eff}-\eps \mu_*)( U^{(0)}_\eps+ \sqrt{\eps} U^{(1)}_\eps+\eps U^{(2)}_\eps)(\vec{w}\otimes \vec{v}_*)\\
  &=\begin{pmatrix}
      P^\parallel_0\otimes T_\eps\\
        P^\bot_0\otimes T_\eps
  \end{pmatrix}^T\left[\begin{pmatrix}
      \eps (\mathfrak{h}_{\rm main}^{\rm eff}-\mu_*)U^{(0)}\\ (h_{e,0}\otimes \1_{L^2(\R^d)}-\eps \mu_*)(\sqrt{\eps}U^{(1)} +\eps U^{(2)})
  \end{pmatrix}\right.\\
  &\quad+ \sqrt{\eps}\begin{pmatrix}
    (\mathfrak{h}_{12}^{(1)})^*(\sqrt{\eps}U^{(1)}+\eps U^{(2)})\\\mathfrak{h}_{12}^{(1)}U^{(0)}+\mathfrak{h}_{22}^{(1)}(\sqrt{\eps} U^{(1)}+\eps U^{(2)})
  \end{pmatrix} \left.+ \eps\begin{pmatrix}
     \mathfrak{h}^{(2)}(U^{(0)}+\sqrt{\eps}U^{(1)}+\eps U^{(2)})\\  \mathfrak{h}^{(2)}(U^{(0)}+\sqrt{\eps}U^{(1)}+\eps U^{(2)})
  \end{pmatrix}\right](\vec{w}\otimes \vec{v}_*)\\
  &=\eps\begin{pmatrix}
      P^\parallel_0\otimes T_\eps\\
        P^\bot_0\otimes T_\eps
  \end{pmatrix}^T \begin{pmatrix}
  (  \mathfrak{h}^{\rm cor} -\mu_*)U^{(0)}\\ 0
  \end{pmatrix} (\vec{w}\otimes \vec{v}_*)\\
  &\quad+\eps^{3/2}\begin{pmatrix}
      P^\parallel_0\otimes T_\eps\\
        P^\bot_0\otimes T_\eps
  \end{pmatrix}^T\begin{pmatrix}
      \big(\mathfrak{h}_{12}^{(1)}\big)^*U^{(2)} +\mathfrak{h}^{(2)}(U^{(1)}+\sqrt{\eps}U^{(2)})\\ (\mathfrak{h}^{(2)}-\mu_*)(U^{(1)}+\sqrt{\eps}U^{(2)})+ \mathfrak{h}_{22}^{(1)}U^{(2)}
  \end{pmatrix}(\vec{w}\otimes \vec{v}_*)
\end{align*}
where we define
\begin{align}\label{eq:h-correction}
    \mathfrak{h}^{\rm cor}U^{(0)}:&=(P^\parallel_0\otimes \1_{L^2(\R^d)})(\mathfrak{h}_{\rm main}^{\rm eff} + \mathfrak{h}^{(2)}) U^{(0)}+(P^\parallel_0\otimes \1_{L^2(\R^d)})(\mathfrak{h}_{12}^{(1)})^*U^{(1)}\notag\\
    &=(P^\parallel_0\otimes \1_{L^2(\R^d)}) \Big[ \mathfrak{h}_{\rm main}^{\rm eff}+ \mathfrak{h}^{(2)}-\big(\mathfrak{h}_{12}^{(1)}\big)^* \Big(h_{e,0}\otimes\1_{L^2(\R^d)}\Big)^{-1}\mathfrak{h}_{12}^{(1)}\Big]U^{(0)}. 
\end{align}
Now we claim that
\begin{lemma}\label{lem:h-cor}
Let $(\vec{v}_*,\mu_*)$ be an  eigenpair satisfying Assumption \ref{ass:m}, and for any $\vec{a}\in L^2_{\rm per}$ and $\vec{b}\in L^2(\R^d)$, let $\Psi$ be a linear operator defined by
\begin{align*}
    \Psi(\vec{a}\otimes \vec{b})(x,y)=\vec{a}^T(x)\vec{b}(y).
\end{align*}
Then,
    \begin{align}\label{eq:5.4}
\Psi\Big(\mathfrak{h}^{\rm cor}U^{(0)}(\vec{w}\otimes \vec{v}_*)\Big)=\mu_* \Psi\Big( U^{(0)}(\vec{w}\otimes \vec{v}_*)\Big).
\end{align}
\end{lemma}
Using Lemma \ref{lem:h-cor},
\begin{align*}
  \MoveEqLeft  \Phi_\eps\Big((P_0^\parallel\otimes T_\eps)\mathfrak{h}^{\rm cor}U^{(0)}(\vec{w}\otimes \vec{v}_*)\Big)(x)\\
  &=\chi(\eps^{s_1} x-\eps^{s_1-1}X_0) \left[T_{\eps,v}\Psi\Big(\mathfrak{h}^{\rm cor}U^{(0)}(\vec{w}\otimes \vec{v}_*)\Big)\right](x,x)\\
    &=\mu_* \chi(\eps^{s_1} x-\eps^{s_1-1}X_0)\left[T_{\eps,v}\Psi\Big( U^{(0)}(\vec{w}\otimes \vec{v}_*)\Big)\right](x,x)=\mu_* \Phi_\eps\Big( U_\eps^{(0)}(\vec{w}\otimes \vec{v}_*)\Big)(x)
\end{align*}
where $T_{\eps,v}$ is the scaling and translation operator $T_{\eps}$ that only acts on the $\vec{v}_*$ component. Arguing as in the proof of Theorem \ref{th:m}, we obtain the result. In particular, the normalization estimate
\begin{align*}
    \left\|\Phi_\eps\Big((U_\eps^{(0)}+\sqrt\eps U_\eps^{(1)}+\eps U_\eps^{(2)}\big)(\vec{w}\otimes \vec{v}_*)\Big)\right\|_{L^2(\R^d;\C^n)}=\frac{1}{|\Omega|^{1/2}}+\cO(\sqrt{\eps})
\end{align*}
follows from the fact that $\vec{w}\in C^\infty(\R^d)\cap L^2_{\rm per}\subset H^2_{\rm per}$ since $\bA(x,X_0), V(x,X_0)$ are smooth w.r.t. $x$.

\subsection{Proof of Lemma \ref{lem:h-cor}}\label{sec:9.3}
To complete the proof of Theorem \ref{th:m}, we now prove Lemma \ref{lem:h-cor}. We first write the explicit formula for $\mathfrak{h}^{\rm cor}$:
\begin{align*}
  \MoveEqLeft  \mathfrak{h}^{\rm cor}U^{(0)}(\vec{w}\otimes \vec{v}_*)(x,y)=\vec{w}(x)\otimes (\mathcal{F}^{-1}{\rm Op}_1(f^{\rm eff}_m)\mathcal{F}\vec{v}_*)(y)\\
  &\quad+\frac{1}{2}\sum_{\substack{\beta,\gamma\in \N^d\\ |\beta+\gamma|_1=2}}\frac{1}{\beta!\gamma!} (P_0^\parallel h^\bot_{e,\beta,\gamma,0} \vec{w})(x)\otimes \Big((-i\partial_y)^\beta y^\gamma+y^\gamma (-i\partial_y)^\beta\Big) \vec{v}_*(y)\\
    &\quad-\sum_{\substack{\beta',\beta'',\gamma',\gamma''\in \N^d\\|\gamma'|_1=|\gamma''|_1=0\\ |\beta'|_1=|\beta''|_1=1}}(P_0^\parallel P_{\beta',\gamma',0}^\bot h_{e,0}^\bot P_{\beta'',\gamma'',0}^\bot \vec{w})(x)\otimes \Big((-i\partial_y)^{\beta'+\beta''}\vec{v}_*(y)  \Big)\\
    &\quad-\sum_{\substack{\beta',\beta'',\gamma',\gamma''\in \N^d\\|\beta'|_1=|\gamma''|_1=0\\ |\gamma'|_1=|\beta''|_1=1}}(P_0^\parallel P_{\beta',\gamma',0}^\bot h_{e,0}^\bot P_{\beta'',\gamma'',0}^\bot \vec{w})(x)\otimes \Big(y^{\gamma'} (-i\partial_y)^{\beta''}\vec{v}_*(y)\Big) \\
     &\quad-\sum_{\substack{\beta',\beta'',\gamma',\gamma''\in \N^d\\|\gamma'|_1=|\beta''|_1=0\\ |\beta'|_1=|\gamma''|_1=1}}(P_0^\parallel P_{\beta',\gamma',0}^\bot h_{e,0}^\bot P_{\beta'',\gamma'',0}^\bot \vec{w})(x)\otimes   \Big((-i\partial_y)^{\beta'} (y^{\gamma''}\vec{v}_*(y))\Big) \\
     &\quad-\sum_{\substack{\beta',\beta'',\gamma',\gamma''\in \N^d\\|\beta'|_1=|\beta''|_1=0\\ |\gamma'|_1=|\gamma''|_1=1}}(P_0^\parallel P_{\beta',\gamma',0}^\bot h_{e,0}^\bot P_{\beta'',\gamma'',0}^\bot \vec{w})(x)\otimes \Big(y^{\gamma'+\gamma''} \vec{v}_*(y)\Big).
\end{align*}

\medskip

Now we write 
\begin{align*}
    \beta=\beta'+\beta'',\qquad\gamma=\gamma'+\gamma''
\end{align*}
with $\beta',\beta'',\gamma',\gamma''\in \N^d$, and $|\beta'|\leq 1$, $|\beta''|\leq 1$, $|\gamma'|\leq 1$ and $|\gamma''|\leq 1$. Since $P^\bot_0 P^\parallel_0=0$ and
\begin{align*}
    h^\bot_{e}(k,X)= P^\bot(k,X)h_e(k,X)P^\bot(k,X),
\end{align*} 
for $|\beta+\gamma|_1=2$, we have
\begin{align*}
P_0^\parallel h^\bot_{e,\beta,\gamma,0}P_0^\parallel&=\sum_{\substack{\beta'+\beta''=\beta\\
\gamma'+\gamma''=\gamma\\ |\beta'+\gamma'|_1=|\beta''+\gamma''|_1=1}}\frac{\beta!\gamma!}{\beta'!\beta''!\gamma'!\gamma''!}   P_0^\parallel P^\bot_{\beta',\gamma',0}h_{e,0}P^\bot_{\beta'',\gamma'',0}P_0^\parallel\\
&=\sum_{\substack{\beta'+\beta''=\beta\\
\gamma'+\gamma''=\gamma\\ |\beta'+\gamma'|_1=|\beta''+\gamma''|_1=1}}(\beta!\gamma!)  \;\; P_0^\parallel P^\bot_{\beta',\gamma',0}h_{e,0}P^\bot_{\beta'',\gamma'',0}P_0^\parallel
\end{align*}
where we used the fact that
\begin{align*}
    \beta'!\beta''!\gamma'!\gamma''!=1.
\end{align*}
Thus, as $\vec{w}=P^\parallel_0\vec{w}$, 
\begin{align*}
 \MoveEqLeft  \sum_{\substack{\beta,\gamma\in \N^d\\ |\beta+\gamma|_1=2}} \frac{1}{\beta!\gamma!}(P_0^\parallel h^\bot_{e,\beta,\gamma,0} \vec{w})(x)\otimes\Big((-i\partial_y)^\beta y^\gamma+y^\gamma (-i\partial_y)^\beta\Big) \vec{v}_*(y)\\  &=\sum_{\substack{\beta',\beta'',
\gamma',\gamma''\in \N^d\\ |\beta'+\gamma'|_1=|\beta''+\gamma''|_1=1}}(P_0^\parallel P^\bot_{\beta',\gamma',0}h_{e,0}P^\bot_{\beta'',\gamma'',0} \vec{w})(x)\otimes\Big((-i\partial_y)^{\beta'+\beta''} y^{\gamma'+\gamma''}+y^{\gamma'+\gamma''} (-i\partial_y)^{\beta'+\beta''}\Big) \vec{v}_*(y)\\
&=2\sum_{\substack{\beta',\beta'',\gamma',\gamma''\in \N^d\\|\gamma'|_1=|\gamma''|_1=0\\ |\beta'|_1=|\beta''|_1=1}}(P_0^\parallel P_{\beta',\gamma',0}^\bot h_{e,0}^\bot P_{\beta'',\gamma'',0}^\bot \vec{w})(x)\otimes \Big((-i\partial_y)^{\beta'+\beta''}\vec{v}_*(y)  \Big)\\
    &\quad+\sum_{\substack{\beta',\beta'',\gamma',\gamma''\in \N^d\\|\beta'|_1=|\gamma''|_1=0\\ |\gamma'|_1=|\beta''|_1=1}}(P_0^\parallel P_{\beta',\gamma',0}^\bot h_{e,0}^\bot P_{\beta'',\gamma'',0}^\bot \vec{w})(x)\otimes \Big(y^{\gamma'} (-i\partial_y)^{\beta''}\vec{v}_*(y)+(-i\partial_y)^{\beta''}(y^{\gamma'} \vec{v}_*(y))\Big) \\
    &\quad+\sum_{\substack{\beta',\beta'',\gamma',\gamma''\in \N^d\\|\gamma'|_1=|\beta''|_1=0\\ |\beta'|_1=|\gamma''|_1=1}}(P_0^\parallel P_{\beta',\gamma',0}^\bot h_{e,0}^\bot P_{\beta'',\gamma'',0}^\bot \vec{w})(x)\otimes   \Big((-i\partial_y)^{\beta'} (y^{\gamma''}\vec{v}_*(y))+y^{\gamma''}(-i\partial_y)^{\beta'} \vec{v}_*(y)\Big)\\
      &\quad+2\sum_{\substack{\beta',\beta'',\gamma',\gamma''\in \N^d\\|\beta'|_1=|\beta''|_1=0\\ |\gamma'|_1=|\gamma''|_1=1}}(P_0^\parallel P_{\beta',\gamma',0}^\bot h_{e,0}^\bot P_{\beta'',\gamma'',0}^\bot \vec{w})(x)\otimes \Big(y^{\gamma'+\gamma''} \vec{v}_*(y)\Big).
\end{align*}
Inserting this identity into $ \mathfrak{h}^{\rm cor}U^{(0)}$, we obtain
\begin{align*}
  \MoveEqLeft  \mathfrak{h}^{\rm cor}U^{(0)}(\vec{w}\otimes \vec{v}_*)(x,y)=\vec{w}(x)\otimes (\mathcal{F}^{-1}{\rm Op}_1(f^{\rm eff}_m)\mathcal{F}\vec{v}_*)(y)\\
    &\quad+\frac{1}{2}\sum_{\substack{\beta',\beta'',\gamma',\gamma''\in \N^d\\|\beta'|_1=|\gamma''|_1=0\\ |\gamma'|_1=|\beta''|_1=1}}(P_0^\parallel P_{\beta',\gamma',0}^\bot h_{e,0}^\bot P_{\beta'',\gamma'',0}^\bot \vec{w})(x)\otimes \Big([(-i\partial_y)^{\beta''}, y^{\gamma'}]\vec{v}_*(y)\Big) \\
     &\quad+\frac{1}{2}\sum_{\substack{\beta',\beta'',\gamma',\gamma''\in \N^d\\|\gamma'|_1=|\beta''|_1=0\\ |\beta'|_1=|\gamma''|_1=1}}(P_0^\parallel P_{\beta',\gamma',0}^\bot h_{e,0}^\bot P_{\beta'',\gamma'',0}^\bot \vec{w})(x)\otimes   \Big([y^{\gamma''},(-i\partial_y)^{\beta'} ]\vec{v}_*(y)\Big).
\end{align*}
Note that
\begin{align*}
    [(-i\partial_y)^{\beta''}, y^{\gamma'}]= -i\delta_{\beta'',\gamma'}
\end{align*}
where $\delta_{\beta'',\gamma'}$ is the Kronecker delta function, i.e.,
\begin{align*}
    \delta_{\beta'',\gamma'}=\begin{cases}
        1,\qquad \beta''=\gamma'\\
        0,\qquad \beta''\neq \gamma'
    \end{cases}.
\end{align*}
Then,
\begin{align*}
    \MoveEqLeft  \mathfrak{h}^{\rm cor}U^{(0)}(\vec{w}\otimes \vec{v}_*)(x,y)=\vec{w}(x)\otimes (\mathcal{F}^{-1}{\rm Op}_1(f^{\rm eff}_m)\mathcal{F}\vec{v}_*)(y)\\
    &\quad-\frac{i}{2}\sum_{\substack{\beta,\gamma\in \N^d\\|\beta|_1=0,\;|\gamma|_1=1}}\Big(P_0^\parallel \big( P_{\beta,\gamma,0}^\bot h_{e,0}^\bot P_{\gamma,\beta,0}^\bot -P_{\gamma,\beta,0}^\bot h_{e,0}^\bot P_{\beta,\gamma,0}^\bot\big)\vec{w}\Big)(x)\otimes \vec{v}_*(y).
\end{align*}
Thus $ \mathfrak{h}^{\rm cor}U^{(0)}$ is a sum of terms of the form $ \vec{u}\otimes \vec{v}$ with $\vec{u}=P^\parallel_0 \vec{u}\in L^2_{\rm per}$ and $\vec{v}\in L^2(\R^d)$. Now,
\begin{align*}
    \Psi(\vec{u}\otimes \vec{v})(x,y)&=\vec{u}^T(x)\vec{v}(y)=\sum_{j=1}^J u_j(x)v_j(y) \\
    &=\sum_{j,k=1}^J \left<w_k,u_j\right>_{L^2_{\rm per}}w_k(x)v_j(y)=\vec{w}^T\mathcal{M}_{\vec{u}}\vec{v}
\end{align*}
where $\mathcal{M}_{\vec{u}}$ is defined by \eqref{eq:M-u} and we used the fact that $\vec{u}=P^\parallel_0 \vec{u}\in L^2_{\rm per}$ with
\begin{align*}
    P^\parallel_0=\sum_{j=1}^J \left|w_j\right>\left<w_j\right|.
\end{align*}
Therefore,
\begin{align*}
  \MoveEqLeft   \Psi(\mathfrak{h}^{\rm cor}U^{(0)}(\vec{w}\otimes \vec{v}_*))(x,y)=\vec{w}^T(x)(\mathcal{F}^{-1}{\rm Op}_1(f^{\rm eff}_m)\mathcal{F}\vec{v}_*)(y)\\
    &\qquad-\frac{i}{2}\sum_{\substack{\beta,\gamma\in \N^d\\|\beta|_1=0,\;|\gamma|_1=1}}\vec{w}^T\mathcal{M}_{\big( P_{\beta,\gamma,0}^\bot h_{e,0}^\bot P_{\gamma,\beta,0}^\bot -P_{\gamma,\beta,0}^\bot h_{e,0}^\bot P_{\beta,\gamma,0}^\bot\big)\vec{w}} \vec{v}_*(y).
\end{align*}
Using \eqref{eq:M-tilde},
\begin{align*}
 \MoveEqLeft   \sum_{\substack{\beta,\gamma\in \N^d\\|\beta|_1=0,\;|\gamma|_1=1}} (\mathcal{M}_{\big( P_{\beta,\gamma,0}^\bot h_{e,0}^\bot P_{\gamma,\beta,0}^\bot -P_{\gamma,\beta,0}^\bot h_{e,0}^\bot P_{\beta,\gamma,0}^\bot\big)\vec{w}})_{jk}\\
    &=\sum_{\substack{\beta,\gamma\in \N^d\\|\beta|_1=0,\;|\gamma|_1=1}}\left<w_j,P_{\beta,\gamma,0}^\bot h_{e,0}^\bot P_{\gamma,\beta,0}^\bot w_k\right>_{L^2_{\rm per}}-\left<w_j, P_{\gamma,\beta,0}^\bot h_{e,0}^\bot P_{\beta,\gamma,0}^\bot w_k\right>_{L^2_{\rm per}}\\
    &=2i \sum_{\substack{\beta,\gamma\in \N^d\\|\beta|_1=0,\;|\gamma|_1=1}}\left<w_j,{\rm Im} (P_{\beta,\gamma,0}^\bot h_{e,0}^\bot P_{\gamma,\beta,0}^\bot) w_k\right>_{L^2_{\rm per}}= 2i (\widetilde{\mathcal{M}})_{jk}.
\end{align*}
Thus under Assumption \ref{ass:m} for the eigenpair $(\vec{v}_*,\mu_*)$,
\begin{align*}
      \Psi(\mathfrak{h}^{\rm cor}U^{(0)}(\vec{w}\otimes \vec{v}_*))=\vec{w}^T\Big(\mathcal{F}^{-1}{\rm Op}_1(f_2^{\rm eff})\mathcal{F} +\widetilde{\mathcal{M}} \Big)\vec{v}_* =\mu_* \vec{w}^T\vec{v}_*=\mu_* \Psi(U^{(0)}(\vec{w}\otimes \vec{v}_*))
\end{align*}
where $f_2^{\rm eff}$ is defined in Assumption \ref{ass:m}. This ends the proof of Lemma \ref{lem:h-cor}.

\section{Proof of Theorem \ref{th:quadratic}}\label{sec:proofofquadratic}
In this section, we prove Theorem \ref{th:quadratic}. That is, we will show that
\begin{align*}
   \| (H_\eps  -e_0)\Phi_\eps(\vec u\otimes\vec v_\eps)(x) -\chi(\eps^{s_1} x-\eps^{s_1-1}X_0) \mathfrak{h}_\eps^{\rm eff}(\vec u\otimes\vec v_\eps)(x,x) \|_{L^2(\R^d;\C^n)}=\cO( \eps^{\frac{1}{2}m+\frac{1}{4}}).
\end{align*}
Let $k_0$ be given as in Assumption \ref{ass:bandstructure-nondegenerate}. Without loss of generality, it suffices to assume $k_0=0$. The case $k_0\neq 0$ can be reduced to the case $k_0=0$ by using the following lemma.
\begin{lemma}
If Theorem \ref{th:quadratic} holds for $k_0=0$, it also holds for any $k_0\in \Omega^*$.
\end{lemma}
\begin{proof}
This is indeed a consequence of Gauge invariance. Let $k_0\neq 0$, and let $\vec v_\eps$ be defined as in Theorem \ref{th:quadratic}. We define 
\begin{align*}
   \vec{\widetilde{v}}_\eps(x)=e^{-ik_0\cdot x} \vec{v}_\eps(x).
\end{align*}
and
\begin{align*}
  \widetilde{H}_\eps:=  T(-i\nabla_x +k_0+\bA(x,\eps x))+V(x,\eps x)=\cU^{-1}\widetilde{h}(k,i\eps \nabla_k)\cU
\end{align*}
with $\widetilde{h}(k,X):=h(k+k_0,X)$. 

Note that $(E_j(k+k_0,X),\phi_j(k+k_0,X))_{j\geq 1}$ are the eigenpairs of the operator  $\widetilde{h}(k,X)$. Then $\widetilde{h}$ satisfies Assumption \ref{ass:bandstructure-nondegenerate} around $(0,X_0)$ and $e_0$. Moreover, it is easy to check  that $\vec{\widetilde{v}}_\eps$ verifies Assumption \ref{ass:localization}. Let $\widetilde{\Phi}_{\eps}$ be constructed as in Theorem \ref{th:quadratic} for $k_0=0$. Then Theorem \ref{th:quadratic} gives
\begin{align*}
   \| (\widetilde{H}_\eps  -e_0)\Phi_\eps(\vec u\otimes \vec{\widetilde{v}}_\eps) -\Phi_\eps(\widetilde{\mathfrak{h}}_\eps^{\rm eff}(m)(\vec u\otimes \vec{\widetilde{v}}_\eps)) \|_{L^2(\R^d;\C^n)}=\cO( \eps^{\frac{m}{2}+\frac{1}{4}})
\end{align*}
where
\begin{align*}
     \widetilde{\mathfrak{h}}_{\eps}^{\rm eff}(\vec u\otimes \vec v):=  \widetilde{\mathfrak{h}}_{1,\eps}^{\rm eff}(\vec u\otimes \vec v)+  \widetilde{\mathfrak{h}}_{2,\eps}^{\rm eff}(\vec u\otimes \vec v)
\end{align*}
with
\begin{align*}
      \widetilde{\mathfrak{h}}_{1,\eps}^{\rm eff}(\vec u\otimes \vec v)(x,y):= \vec{w}(x)\otimes \Big(\mathcal{F}^{-1}\Op(\widetilde{h}^{\rm eff}) \mathcal{M}_{\vec u} \mathcal{F} \vec{v}\Big)(y)
\end{align*}
\begin{align*}
     \widetilde{\mathfrak{h}}_{2,\eps}^{\rm eff}(\vec u\otimes \vec v)(x,y):=\sum_{\substack{\beta,\gamma\in \N^d\\ |\beta+\gamma|_1\leq m}}\!\!\!\! \frac{1}{\beta!\gamma!}\left(h^\bot_{e,\beta,\gamma,0} \vec u\right)(x) \otimes \Big(\mathcal{F}^{-1}\Op(k^\beta(X-X_0)^\gamma ) \mathcal{F}\vec v \Big)(y)
\end{align*}
and
\begin{align*}
    \widetilde{h}^{\rm eff}(k,X):=f_m^{\rm eff}(k,X-X_0).
\end{align*}
Here $\widetilde{\mathfrak{h}}_{1,\eps}^{\rm eff}$ and $\widetilde{\mathfrak{h}}_{2,\eps}^{\rm eff}$ are exactly the same as $\mathfrak{h}_{1,\eps}^{\rm eff}$ and $\mathfrak{h}_{2,\eps}^{\rm eff}$ except we replace $k-k_0$ by $k$ in the Weyl quantization $\Op(\bullet)$. In above definition, we also used the fact that
\begin{align*}
  \partial_k^\beta \partial_X^\gamma \widetilde{h}^\bot_e(0,X_0)=h^\bot_{e,\beta,\gamma,0}
\end{align*}
where $\widetilde{h}^\bot_e(k,X)$ is defined analogously to $h^\bot_e(k,X)$ by replacing the operator $h(k,X)$ with $\widetilde{h}(k,X)$ in the definition.

By direct calculation,
\begin{align*}
    \mathcal{F}^{-1}\Op(\widetilde{h}^{\rm eff}) \mathcal{M}_{\vec u} \mathcal{F} \vec{\widetilde{v}}_\eps=e^{-ik_0\cdot x}\mathcal{F}^{-1}\Op(h^{\rm eff}) \mathcal{M}_{\vec u} \mathcal{F} \vec{v}_\eps
\end{align*}
and
\begin{align*}
    \mathcal{F}^{-1}\Op(k^\beta(X-X_0)^\gamma ) \mathcal{F}\vec{\widetilde{v}}_\eps=e^{-ik_0\cdot x}\mathcal{F}^{-1}\Op((k-k_0)^\beta(X-X_0)^\gamma ) \mathcal{F}\vec{v}_\eps.
\end{align*}
Moreover, by Gauge invariance, 
\begin{align*}
     (\widetilde{H}_\eps-e_0) \Phi_\eps(\vec u\otimes \vec {\widetilde{v}}_\eps)=e^{-ik_0\cdot x}(H_\eps  -e_0)\Phi_\eps(\vec u\otimes \vec v_\eps).
\end{align*}
These equations imply \eqref{eq:Hepsilon-Phi-epsilon}, hence proving Theorem \ref{th:quadratic}.
\end{proof}

It now suffices to prove Theorem \ref{th:quadratic} with $k_0=0$. Using the Bloch transform and \eqref{op-h}, this is equivalent to showing
\begin{align*}
   \MoveEqLeft  \left\| \Op(h_e) \Big(\cU\Phi_\eps(\vec u\otimes\vec v_\eps)\Big)_{\bullet} -\Big(\cU \Phi_\eps\big( \mathfrak{h}_\eps^{\rm eff}(\vec u\otimes\vec v_\eps)\big)\Big)_\bullet\right\|_\cH\\
   &= \left(\fint_{\Omega^*}\left\|\Op(h_e) \Big(\cU\Phi_\eps(\vec u\otimes\vec v_\eps)\Big)_k -\Big(\cU \Phi_\eps\big( \mathfrak{h}_\eps^{\rm eff}(\vec u\otimes\vec v_\eps)\big)\Big)_k\right\|_{L^2_{\rm per}}^2dk\right)^{1/2}\\
   &=\cO( \eps^{\frac{1}{2}m+\frac{1}{4}}).
\end{align*}

We now outline the main ideas of the proof, and the relevant lemmas are proved in Subsections \ref{sec:3.2}-\ref{sec:3.6} below.
\begin{proof}[Sketch of proof of Theorem \ref{th:quadratic} with $k_0=0$]
We assume henceforth $k_0=0$. For simplicity, set
\begin{align*}
    \Phi_\eps(x):=\Phi_\eps(\vec{u}\otimes \vec{v}_\eps)(x).
\end{align*}

The proof is divided into the following steps.

\medskip

\noindent{\bf Step 1. Reduction of the problem.} 

Recall that $s_1= \frac{1}{2}+ \frac{1}{2\mathfrak{n} d\left(m+1\right)}$ is defined in Assumption \ref{ass:bandstructure-nondegenerate}. Let
\begin{align*}
    \chi_\eps(X):=\chi(\eps^{s_1-1}(X-X_0)).
\end{align*}
Analogously to \eqref{op-h}, we have
\begin{align*}
    (\cU\Phi_\eps)=\Op(\chi_\eps)(\cU u_\eps).
\end{align*}
with
\begin{align*}
    u_\eps:=\vec{u}^T\vec{v}_\eps.
\end{align*}
Thus,
\begin{align*}
    \Op(h_e)(\cU\Phi_\eps)=  \Op(h_e)\Op(\chi_\eps)(\cU u_\eps).
\end{align*}
Our first step is to replace operator $\Op(h_e)\Op(\chi_\eps)$ by $\Op(h_e \chi_\eps)$ in the following sense.
\begin{lemma}\label{lem:3.2}
With the same assumptions as in Theorem \ref{th:quadratic}, for any $\eps<1$,
    \begin{align*}
        \|\Op(h_e)(\cU\Phi_\eps)-\Op(h_e \chi_\eps)(\cU u_\eps)\|_{\cH}=\cO(\eps^{\frac{1}{2}(m+1)}).
    \end{align*}
\end{lemma}

\medskip

\noindent{\bf Step 2. Reduction of $\Op(h_e \chi_\eps)(\cU u_\eps)$.} 

Using Lemma \ref{lem:3.2}, we now focus on the study of $\Op(h_e\chi_\eps)\cU(u_\eps)$. By the definition of Weyl quantization,
\begin{align*}
    \Big(\Op(h_e\chi_\eps)\cU(u_\eps)\Big)(k)=\frac{1}{(2\pi \eps)^d}\int_{\R^d\times \R^d}  e^{-i\frac{(k-k')\cdot X}{\eps}}\chi_\eps(X)h_e\left(\frac{k+k'}{2},X\right)(\cU u_\eps)_{k'}\;dk'dX.
\end{align*}
We split the integration w.r.t. $k'$ into two parts: close to $k$ and away from $k$. To do so, recall
\begin{align}
    s_2=\frac{1}{2}-\frac{1}{2\mathfrak{n} d(m+1)}
\end{align}
and let
\begin{align}
    \widetilde{\chi}_\eps(k)=\chi(\eps^{-s_2} k).
\end{align}
Then 
\begin{align}\label{eq:hchi-integration}
    \MoveEqLeft \Big(\Op(h_e\chi_\eps)\cU(u_\eps)\Big)(k)\notag\\
    &=\frac{1}{(2\pi \eps)^d}\int_{\R^d\times \R^d}  e^{-i\frac{(k-k')\cdot X}{\eps}}\chi_\eps(X)\widetilde{\chi}_\eps(k-k') h_e\left(\frac{k+k'}{2},X\right)(\cU u_\eps)_{k'}\;dk'dX\notag\\
    &\quad +\frac{1}{(2\pi \eps)^d}\int_{\R^d\times \R^d} e^{-i\frac{(k-k')\cdot X}{\eps}}\chi_\eps(X)\Big(1-\widetilde{\chi}_\eps(k-k') \Big)h_e\left(\frac{k+k'}{2},X\right)(\cU u_\eps)_{k'}\;dk'dX\notag\\
    &=:I(k)+II(k).
\end{align}
Our next lemma shows that
\begin{lemma}\label{lem:3.3}
 For $\eps<1$, we have
    \begin{align*}
      \|II\|_{\cH}=\cO(\eps^{\frac12(m+1)}).
    \end{align*}
\end{lemma}

\medskip

\noindent{\bf Step 3. Reduction of $I(k)$.}

We recall the Poisson summation formula for $u_\eps$:
\begin{align}\label{poisson}
    \left(\cU u_\eps \right)_k&=\vec{u}^T(x)\sum_{R\in \mathbb{L}}\vec{v}_\eps(x+R)e^{-ik\cdot(x+R)}\notag \\
    &=|\Omega|^{-1}\vec{u}^T(x) \sum_{G\in \mathbb{L}^*}e^{i G\cdot x} \mathcal{F}(\vec{v}_\eps)\left(G+k\right)
\end{align}
According to Assumption \ref{ass:localization}, $\mathcal{F}(\vec{v}_\eps)$ decays rapidly. Thus, we split $\|I\|_\cH$ into the following two parts for $\eps$ small enough:
\begin{align*} \|I\|_{\cH}^2&=\fint_{\Omega^*}\|I(k)\|_{L^2_{\rm per}}^2dk\\
    &=\frac{1}{|\Omega^*|}\int_{\substack{k\in \Omega^*\\{\rm dist}(k,\mathbb L^*)< 4\eps^{s_2}}}\|I(k)\|_{L^2_{\rm per}}^2dk+ \frac{1}{|\Omega^*|}\int_{\substack{k\in \Omega^*\\ {\rm dist}(k,\mathbb L^*)\geq  4\eps^{s_2}}}\|I(k)\|_{L^2_{\rm per}}^2dk\\
    &=\frac{1}{|\Omega^*|}\int_{|k|< 4\eps^{s_2}}\|I(k)\|_{L^2_{\rm per}}^2dk+ \frac{1}{|\Omega^*|}\int_{\substack{k\in \Omega^*\\ {\rm dist}(k,\mathbb L^*)\geq  4\eps^{s_2}}}\|I(k)\|_{L^2_{\rm per}}^2dk
\end{align*}
where in the last equality, by the definition of $\Omega^*$, i.e., \eqref{eq:Omega*}, for any $k\in \Omega^*$,
\begin{align}
    |k|={\rm dist}(k,\mathbb L^*).
\end{align}
The following is a consequence of the decay of $\mathcal{F}(\vec{v}_\eps)$.
\begin{lemma}\label{lem:3.4}
For any $0<\eps<1$, we have
    \begin{align*}
       \left( \int_{\substack{k\in \Omega^*\\ {\rm dist}(k,\mathbb L^*)\geq  4\eps^{s_2}}}\|I(k)\|_{L^2_{\rm per}}^2dk\right)^{1/2} =\cO(\eps^{\frac{m+1}{2}}).
    \end{align*}
\end{lemma}

\medskip

\noindent{\bf Step 4. From $h(k,X)$ to symbols of effective Hamiltonian}

It remains to study $I(k)$ with $|k|={\rm dist}(k,\mathbb L^*)< 4\eps^{s_2}$. For $k'\in \Supp(\widetilde{\chi}_\eps(k-\cdot))$, we have $|k'|\leq 6\eps^{s_2}$. Then using the Poisson summation \eqref{poisson} and the decay of $\mathcal{F}(\vec{v}_\eps)$ again, 
\begin{align*}
     \left(\cU u_\eps \right)_{k'}\approx |\Omega|^{-1}  \vec{u}^T(x)\mathcal{F}(\vec{v}_\eps)\left(k'\right).
\end{align*}
Next, we decompose $h_e(k,X)$ as
\begin{align}\label{eq:h-decomp}
    h_e(k,X)&=\sum_{1\leq j\leq J}\widetilde{E}_{j}(k,X)\left|\phi_{j}\right>\left<\phi_{j}\right|(k,X)+h_{e}^\bot(k,X)
\end{align}
with
\begin{align*}
\widetilde{E}_{j}\left(k,X\right):=E_{j}\left(k,X\right) -e_0.
\end{align*}
By Assumption \ref{ass:bandstructure-nondegenerate} and \eqref{eq:ass1-3}, around $(k_0,X_0)$,
\begin{align*} 
\MoveEqLeft \sum_{1\leq j\leq J}\widetilde{E}_{j}\left|\phi_{j}\right>\left<\phi_{j}\right|(k,X)\approx \sum_{1\leq j\leq J}\Big|\vec{w}^T \vec{\alpha}_j(k,X) \Big>\Big< \vec{w}^T \Big(h^{\rm eff}(k,X) \vec{\alpha}_j(k,X)\Big)\Big|.
\end{align*}
For the part of the operator $h(k,X)$ around $e_0$, we construct the following integration:
\begin{align}\label{eq:I1-eff}
    I^{\rm eff}_1(k):&=\frac{1}{(2 \pi\eps)^d|\Omega|}\int_{\R^d\times \R^d}  e^{-i\frac{(k-k')\cdot X}{\eps}} \chi_\eps(X)\widetilde{\chi}_\eps(k-k')\notag\\
   &\qquad\qquad\qquad\times \vec{w}^T(x) h^{\rm eff}\left(\tfrac{k+k'}{2},X\right) \mathcal{M}_{\vec{u}}\mathcal{F}(\vec{v}_\eps)(k')  \;dk'dX.
\end{align}
For the part $h_e^\bot(k,X)$, using Taylor's expansion \eqref{eq:taylor expansion}, we construct the following integration:
\begin{align}\label{eq:I2-eff}
     I^{\rm eff}_2(k):&=\frac{1}{(2\pi \eps)^d|\Omega|}\int_{\R^d\times \R^d}  e^{-i\frac{(k-k')\cdot X}{\eps}}\chi_\eps(X)\widetilde{\chi}_\eps(k-k') \notag\\
   & \times \sum_{\substack{\beta,\gamma\in \N^d\\ |\beta+\gamma|_1\leq m}}\frac{1}{\beta!\gamma!} (h^\bot_{e,\beta,\gamma,0} \vec{u})^T \left[ \left(\tfrac{k+k'}{2}\right)^{\beta} (X-X_0)^\gamma \mathcal{F}(\vec{v}_{\eps})(k') \right] dk'dX.
\end{align}
Then we will show the following.
\begin{lemma}\label{lem:3.5}
 For any $\eps<1$, we have
    \begin{align*}
       \left( \int_{|k|< 4\eps^{s_2}}\|I(k)- I_1^{\rm eff}(k)-I_2^{\rm eff}(k)\|_{L^2_{\rm per}}^2dk\right)^{1/2}=\cO(\eps^{\frac{1}{2}m+\frac{1}{4}}).
    \end{align*}
\end{lemma}

\medskip

\noindent{\bf Step 5. Effective Hamiltonian.} We now focus on $I_1^{\rm eff}(k)$ and $I_2^{\rm eff}(k)$. Arguing as for Step 3. Lemma \ref{lem:3.4}, we obtain the following.
\begin{lemma}\label{lem:3.6'}
For any $0<\eps<1$, we have
\begin{align*}
\left( \int_{\substack{k\in \Omega^*\\ {\rm dist}(k,\mathbb L^*)\geq  4\eps^{s_2}}}\|I_1^{\rm eff}(k)+I_2^{\rm eff}(k)\|_{L^2_{\rm per}}^2dk\right)^{1/2} =\cO(\eps^{\frac{1}{2}m+\frac{1}{4}}).
\end{align*}
\end{lemma}
Next arguing as in Step 1. and Step 2., we have the approximations
\begin{align*}
    I^{\rm eff}_1(k)\approx   \widetilde{\widetilde{\mathfrak{h}}}^{\rm eff}_{1,\eps}(\vec{u},\vec{v}_\eps)(k),\qquad  I^{\rm eff}_2(k)\approx   \widetilde{\widetilde{\mathfrak{h}}}^{\rm eff}_{2,\eps}(\vec{u},\vec{v}_\eps)(k)
\end{align*}
with
\begin{align*}
    \widetilde{\widetilde{\mathfrak{h}}}^{\rm eff}_{1,\eps}(\vec{u},\vec{v}_\eps)(k):= |\Omega|^{-1}\vec{w}^T(x)\Op(\chi_\eps)\Op(h^{\rm eff})\mathcal{M}_{\vec{u}}\mathcal{F}(\vec{v}_\eps)(k)
\end{align*}
and
\begin{align*}
  \MoveEqLeft \widetilde{\widetilde{\mathfrak{h}}}^{\rm eff}_{2,\eps}(\vec{u},\vec{v}_\eps)(k):=|\Omega|^{-1}\!\!\!\!\!\!\!\!\sum_{\substack{\beta,\gamma\in \N^d\\ |\beta+\gamma|_1\leq m}}\!\!\!\!\frac{1}{\beta!\gamma!}(h^\bot_{e,\beta,\gamma,0} \vec{u})^T(x)\Op(\chi_\eps)\Op((k-k_0)^\beta (X-X_0)^{\gamma})\mathcal{F}(\vec v_{\eps})(k).
\end{align*}
Then,
\begin{lemma}\label{lem:3.6}
For any $0<\eps<1$ and $j\in \{1,2\}$, we have
    \begin{align*}
          \left( \fint_{\Omega^*}\|  I^{\rm eff}_j(k)-\widetilde{\widetilde{\mathfrak{h}}}^{\rm eff}_{j,\eps}(\vec{u},\vec{v}_\eps) \|_{L^2_{\rm per}}^2dk\right)^{1/2}=\cO(\eps^{\frac{1}{2}(m+1)}).
    \end{align*}
\end{lemma}

\medskip

\noindent{\bf Step 6. Conclusion.} From Lemmas \ref{lem:3.2}-\ref{lem:3.6}, for $\eps$ small enough,
\begin{align}\label{eq:6.9}
   \left\|\Op(h_e)(\cU\Phi_\eps)_{\bullet}-\widetilde{\widetilde{\mathfrak{h}}}^{\rm eff}_{1,\eps}(\vec{u},\vec{v}_\eps)-\widetilde{\widetilde{\mathfrak{h}}}^{\rm eff}_{2,\eps}(\vec{u},\vec{v}_\eps)\right\|_\cH=\cO(\eps^{\frac{1}{2}m+\frac{1}{4}}).
\end{align}
Now,
\begin{align*}
  \MoveEqLeft  |\Omega|\mathcal{F}^{-1}_k\big(\widetilde{\widetilde{\mathfrak{h}}}^{\rm eff}_{1,\eps}(\vec{u},\vec{v}_\eps)\big)(x)\\
  &=\vec{w}^T\mathcal{F}^{-1}\Big(\Op(\chi_\eps)\Op(h^{\rm eff})\mathcal{M}_{\vec{u}}\mathcal{F}(\vec{v}_\eps)\Big)(x)\\
  &=\chi_\eps(\eps x) \vec{w}^T\mathcal{F}^{-1}\Big(\Op(h^{\rm eff})\mathcal{M}_{\vec{u}}\mathcal{F}(\vec{v}_\eps)\Big)(x)
\end{align*}
where in the last equation, we used the fact that
\begin{align*}
   \Op(\chi_\eps)(i\eps \nabla_k)= \chi_\eps(i\eps\nabla_k).
\end{align*}
Thus
\begin{align*}
|\Omega|\mathcal{F}^{-1}_k\big(\widetilde{\widetilde{\mathfrak{h}}}^{\rm eff}_{1,\eps}(\vec{u},\vec{v}_\eps)\big)(x)=\Phi_\eps(\mathfrak{h}_{1,\eps}^{\rm eff}(\vec{u}\otimes\vec{v}_\eps))(x).
\end{align*}
Analogously,
\begin{align*}
|\Omega|\mathcal{F}^{-1}_k\big(\widetilde{\widetilde{\mathfrak{h}}}^{\rm eff}_{2,\eps}(\vec{u},\vec{v}_\eps)\big)(x)=\Phi_\eps(\mathfrak{h}_{2,\eps}^{\rm eff}(\vec{u}\otimes\vec{v}_\eps))(x).
\end{align*}
Thus under Assumption \ref{ass:localization}, using Lemma \ref{lem:Uf-cH} again, for $j=1,2$
\begin{align*}
  \left\||\Omega|\cU\Big( \mathcal{F}^{-1}_k\big(\widetilde{\widetilde{\mathfrak{h}}}^{\rm eff}_{j,\eps}(\vec{u},\vec{v}_\eps)\big) \Big)_\bullet -\widetilde{\widetilde{\mathfrak{h}}}^{\rm eff}_{j,\eps}(\vec{u},\vec{v}_\eps)\right\|_\cH=\cO(\eps^{\frac12(m+1)}).
\end{align*}
Combining this with \eqref{eq:6.9},
\begin{align*}
       \left\|\Op(h_e)(\cU\Phi_\eps(\vec{u}\otimes\vec{v}_\eps))_{\bullet}-\cU\Big(\Phi_\eps(\mathfrak{h}_{\eps}^{\rm eff}(\vec{u}\otimes\vec{v}_\eps))\Big)_\bullet\right\|_\cH=\cO(\eps^{\frac12m+\frac{1}{4}}).
\end{align*}
As a result,
\begin{align*}
       \left\|(H_\eps-e_0)\Phi_\eps(\vec{u}\otimes\vec{v}_\eps)-\Phi_\eps(\mathfrak{h}_{\eps}^{\rm eff}(\vec{u}\otimes\vec{v}_\eps))\right\|_{L^2(\R^d;\C^n)}=\cO(\eps^{\frac12m+\frac{1}{4}}).
\end{align*}
This proves the theorem.
\end{proof}

We now prove Lemmas \ref{lem:3.2}-\ref{lem:3.6}.

\subsection{Proof of Lemma \ref{lem:3.2}}\label{sec:3.2}

Note that $k\mapsto h(k,\bullet)$ is a polynomial of degree at most $2$ and $X\mapsto \chi_\eps(X)$ is independent of $k$. According to Moyal product for Weyl quantization, 
\begin{align*}
    \MoveEqLeft \Op(h_e)\Op(\chi_\eps)=\Op(h_e\chi_\eps)\\
     &-\frac{i \eps}{2}\sum_{1\leq j\leq d}\Op\big(\partial_{k_j} h_e \cdot \partial_{X_j} \chi_\eps\big)-\frac{\eps^2}{8}\sum_{1\leq j,\ell\leq d}\Op(\partial_{k_j}\partial_{k_\ell} h_e\partial_{X_j}\partial_{X_\ell}\chi_\eps)
\end{align*}
where in the last term, we use the fact that $(\partial_{k_j}\partial_{k_\ell} h_e)$ is a constant independent of $k$ and $X$. Thus,
\begin{align}\label{eq:hphi-hchiu}
   \MoveEqLeft  \|\Op(h_e)(\cU\Phi_\eps)-\Op(h_e \chi_\eps)(\cU u_\eps)\|_{\cH}\notag\\
   &\lesssim  \eps\sum_{1\leq j\leq d}\|\Op\big(\partial_{k_j} h_e \cdot \partial_{X_j} \chi_\eps\big)(\cU u_\eps)\|_{\cH}+\eps^2\sum_{1\leq j,\ell\leq d}\|\Op (\partial_{X_j}\partial_{X_\ell}\chi_\eps)(\cU u_\eps)\|_{\cH}.
\end{align}

We first study the last term of \eqref{eq:hphi-hchiu}. Note that
\begin{align*}
    \eps^2\Op (\partial_{X_j}\partial_{X_\ell}\chi_\eps)(\cU u_\eps)=\eps^{2s_1}\left(\cU\Big((\partial_{j}\partial_\ell \chi)(\eps^{s_1} x-\eps^{s_1-1}X_0)\vec{u}^T\vec{v}_\eps  \Big)\right)
\end{align*}
By Assumption \ref{ass:localization} and Lemma \ref{lem:Uf-cH} with $n=d+1$,
\begin{align*}
\MoveEqLeft\eps^{2s_1}\Big\|\cU\Big((\partial_{j}\partial_\ell \chi)(\eps^{s_1} x-\eps^{s_1-1}X_0)\vec{u}^T\vec{v}_\eps \Big)\Big\|_{\cH}\\
   & \lesssim \eps^{2s_1}\int_{\R^d} \left|(\partial_{j}\partial_\ell \chi)(\eps^{s_1} x-\eps^{s_1-1}X_0)v_\eps\right|dx\\
   &\quad+ \eps^{2s_1}\int_{\R^d} \left|(-\Delta)^{(d+1)/2}[\partial_{j}\partial_\ell \chi)(\eps^{s_1} x-\eps^{s_1-1}X_0)\vec{v}_\eps]\right|dx\\
   &\lesssim \eps^{2s_1}\|\vec{v}_\eps\|_{W^{d+1,1}(\R^d\setminus B_{\eps^{-s_1}}(\eps^{-1}X_0))}=\cO(\eps^{\frac12(m+1)})
\end{align*}
where we used the fact that $\Supp((\partial_{j}\partial_\ell \chi)(\eps^{s_1}\bullet-\eps^{s_1-1}X_0))\cap B_{\eps^{-s_1}}(\eps^{-1}X_0)=\emptyset$. Thus,
\begin{align}\label{eq:partial2chi u}
     \MoveEqLeft  \eps^2\|\Op (\partial_{X_j}\partial_{X_\ell}\chi_\eps)(\cU u_\eps)\|_\cH=\cO(\eps^{\frac12(m+1)}).
\end{align}

\medskip

Now consider the first term on the right-hand side of \eqref{eq:hphi-hchiu}. By Moyal product,
\begin{align*}
    \Op(\partial_{k_j}h_e \partial_{X_j}\chi_\eps)=\Op(\partial_{k_j}h_e)\Op( \partial_{X_j}\chi_\eps)+\frac{i\eps}{2}\Op(\nabla_k\partial_{k_j}h_e \cdot \nabla_{X}\partial_{X_j}\chi_\eps)
\end{align*}
Arguing as for \eqref{eq:partial2chi u}, we infer
\begin{align*}
     \MoveEqLeft  \eps^2\|\Op(\nabla_k\partial_{k_j}h_e \cdot \nabla_{X}\partial_{X_j}\chi_\eps)(\cU u_\eps)\|_\cH=\cO(\eps^{\frac{1}{2}(m+1)}).
\end{align*}
For the term $\Op(\partial_{k_j}h_e)\Op( \partial_{X_j}\chi_\eps)$,
\begin{align*}
 \MoveEqLeft   \eps \Op(\partial_{k_j}h_e)\Op( \partial_{X_j}\chi_\eps)(\cU u_\eps)\\
 &=\eps^{s_1}\cU\Big((\partial_{j}T)(-i\nabla_x+\bA(x,\eps x))\big[ (\partial_{j}\chi)(\eps^{s_1}x -\eps^{s_1-1}X_0)\vec{u}^T \vec{v}_\eps \big]\Big).
\end{align*}
Then,
\begin{align*}
  \MoveEqLeft \left\|\cU\Big((\partial_{j}T)(-i\nabla_x+\bA(x,\eps x))\big[ (\partial_{j}\chi)(\eps^{s_1}x -\eps^{s_1-1}X_0) \vec{u}^T\vec{v}_\eps \big]\Big)\right\|_{\cH}\\
   &=\Big\|(\partial_{j}T)(-i\nabla_x+\bA(x,\eps x))\big[ (\partial_{j}\chi)(\eps^{s_1}x -\eps^{s_1-1}X_0) \vec{u}^T\vec{v}_\eps \big]\Big\|_{L^2(\R^d)}\\
   &\leq \Big\|(1-\Delta_x)^{1/2}\big[ (\partial_{j}\chi)(\eps^{s_1}x -\eps^{s_1-1}X_0) \vec{u}^T\vec{v}_\eps \big]\Big\|_{L^2(\R^d)}\\
   &\quad+\|\bA(x,\eps x)(\partial_{j}\chi)(\eps^{s_1}x -\eps^{s_1-1}X_0) \vec{u}^T\vec{v}_\eps\|_{L^2(\R^d)}.
\end{align*}
Using Lemma \ref{lem:Uf-cH} again and arguing as above  with $\vec{u}\in H^2_{\rm per}$,
\begin{align*}
  \MoveEqLeft  \eps^{s_1}\left\|\cU\Big((\partial_{\xi_j}T)(-i\nabla_x+\bA(x,\eps x))\big[ (\partial_{j}\chi)(\eps^{s_1}x -\eps^{s_1-1}X_0) \vec{u}^T\vec{v}_\eps \big]\Big)\right\|_{\cH}\\
  &\lesssim \eps^{s_1}\|\vec{v}_\eps\|_{W^{d+2,1}(\R^d\setminus B_{\eps^{-s_1}}(\eps^{-1}X_0))}\\
  &\quad+\eps^{s_1}\sup_{y\in \R^d}\|\bA(y,\eps\bullet)\vec{v}_\eps\|_{W^{d+1,1}\big(B_{2\eps^{-s_1}}(\eps^{-1}X_0)\setminus B_{\eps^{-s_1}}(\eps^{-1}X_0)\big)}\\
  &\lesssim \eps^{s_1}\|\vec{v}_\eps\|_{W^{d+2,1}(\R^d\setminus B_{\eps^{-s_1}}(\eps^{-1}X_0))}=\cO(\eps^{\frac{1}{2}(m+1)}).
\end{align*}
where we used
\begin{align*}
\|\bA(x,X)\|_{W^{d+2,\infty}\big(\R^{d}\times B_{2\eps^{1-s_1}}(X_0)\big)}<\infty
\end{align*}
since $\Supp(\chi)\subset B_{2}(0)$, $\bA(x,X)\in C^\infty(\R^d\times \R^d)$ and $x\mapsto \bA(x,\cdot)$ is $\mathbb L$-periodic. Thus,
\begin{align*}
   \eps \|  \Op(\partial_{k_j}h_e \partial_{X_j}\chi_\eps)\cU(u_\eps)\|_{\cH}=\cO(\eps^{\frac{1}{2}(m+1)}).
\end{align*}
This and \eqref{eq:partial2chi u} show that\begin{align*}
        \|\Op(h_e)\cU(\Phi_\eps)-\Op(h_e \chi_\eps)\cU(u_\eps)\|_{\cH}=\cO(\eps^{\frac{1}{2}(m+1)}).
    \end{align*}
Thus the proof of Lemma \ref{lem:3.2} is completed.

\subsection{Proof of Lemma \ref{lem:3.3}}\label{sec:3.3}
As $(-\Delta_X)e^{-i\frac{(k-k')}{\eps}\cdot X} = \frac{|k-k'|^2}{\eps^2} e^{-i\frac{(k-k')}{\eps}\cdot X} $, by integration by parts, for $M\in \N$,
\begin{align*}
 II(k)&=\frac{1}{(2\pi \eps)^d}\int_{\R^d\times \R^d} \frac{\eps^{2M}}{|k-k'|^{2M}} e^{-i\frac{(k-k')\cdot X}{\eps}}\Big(1-\widetilde{\chi}_\eps(k-k')\Big) \\
     &\qquad\qquad\times (-\Delta_X)^M\left[\chi_\eps(X)h_e\left(\frac{k+k'}{2},X\right)(\cU u_\eps)_{k'}\right] dk'dX.
\end{align*}
For any $f\in \cS(\R^d)$,
\begin{align*}
    \left(-i\nabla_x+\frac{k+k'}{2} \right)(\cU f)_{k'}= (\cU(-i\nabla_x f))_{k'} + \frac{k-k'}{2} (\cU f)_{k'}.
\end{align*}
As $\bA(x,X), V(x,X)$ are $\mathbb L$-periodic w.r.t. $x$ and are smooth w.r.t. $x$ and $X$,
\begin{align*}
 \MoveEqLeft  \left\|(-\Delta_X)^M\left[\chi_\eps(X)h_e\left(\frac{k+k'}{2},X\right)(\cU u_\eps)_{k'}\right]\right\|_{L^2_{\rm per}}\\
   &\lesssim \eps^{2(s_1-1)M}(1+|k-k'|^2)\left(\|(\cU u_\eps)_{k'}\|_{L^2_{\rm per}}+\|(\cU (-\Delta u_\eps))_{k'}\|_{L^2_{\rm per}}\right)\1_{B_{2\eps^{1-s_1}}(X_0)}(X)
\end{align*}
where we used the fact that
\begin{align*}
    \Supp(\chi_\eps)\subset B_{2\eps^{1-s_1}}(X_0)
\end{align*}
and for any $\beta\in \N^d$,
\begin{align*}
    |\partial_X^\beta \chi_\eps|\lesssim \eps^{(s_1-1)|\beta|_1}.
\end{align*}
Thus, for any $k\in \Omega^*$ and $M\geq \frac{d}{2}$, by Lemma \ref{lem:Uf-cH} with $n=d+1$,
\begin{align*}
  \MoveEqLeft  \|II(k)\|_{L^2_{\rm per}}\\
  &\lesssim \eps^{2Ms_1-d}\int_{\substack{|k-k'|\geq \eps^{s_2}\\ |X-X_0|\leq 2\eps^{1-s_1}}} \frac{1+|k-k'|^2}{|k-k'|^{2M}}  \left(\|(\cU u_\eps)_{k'}\|_{L^2_{\rm per}}+\|(\cU (-\Delta u_\eps))_{k'}\|_{L^2_{\rm per}}\right)\;dk' dX\\
    &\lesssim \eps^{(2M-d)(s_1-s_2)} \sup_{k'\in \Omega^*}\left(\|(\cU u_\eps)_{k'}\|_{L^2_{\rm per}}+\|(\cU (-\Delta u_\eps))_{k'}\|_{L^2_{\rm per}}\right)\\
    &\lesssim \eps^{(2M-d)(s_1-s_2)}\|\vec{v}_\eps\|_{W^{d+3,1}(\R^d)}\lesssim \eps^{(2M-d)(s_1-s_2)-\frac{d}{4}}
\end{align*}
where in the last estimate we used Assumption \ref{ass:localization}. Note that
\begin{align}\label{s_1-s_2}
    s_1-s_2=\frac{1}{\mathfrak{n} d(m+1)}>0.
\end{align}
Choosing $M$ large enough such that $(2M-d)(s_1-s_2)-\frac{d }{4}\geq \frac{1}{2}(m+1)$, i.e.,
\begin{align*}
    M\geq \frac{d}{2}+\frac{2m+d+2}{4(s_1-s_2)},
\end{align*}
then
\begin{align}
    \|II\|_{\cH}=\cO(\eps^{\frac{1}{2}(m+1)}).
\end{align}
This proves Lemma \ref{lem:3.3}.

\subsection{Proof of Lemma \ref{lem:3.4}}\label{sec:3.4} 
Consider $k\in \Omega^*$ with ${\rm dist}(k,\mathbb{L}^*)\geq 4\eps^{s_2}$. By \eqref{poisson},
\begin{align*}
   \MoveEqLeft I(k)=\frac{1}{(2\pi \eps)^d|\Omega|}\int_{\R^d\times \R^d}  e^{-i\frac{(k-k')\cdot X}{\eps}}\chi_\eps(X)\widetilde{\chi}_\eps(k-k') \\
   &\quad \times h_e\left(\frac{k+k'}{2},X\right)\sum_{G\in \mathbb{L}^*} e^{iG\cdot x}\vec{u}^T(x)\mathcal{F}(\vec{v}_\eps)\left(G+k'\right)dk'dX\\
   &=\frac{1}{(2\pi \eps)^d |\Omega|}\int_{\R^d\times \R^d}  e^{-i\frac{(k-k')\cdot X}{\eps}}\chi_\eps(X)\widetilde{\chi}_\eps(k-k') \\
   &\quad \times \sum_{G\in \mathbb{L}^*} e^{iG\cdot x} h_e\left(\frac{k+k'}{2}+G,X\right) \vec{u}^T(x)\mathcal{F}(\vec{v}_\eps)\left(G+k'\right)dk'dX
\end{align*}
where in the last equation, we used the Gauge invariance of the operator $T(-i\nabla_x+\bA(x,X))$. 

For $k'\in \Supp(\widetilde{\chi}_\eps(k-\cdot))$ and ${\rm dist}(k,\mathbb{L}^*)\geq 4\eps^{s_2}$, we have ${\rm dist}(k',\mathbb{L}^*)\geq 2\eps^{s_2}$. Then
\begin{align}\label{eq:v-G}
   \mathcal{F}(\vec{v}_\eps)\left(G+k'\right)=\frac{1}{|G+k'|^{2M}}\mathcal{F}\left((-\Delta)^M \vec{v}_\eps\right)\left(G+k'\right).
\end{align}
Hence for $k\in \Omega^*$ with ${\rm dist}(k,\mathbb{L}^*)\geq 4\eps^{s_2}$, using \eqref{eq:v-G} for $k_0=0$ and Assumption \ref{ass:localization}, 
\begin{align*}
    \|I(k)\|_{L^2_{\rm per}}&\lesssim \eps^{(s_2-s_1)d} \sup_{\substack{k'\in \Supp(\widetilde{\chi}_\eps(k-\cdot))\\ {\rm dist}(k',\mathbb{L}^*)\geq 2\eps^{s_2}}}\sum_{G\in\mathbb{L}^*} \frac{ (1+|G|^2)}{|G+k'|^{2M}}\left|\mathcal{F}\left((-\Delta)^M\vec{v}_\eps\right)\left(G+k'\right)\right| \left\|\vec{u}\right\|_{H^2_{\rm per}}\\
    &\lesssim  \eps^{(s_2-s_1)d} \sup_{\substack{k'\in \Supp(\widetilde{\chi}_\eps(k-\cdot))\\ {\rm dist}(k',\mathbb{L}^*)\geq 2\eps^{s_2}}} \sum_{G\in\mathbb{L}^*} \frac{\eps^{M-\frac{d}{4}}(1+|G+k'|^2) }{|G+k'|^{2M}}
\end{align*}
where in the first inequality we used the fact that for $X\in \Supp(\chi_\eps)$ and $k'\in \Supp(\widetilde{\chi}_\eps(k-\cdot))$ with $k\in \Omega^*$,
\begin{align*}
    \left\|h_e\left(\frac{k+k'}{2}+G,X\right) u\right\|_{L^2_{\rm per}}\lesssim  (1+|G|^2)\|u\|_{H^2_{\rm per}},
\end{align*} 
and in the second we used
\begin{align*}
  \sup_{k'\in \R^d}  \left|\mathcal{F}\left((-\Delta)^M\vec{v}_\eps\right)\left(k'\right)\right|\lesssim \|(-\Delta)^M\vec{v}_\eps\|_{L^1(\R^d)}\lesssim\eps^{M-\frac{d}{4}}.
\end{align*}

By the definition of $\Omega^*$, there exists a constant $C$ such that for any $G\in \mathbb L^*\setminus\{0\}$,
 \begin{align}\label{eq:distance-Omega-G}
    {\rm dist}(G,\Omega^*)\geq C |G|>0.
 \end{align}
Thus we consider the case $G=0$ and $G\in \mathbb L^*\setminus\{0\}$ separately, for $2M> d+2$, 
\begin{align*}
     \|I(k)\|_{L^2_{\rm per}}&\lesssim \eps^{\frac{2M+(2s_2-2s_1-1/2)d}{2}}\left(\sup_{\substack{k'\in \Supp(\widetilde{\chi}_\eps(k-\cdot))\\ |k'|\geq 2\eps^{s_2}}}  \frac{(1+|k'|^2) }{|k'|^{2M}} + \sup_{\substack{k'\in \Omega^*}} \sum_{G\in\mathbb{L}^*\setminus\{0\}} \frac{(1+|G+k'|^2) }{|G+k'|^{2M}}\right)\\
     &\lesssim \eps^{\frac{2M+(2s_2-2s_1-1/2)d}{2}}\left(\sup_{|k'|\geq 2\eps^{s_2}}\frac{(1+|k'|^2)}{|k'|^{2M}}+ \sum_{G\in \mathbb L^*\setminus\{0\}} \frac{(1+|G|^2)}{|G|^{2M}}\right)\\
     &=\cO(\eps^{\frac{2M (1-2s_2)+(2s_2-2s_1-1/2)d}{2}}).
\end{align*}
Note that
\begin{align*}
    \frac{1}{2}-s_2=\frac{1}{2\mathfrak{n} d(m+1)}>0.
\end{align*}
Then choosing $M$ such that $2M (1-2s_2)+(2s_2-2s_1-1/2)d\geq (m+1)$, i.e.,
\begin{align*}
    M&\geq \frac{(m+1)-(2s_2-2s_1-1/2)d}{(2-4s_2)},
\end{align*}
gives
 \begin{align*}
       \left( \int_{\substack{k\in \Omega^*\\ {\rm dist}(k-k_0,\mathbb L^*)\geq  4\eps^{s_2}}}\|I(k)\|_{L^2_{\rm per}}^2dk\right)^{1/2}=\cO(\eps^{\frac{1}{2}(m+1) }).
    \end{align*}
Hence Lemma \ref{lem:3.4}.

\subsection{Proof of Lemma \ref{lem:3.5}}\label{sec:3.5}
We now consider the case $|k|< 4\eps^{s_2}$. For any $k'\in \Supp(\widetilde{\chi}_\eps(k-\cdot))$, we have
\begin{align*}
    |k'|\leq 6\eps^{s_2}.
\end{align*} 
To prove this lemma, we split the proof into the following four steps.

\medskip

\noindent{\bf Step 1. Reduction of $(\cU u_\eps)_{k'}(x)$.} We replace $(\cU u_\eps)_{k'}(x)$ by $\vec{u}^T(x)\mathcal{F}(\vec{v}_\eps)(k')$ in $I(k)$. Let
\begin{align*}
    I_1(k):=\frac{1}{(2\pi\eps)^d |\Omega|}\int_{\R^d\times \R^d}  e^{-i\frac{(k-k')\cdot X}{\eps}}\chi_\eps(X)\widetilde{\chi}_\eps(k-k') h_e\left(\frac{k+k'}{2},X\right) \vec{u}^T\mathcal{F}(\vec{v}_\eps)\left(k'\right)\;dk'dX.
\end{align*}
For $|k|< 4\eps^{s_2}$, using \eqref{poisson} and \eqref{eq:v-G} again, and arguing as for Lemma \ref{lem:3.4} with same $M$,
\begin{align*}
 \MoveEqLeft   \left\|I(k)-I_1(k)\right\|_{L^2_{\rm per}}\\
 &\lesssim \sum_{G\in\mathbb L^*\setminus\{0\}}\eps^{-d}\int_{\R^d\times \R^d}  \chi_\eps(X)\widetilde{\chi}_\eps(k-k') |\mathcal{F}(\vec{v}_\eps)\left(k'+G\right)|\left\|h_e\left(\frac{k+k'}{2}+G,X\right) \vec{u}\right\|_{L^2_{\rm per}}\;dk'dX \\
  &\lesssim \eps^{-d}\sum_{G\in\mathbb L^*\setminus\{0\}}\int_{\R^d\times \R^d}  \chi_\eps(X)\widetilde{\chi}_\eps(k-k') \left|\mathcal{F}(\vec{v}_\eps)\left(k'+G\right)\right|(1+|G|^2) \left\|\vec{u}\right\|_{H^2_{\rm per}}\;dk'dX \\
 &\lesssim \eps^{(s_2-s_1)d}\sup_{|k'|\leq 6\eps^{s_2}}\sum_{G\in\mathbb{L}^*\setminus\{0\}} \frac{(1+|G|^2)}{|G+k'|^{2M}}\left|\mathcal{F}\left((-\Delta)^M\vec{v}_\eps\right)\left(G+k'\right)\right| \\
 &=\cO(\eps^{\frac12(m+1)})
\end{align*}
where in the last estimate we used the fact that for any $G\in \mathbb L^*\setminus \{0\}$,
\begin{align*}
   \inf_{|k'|\leq 6\eps^{s_2}} |G+k'|\geq C|G|>0.
\end{align*}
Thus,
 \begin{align}\label{eq:I-I1}
       \left( \int_{|k|<  4\eps^{s_2}}\|I(k)-I_1(k)\|_{L^2_{\rm per}}^2dk\right)^{1/2} =\cO(\eps^{\frac12(m+1)}).
    \end{align}

\medskip

\noindent{\bf Step 2. Reduction of $h_e^\bot$.} We approximate $h(k,X)$. Note that $X\in \Supp(\chi_\eps(X))$ implies
\begin{align*}
    |X-X_0|\leq 2\eps^{1-s_1},
\end{align*}
and $|k|<4\eps^{s_2}$ and $k'\in \Supp(\widetilde{\chi}_\eps(k-\cdot))$ imply that
\begin{align*}
    \left|\frac{k+k'}{2}\right|\leq 5\eps^{s_2}.
\end{align*}
Using the decomposition \eqref{eq:h-decomp}, write
\begin{align*}
     I_1(k)=I_{1,1}(k)+I_{1,2}(k)
\end{align*}
where $ I_{1,1}(k)$ corresponds to the first term on the right hand side of \eqref{eq:h-decomp}:
\begin{align*}
    I_{1,1}(k)&:=\frac{1}{(2\pi\eps)^d|\Omega| }\sum_{j=1}^J\int_{\R^d\times \R^d}  e^{-i\frac{(k-k')\cdot X}{\eps}}\chi_\eps(X)\widetilde{\chi}_\eps(k-k')  \widetilde{\phi}_{j}\left(\frac{k+k'}{2},k',X\right)  \;dk'dX
\end{align*}
with
\begin{align}\label{eq:phi-tilde-k-k'-X}
 \MoveEqLeft \widetilde{\phi}_{j}\left(k,k',X\right) :=  (E_{j}\left(k,X\right)-e_0)\left<\phi_{j}\left(k,X\right) ,\vec{u}^T\mathcal{F}(\vec{v}_\eps)(k')\right>_{L^2_{\rm per}}\phi_{j}\left(k,X\right);
\end{align}
and $ I_{1,2}(k)$ corresponds to the second term on the right hand side of \eqref{eq:h-decomp}:
\begin{align*}
     I_{1,2}(k):=\frac{1}{(2\pi\eps)^d |\Omega|}\int_{\R^d\times \R^d}  e^{-i\frac{(k-k')\cdot X}{\eps}}\chi_\eps(X)\widetilde{\chi}_\eps(k-k') h_e^\bot\left(\frac{k+k'}{2},X\right) \vec{u}^T\mathcal{F}(\vec{v}_\eps)\left(k'\right)\;dk'dX.
\end{align*}
Using \eqref{eq:I2-eff} and \eqref{eq:taylor expansion},
\begin{align*}
\MoveEqLeft\|I_1(k)-I_{1,1}(k)-I_2^{\rm eff}(k)\|_{L^2_{\rm per}}=\|I_{1,2}(k)-I_2^{\rm eff}(k)\|_{L^2_{\rm per}} \\
&\lesssim \eps^{-d}\int_{\R^d\times \R^d}\chi_\eps(X)\widetilde{\chi}_\eps(k-k') (|k+k'|^{m+1}+|X-X_0|^{m+1})\left|\mathcal{F}(\vec{v}_\eps)(k')\right| dk'dX \\
&=\cO(\eps^{(s_2-s_1)d+(m+1)\min\{1-s_1,s_2\}-\frac{d}{4}}).
\end{align*}
where we also used the fact that, by Assumption \ref{ass:localization},
\begin{align*}
    \left|\mathcal{F}(\vec{v}_\eps)(k')\right|\lesssim \|\vec{v}_\eps\|_{L^1(\R^d)}=\cO(\eps^{-\frac{d}{4}}).
\end{align*}
By the definition of $s_1$ and $s_2$ in Assumption \ref{ass:bandstructure-nondegenerate},
\begin{align*}
\MoveEqLeft (s_2-s_1)d+(m+1)\min\{1-s_1,s_2\}+(2s_2-1)\frac{d}{4}\\
&\geq  -\frac{1}{\mathfrak{n}(m+1)} +\frac{1}{2}(m+1)- \frac{1}{2\mathfrak{n}d}- \frac{1}{4\mathfrak{n}(m+1)}\geq \frac{1}{2}m+\frac{1}{4}.
\end{align*}
Thus,
\begin{align*}
    \|I_1(k)-I_{1,1}(k)-I_2^{\rm eff}(k)\|_{L^2_{\rm per}}=\cO(\eps^{\frac{1}{2}m+\frac{1}{4} -\frac{s_2}{2}d }).
\end{align*}
Therefore,
\begin{align}\label{eq:eq:I1-I11-I2-eff}
    \MoveEqLeft\left( \int_{|k|<  4\eps^{s_2}}\|I_1(k)-I_{1,1}(k)-I_2^{\rm eff}(k)\|_{L^2_{\rm per}}^2dk\right)^{1/2}=\cO(\eps^{\frac{1}{2} m+\frac{1}{4}}).
\end{align}

\medskip

\noindent{\bf Step 3. Reduction of the band around $e_0$.} We  now approximate $I_{1,1}(k)$ by $\widetilde{I}_{1}^{\rm eff}(k)$, defined by
\begin{align*}
 \MoveEqLeft   \widetilde{I}_{1}^{\rm eff}(k):= \frac{1}{(2 \pi\eps)^d|\Omega|}\int_{\R^d\times \R^d}\;dk'dX\;\;\;  e^{-i\frac{(k-k')\cdot X}{\eps}} \chi_\eps(X)\widetilde{\chi}_\eps(k-k')\\
   &\times \sum_{1\leq j\leq J}\vec{w}^T(x)\vec{\alpha}_j\left(\tfrac{k+k'}{2},X\right)\left<\vec{w}^T  \Big(h^{\rm eff}\left(\tfrac{k+k'}{2},X\right) \vec{\alpha}_j\left(\tfrac{k+k'}{2},X\right)\Big) , \vec{u}^T\mathcal{F}(\vec{v}_\eps)(k')  \right>_{L^2_{\rm per}}.
\end{align*}

By Assumption \ref{ass:bandstructure-nondegenerate}, for $|X-X_0|\leq 2\eps^{1-s_1}$ and $ \left|k\right|<4\eps^{s_2}$,
\begin{align*}
   \widetilde{E}_{j}(k,X)=E_{j}(k,X)-e_0= \lambda^{\rm eff}_{j}(k,X)+\cO(|k|^{m+1}+|X-X_0|^{m+1})
\end{align*}
and
\begin{align*}
      h^{\rm eff}(k,X)\vec{\alpha}_j(k,X)= \lambda^{\rm eff}_j(k,X)\vec{\alpha}_j(k,X)
\end{align*}
with
\begin{align*}
    \|\phi_{j}(k,X)-\vec{w}^T\, \vec{\alpha}_j(k,X)\|_{L^2_{\rm per}}=\cO(|k|+|X-X_0|).
\end{align*}
Moreover, under Assumption \ref{ass:bandstructure-nondegenerate}, as $f^{\rm eff}$ is a homogeneous function of degree $m$,
\begin{align*}
    |\lambda_j^{\rm eff}(k,X)| \lesssim | h^{\rm eff}(k,X)|_2\lesssim |k|^{m}+|X-X_0|^{m}
\end{align*}
and by \eqref{eq:ass1-1},
\begin{align*}
    |\widetilde{E}_{j}\left(k,X\right)|\lesssim |k-k_0|^{m}+|X-X_0|^{m}.
\end{align*}
Thus, for $|X-X_0|\leq 2\eps^{1-s_1}$ and $ \left|k\right|<4\eps^{s_2}$, and for $1\leq j\leq J$,
\begin{align*}
  \MoveEqLeft \Big|\widetilde{E}_{j}\left(k,X\right)  \left<\phi_{j}\left(k,X\right),\vec{u}^T\mathcal{F}(\vec{v}_\eps)(k')\right>_{L^2_{\rm per}}\\
  &\quad -\lambda^{\rm eff}_j(k,X) \left<\vec{w}^T \vec{\alpha}_j(k,X),\vec{u}^T\mathcal{F}(\vec{v}_\eps)(k')\right>_{L^2_{\rm per}}\Big|\lesssim (|k|^{m +1}+|X-X_0|^{m+1}) |\mathcal{F}(\vec{v}_\eps)(k')|.
\end{align*}
Define the total error term by
\begin{align*}
e(k,k',X):&=\sum_{j=1}^J\widetilde{\phi}_{j}\left(k,k',X\right)-\sum_{j=1}^J\left<\vec{w}^T  \big(h^{\rm eff}\vec{\alpha}_j\big)(k,X),\vec{u}^T\mathcal{F}(\vec{v}_\eps)(k')\right>_{L^2_{\rm per}} \vec{w}^T \vec{\alpha}_j(k,X)\\
&=\sum_{j=1}^J\widetilde{E}_{j}\left(k,X\right)  \left<\phi_{j}\left(k,X\right),\vec{u}^T\mathcal{F}(\vec{v}_\eps)(k')\right>_{L^2_{\rm per}}\phi_{j}(k,X)\\
&\quad-\sum_{j=1}^J\left<\vec{w}^T  \big(h^{\rm eff}\vec{\alpha}_j\big)(k,X),\vec{u}^T\mathcal{F}(\vec{v}_\eps)(k')\right>_{L^2_{\rm per}} \vec{w}^T \vec{\alpha}_j(k,X) .
\end{align*}
Using above estimate and by \eqref{eq:ass1-2} in addition,
\begin{align*}
\|e(k,k',X)\|_{L^2_{\rm per}} \lesssim (|k|^{m +1}+|X-X_0|^{m+1}) |\mathcal{F}(\vec{v}_\eps)(k')|
\end{align*}
where we also used 
\begin{align*}
    h^{\rm eff}(k,X)\vec{\alpha}_j(k,X)= \lambda^{\rm eff}_j(k,X)\vec{\alpha}_j(k,X).
\end{align*}
Analogous to Step. 2 in this proof, this and Assumption \ref{ass:localization} give
\begin{align}\label{eq:I11-I1-tilde-eff}
   \MoveEqLeft    \left( \int_{|k|< 4\eps^{s_2}}\|I_{1,1}(k)-\widetilde{I}_1^{\rm eff}(k)\|_{L^2_{\rm per}}^2dk\right)^{1/2}=\cO(\eps^{\frac{m}{2}+\frac{1}{4}}).
\end{align}

\medskip

\noindent{\bf Step. 4. Study of $\widetilde{I}_{1}^{\rm eff}(k)$.} We now show
\begin{align*}
     \widetilde{I}_{1}^{\rm eff}(k)= I_{1}^{\rm eff}(k).
\end{align*}
Indeed, $\vec{w}^T  \big(h^{\rm eff}(k,X)\vec{\alpha}_j(k,X))=\sum_{1\leq \ell\leq J}w_\ell (h^{\rm eff}\vec{\alpha}_j)_\ell$, and
\begin{align*}
  \MoveEqLeft  \left<\vec{w}^T  \big(h^{\rm eff}(k,X)\vec{\alpha}_j(k,X)\big),\vec{u}^T\mathcal{F}(\vec{v}_\eps)(k')\right>_{L^2_{\rm per}} \\
  &=  \sum_{1\leq \ell,\ell'\leq J}\left<w_{\ell'},u_\ell\right>_{L^2_{\rm per}}  \overline{\big(h^{\rm eff}(k,X)\vec{\alpha}_j(k,X)\big)}_{\ell'}\;  \mathcal{F}(v_{\ell,\eps})(k')\\
  &=\overline{\big(h^{\rm eff}(k,X)\vec{\alpha}_j(k,X)\big)}^T \mathcal{M}_{\vec{u}}\mathcal{F}(\vec{v}_\eps)(k')\\
  &=\overline{\vec{\alpha}}_j^T(k,X)h^{\rm eff}(k,X)\mathcal{M}_{\vec{u}}\mathcal{F}(\vec{v}_\eps)(k').
\end{align*}
where we recall that $\mathcal{M}_{\vec{u}}$ is defined by \eqref{eq:M-u} and $h^{\rm eff}$ is Hermitian, i.e., $\overline{h^{\rm eff}}^T=h^{\rm eff}$.

According to Assumption \ref{ass:bandstructure-nondegenerate}, $(\vec{\alpha}_j)_{1\leq j\leq J}$ forms an orthonormal basis of $\C^J$, then
\begin{align*}
   \sum_{1\leq j\leq J}\vec{\alpha}_j \overline{\vec{\alpha}}_j^T=\1_{\C^J}.
\end{align*}
Thus,
\begin{align*}
\MoveEqLeft  \sum_{1\leq j\leq J}\vec{\alpha}_j(k,X) \left<\vec{w}^T  \big(h^{\rm eff}(k,X)\vec{\alpha}_j(k,X)\big),\vec{u}^T\mathcal{F}(\vec{v}_\eps)(k')\right>_{L^2_{\rm per}}\\
  &= \left(\sum_{1\leq j\leq J}\vec{\alpha}_j\overline{\vec{\alpha}}_j^T\right)(k,X)h^{\rm eff}(k,X)\mathcal{M}_{\vec{u}}\mathcal{F}(\vec{v}_\eps)(k')= h^{\rm eff}(k,X)\mathcal{M}_{\vec{u}}\mathcal{F}(\vec{v}_\eps)(k').
\end{align*}
As a result,
\begin{align*}
    \widetilde{I}_{1}^{\rm eff}(k)&= \frac{1}{(2 \pi\eps)^d |\Omega|}\int_{\R^d\times \R^d}e^{-i\frac{(k-k')\cdot X}{\eps}} \chi_\eps(X)\widetilde{\chi}_\eps(k-k')\\
    &\qquad\qquad \times \vec{w}^T(x) h^{\rm eff}\left(\tfrac{k+k'}{2},X\right)\mathcal{M}_{\vec{u}}\mathcal{F}(\vec{v}_\eps)(k')\;dk'dX.
\end{align*}
This is exactly \(I_1^{\mathrm{eff}}(k)\) from \eqref{eq:I1-eff}, i.e.,
\begin{align*}
      \widetilde{I}_{1}^{\rm eff}(k)=  I_{1}^{\rm eff}(k).
\end{align*}
Combining \eqref{eq:I-I1}, \eqref{eq:eq:I1-I11-I2-eff} and \eqref{eq:I11-I1-tilde-eff}, we obtain
  \begin{align*}
       \left( \int_{|k|< 4\eps^{s_2}}\|I(k)-I_1^{\rm eff}(k)-I_2^{\rm eff}(k)\|_{L^2_{\rm per}}^2dk\right)^{1/2}=\cO(\eps^{\frac{1}{2} m+\frac{1}{4}})
    \end{align*}
This proves Lemma \ref{lem:3.5}.

\subsection{Proof of Lemma \ref{lem:3.6'}}\label{sec:3.6'}

 The proof is essentially the same as for Lemma \ref{lem:3.4} in Section \ref{sec:3.4}. As $|k|\geq 4\eps^{s_2}$ and $k'\in \Supp(\widetilde{\chi}_\eps(k-\cdot))$, we have $|k'|\geq 2\eps^{s_2}$. Replacing \eqref{eq:v-G} by
 \begin{align*}
    \mathcal{F}(\vec{v}_\eps)(k')=\frac{1}{|k'|^{2M}} \mathcal{F}((-\Delta)^{M}\vec{v}_\eps)(k').
\end{align*}
and following the same argument as in Lemma \ref{lem:3.4} yields the lemma.

\subsection{Proof of Lemma \ref{lem:3.6}}\label{sec:3.6}
We only prove the estimate for
    \begin{align*}
          \left( \fint_{\Omega^*}\|  I^{\rm eff}_1(k)-\widetilde{\widetilde{\mathfrak{h}}}^{\rm eff}_{1,\eps}(\vec{u},\vec{v}_\eps) \|_{L^2_{\rm per}}^2dk\right)^{1/2}=\cO(\eps^{\frac{m+1}{2}}).
    \end{align*}
The other term can be studied in the same manner. 

Let
\begin{align*}
   \widetilde{\widetilde{I}}_1^{\rm eff}(k):&=\frac{1}{(2 \pi\eps)^d |\Omega|}\int_{\R^d\times \R^d}  e^{-i\frac{(k-k')\cdot X}{\eps}} \chi_\eps(X) \vec{w}^T(x) h^{\rm eff}\left(\tfrac{k+k'}{2},X\right) \mathcal{M}_{\vec{u}}\mathcal{F}(\vec{v}_\eps)(k')  \;dk'dX.
\end{align*}
Analogously to the proof of Lemma \ref{lem:3.3} with same $M$, by integration by parts, for $M'=M+m$
\begin{align*}
 \MoveEqLeft \left\| \widetilde{\widetilde{I}}_1^{\rm eff}(k)-  I_1^{\rm eff}(k)\right\|_{L^2_{\rm per}}\\
 &\lesssim\eps^{-d}\left|\int_{\R^d\times \R^d}  e^{-i\frac{(k-k')\cdot X}{\eps}} \chi_\eps(X)(1-\widetilde{\chi}_\eps(k-k')) h^{\rm eff}\left(\frac{k+k'}{2},X\right)\mathcal{F}(\vec v_\eps)\left(k'\right)\;dk'dX\right|\\
 &\lesssim \eps^{2M'-d}\int_{\R^d\times \R^d}   \frac{ (1-\widetilde{\chi}_\eps(k-k'))}{|k-k'|^{2M'}}\left|(-\Delta_X)^{M'}\left[\chi_\eps(X)h^{\rm eff}\left(\frac{k+k'}{2},X\right)\right]\right|\left|\mathcal{F}(\vec{v}_\eps)(k')\right|\;dk'dX\\
 &\lesssim \eps^{2s_1M'-\frac{5d}{4}}\int_{\substack{|k-k'|\geq \eps^{s_2}\\ |X-X_0|\leq 2\eps^{1-s_1}}} \frac{1+|k-k'|^m}{|k-k'|^{2M'}}\;dk'dX=\cO(\eps^{\frac{1}{2}(m+1)}).
\end{align*}
Hence,
\begin{align*}
    \left(\fint_{\Omega^*}\left\| \widetilde{\widetilde{I}}_1^{\rm eff}(k)-  I_1^{\rm eff}(k)\right\|_{L^2_{\rm per}}^2dk\right)^{1/2}=\cO(\eps^{\frac{1}{2}(m+1)}).
\end{align*}

Concerning $ \widetilde{\widetilde{I}}_1^{\rm eff}(k)$, we observe that
\begin{align*}
    \widetilde{\widetilde{I}}_1^{\rm eff}(k)=|\Omega|^{-1} \vec{w}^T(x)\Op(h^{\rm eff}\chi_\eps)\mathcal{M}_{\vec u}\mathcal{F}(\vec v_\eps)(k).
\end{align*}
Under Assumption \ref{ass:bandstructure-nondegenerate}, $h^{\rm eff}$ is a matrix-valued polynomial function  of the degree at most $m$. Arguing as for Lemma \ref{lem:3.2},
\begin{align*}
     \widetilde{\widetilde{I}}_1^{\rm eff}(k)-\widetilde{\widetilde{\mathfrak{h}}}^{\rm eff}_{1,\eps}(\vec{u},\vec{v}_\eps)(k)=|\Omega|^{-1} \vec{w}^T(x)\left(\Op(h^{\rm eff}\chi_\eps)-\Op(\chi_\eps)\Op(h^{\rm eff})\right)\mathcal{M}_{\vec{u}}\mathcal{F}(\vec v_\eps)(k)
\end{align*}
and $\left(\Op(h^{\rm eff}\chi_\eps)-\Op(\chi_\eps)\Op(h^{\rm eff})\right)\mathcal{M}_{\vec{u}}\mathcal{F}(\vec v_\eps)(k)$ is a combination of terms of the form
\begin{align*}
&\eps^{|\gamma|_1}
\Op(\partial_k^\gamma h^{\rm eff})
\Op(\partial_X^\gamma\chi_\eps)
\mathcal M_{\vec u}\mathcal F(\vec v_\eps)\\
&\quad=
\eps^{s_1|\gamma|_1}
\Op(\partial_k^\gamma h^{\rm eff})
\mathcal F\left[
(\partial^\gamma\chi)
(\eps^{s_1}\bullet-\eps^{s_1-1}X_0)
\mathcal M_{\vec u}\vec v_\eps
\right].
\end{align*}
for some $\gamma:=(\gamma_1,\cdots,\gamma_d)\in \N^{d}$ with $1\leq |\gamma|_1\leq m$.

Since $\partial_k^\gamma h^{\rm eff}$ is a polynomial of degree $m-|\gamma|_1$, for any function $g\in \mathcal{S}(\R^d)$,
\begin{align*}
  \|\Op(\partial_k^\gamma h^{\rm eff}) g\|_{L^2(\R^d)}\lesssim \||k|^{m-|\gamma|_1} g\|_{L^2(\R^d)}+\|(i\eps\nabla_k-X_0)^{m-|\gamma|_1} g\|_{L^2(\R^d)}.
\end{align*}
Then,
\begin{align*}
\MoveEqLeft   \left\| \eps^{|\gamma|_1} \Op(\partial_k^\gamma h^{\rm eff})\Op(\partial^\gamma \chi_{\eps})\mathcal{M}_{\vec{u}}\mathcal{F}(\vec v_\eps)(k)\right\|_{L^2(\R^d)}\\
&\lesssim \eps^{|\gamma|_1s_1}\|(-\Delta)^{\frac{m-|\gamma|_1}{2}}\left((\partial^\gamma \chi)(\eps^{s_1}\bullet-\eps^{s_1-1}X_0 )\vec v_\eps\right)\|_{L^2(\R^d)}\\
&\quad+\eps^{|\gamma|_1s_1}\|(\eps x-X_0)^{m-|\gamma|_1}\left((\partial^\gamma \chi)(\eps^{s_1}\bullet-\eps^{s_1-1}X_0 )\vec v_\eps\right)\|_{L^2(\R^d)}\\
&\lesssim \eps^{|\gamma|_1s_1}\|\vec v_\eps\|_{H^{m-|\gamma|_1}(\R^d\setminus B_{\eps^{-s_1}}(\eps^{-1}X_0))}\lesssim \eps^{|\gamma|_1s_1} \|\vec v_\eps\|_{W^{d+1+m-|\gamma|_1,1}(\R^d\setminus B_{\eps^{-s_1}}(\eps^{-1}X_0))}=\cO(\eps^{\frac{1}{2}(m+1)})
\end{align*}
where we used the fact that
\begin{align*}
\Supp(\partial^\gamma\chi(\eps^{s_1}\bullet-\eps^{s_1-1}X_0))\subset B_{2\eps^{-s_1}}(\eps^{-1}X_0)\setminus  B_{\eps^{-s_1}}(\eps^{-1}X_0).
\end{align*}
Thus
\begin{align*}
     \MoveEqLeft     \left( \fint_{\Omega^*}\|  I^{\rm eff}_1(k)-\widetilde{\widetilde{\mathfrak{h}}}^{\rm eff}_{1,\eps}(\vec{u},\vec{v}_\eps) \|_{L^2_{\rm per}}^2dk\right)^{1/2}\\
     &\lesssim \|\vec{w}\|_{L^2_{\rm per}}\left( \int_{\R^d} \left|\left(\Op(h^{\rm eff}\chi_\eps)-\Op(\chi_\eps)\Op(h^{\rm eff})\right)\mathcal{M}_{\vec{u}}\mathcal{F}(\vec v_\eps)(k)\right|^2dk \right)^{1/2}\\
     &=\cO(\eps^{\frac{1}{2}(m+1)}).
    \end{align*}
This completes the proof of Lemma \ref{lem:3.6}.

\section{Application: quantum harmonic oscillator effective Hamiltonian}\label{sec:quantum-harmonic-oscillator}
In this section, we consider the case $m=2$ relevant to quantum harmonic oscillator operator. In this case, we need the following assumption.

\begin{assumption}\label{ass:harmonic-oscillator}
Let $\j\in \N^+$. We assume that  the operator $H_\eps$ or equivalently the family of operators $(h(k,X))_{k,X}$ satisfies the following property:
\begin{itemize}
    \item The eigenvalue $e_0=E_{\j}(k_0,X_0)$ is a non-degenerate eigenvalue of the operator $h(k_0,X_0)$ associated with the eigenfunction $w\in L^2_{\rm per}$;
    \item For the mapping $(k,X)\mapsto E_\j(k,X)$,
    \begin{align*}
        \nabla E_\j(k_0,X_0)=0,\qquad \nabla:=(\nabla_k,\nabla_X),
    \end{align*}
    and the Hessian matrix
    \begin{align*}
        \nabla^2 E_\j(k_0,X_0)=\begin{pmatrix}
           ( \partial_{k_i}\partial_{k_j} E_\j(k_0,X_0))_{1\leq i,j\leq d} & ( \partial_{k_i}\partial_{X_j} E_\j(k_0,X_0))_{1\leq i,j\leq d}\\ ( \partial_{X_i}\partial_{k_j} E_\j(k_0,X_0))_{1\leq i,j\leq d}&  ( \partial_{X_i}\partial_{X_j} E_\j(k_0,X_0))_{1\leq i,j\leq d}
        \end{pmatrix}
    \end{align*}
    is either strictly positive or strictly negative definite.
\end{itemize}
\end{assumption}
Then,
\begin{theorem}\label{th:harmonic-oscillator}
We assume that $H_\eps$ satisfies Assumption \ref{ass:harmonic-oscillator}. Depending on the sign of the Hessian matrix $ \nabla^2 E_\j(k_0,X_0)$, let
$( v_j,\mu_j)$ be the $j$-th  positive or negative eigenpair of the general quantum harmonic oscillator operator with an energy shift:
\begin{align*}
  \mathfrak{h}:=\frac{1}{2}\begin{pmatrix}-i\nabla_x\\ x
  \end{pmatrix}^T  \nabla^2 E_\j(k_0,X_0) \begin{pmatrix}-i\nabla_x\\ x
      \end{pmatrix} + \sum_{\beta=0,\;|\gamma|_1=1} \left<w, {\rm Im}(P_{\beta,\gamma,0}^\bot h_{e,0} P_{\gamma,\beta,0}^\bot) w\right>_{L^2_{\rm per}},
\end{align*}
and let
\begin{align*}
    \Phi_{j,\eps}:=\Phi_\eps\Big((U^{(0)}_\eps+\sqrt{\eps} U^{(1)}_\eps+\eps U^{(2)}_\eps)(w\otimes v_j)\Big)
\end{align*}
where $U_\eps^{(0)}$, $U_\eps^{(1)}$ and $U^{(2)}_\eps$ are defined by \eqref{eq:U0-eps}, \eqref{eq:U1-eps} and \eqref{eq:U2-epsilon} respectively, and $w$ is defined as in Assumption \ref{ass:harmonic-oscillator}. Then  for $\eps$ small enough,
\begin{align}\label{eq:harmonic-oscillator}
    \|(H_\eps-e_0-\eps \mu_j)\Phi_{j,\eps}\|_{L^2(\R^d;\C^n)}=\cO(\eps^{\frac{5}{4}})
\end{align}
with
\begin{align*}
    \left\|\Phi_{j,\eps}\right\|_{L^2(\R^d;\C^n)}=\frac{1}{|\Omega|^{1/2}}+\cO(\sqrt{\eps}).
\end{align*}
\end{theorem}

\begin{proof}
    Note that
    \begin{align*}
        (k,X)\mapsto h(k,X)
    \end{align*}
    is a smooth operator-valued function. Thus according to perturbation theory (see e.g., \cite[Ch. XII]{Reed-Simon-1978-operator}), as $e_0$ is a non-degenerate eigenvalue of $h(k_0,X_0)$, we know that in a small neighborhood of $(k_0,X_0)$,
    \begin{itemize}
        \item  $(k,X)\mapsto E_{\j}(k,X)$ is smooth, \begin{align*}
     E_{\j-1}(k,X)<   E_{\j}(k,X)<E_{\j+1}(k,X)
    \end{align*}
    and
    \begin{align*}
        E_{\j}(k,X)=e_0+\frac{1}{2}\begin{pmatrix}
        k-k_0\\X-X_0
    \end{pmatrix}^T\nabla^2 E_\j(k_0,X_0)\begin{pmatrix}
        k-k_0\\X-X_0
    \end{pmatrix}+\cO(|k-k_0|^3+|X-X_0|^3);
    \end{align*}
    \item $(k,X)\mapsto \phi_{\j}(k,X)$ is smooth, and
    \begin{align*}
        \|\phi_{\j}(k,X)-\phi_{\j}(k_0,X_0)\|_{L^2_{\rm per}}=\cO(|k-k_0|+|X-X_0|).
    \end{align*}
    \end{itemize}
Thus in this case, Assumption \ref{ass:bandstructure-nondegenerate} is fulfilled with $m=2$, $J=1$ and
\begin{align}\label{eq:f-quadratic}
    f^{\rm eff}_{m}=\frac{1}{2}\begin{pmatrix}
        k\\X
    \end{pmatrix}^T\nabla^2 E_\j(k_0,X_0)\begin{pmatrix}
        k\\X
    \end{pmatrix}.
\end{align}
As an eigenfunction of quantum harmonic oscillator $\mathfrak{h}$, $v_j\in C^\infty(\R^d)$ and decays exponentially: for some $\kappa>0$ and any given $j\in \N^+$,
\begin{align*}
    e^{\kappa|x|}v_j(x)\in L^1(\R^d)\cap L^\infty(\R^d).
\end{align*}
Thus \eqref{eq:v2-m} and \eqref{eq:v3-m} with $m=2$ are satisfied. Hence this theorem follows from Theorem \ref{th:m}.
\end{proof}
\begin{remark}[General multi-dimensional quantum harmonic oscillator operator]\label{rem:harmonic-oscillator}
    We now briefly explain why $ \mathfrak{h}$ defined in Theorem \ref{th:harmonic-oscillator} is a multi-dimensional quantum harmonic oscillator operator. For simplicity, assume $A:=   \nabla^2 E_\j(k_0,X_0)$ is strictly positive.
    
    According to the Williamson theorem, there exists a symplectic matrix $S\in{\rm Sp}(2d,\R)$ such that
    \begin{align*}
        S^T A S=\begin{pmatrix}
            \Omega & 0\\0&\Omega
        \end{pmatrix},
    \end{align*}
    where 
    \begin{align*}
        \Omega={\rm diag}(\omega_1,\cdots, \omega_d),\qquad \omega_j>0,\quad j=1,\cdots,d.
    \end{align*}
    Thus by \eqref{eq:f-quadratic} below,
    \begin{align*}
         f^{\rm eff}_{m}=\frac{1}{2}\begin{pmatrix}
        k\\X
    \end{pmatrix}^TA\begin{pmatrix}
        k\\X
    \end{pmatrix}=\frac{1}{2}\left(S^{-1}\begin{pmatrix}
        k\\X
    \end{pmatrix}\right)^T \begin{pmatrix}
            \Omega & 0\\0&\Omega
        \end{pmatrix}\left(S^{-1}\begin{pmatrix}
        k\\X
    \end{pmatrix}\right).
    \end{align*}
    Inserting this decomposition into the Weyl quantization of the symbol $ f^{\rm eff}_{m}$ and using the following change of variable
     \begin{align*}
        \begin{pmatrix}
            k'\\X'
        \end{pmatrix}:=S^{-1} \begin{pmatrix}
            k\\X
        \end{pmatrix},
    \end{align*}
    one can obtain a unitary operator $U_S$ in $L^2(\R^d)$ such that
    \begin{align*}
       U^*_S {\rm Op}_1( f^{\rm eff}_{m})(k,i\nabla_k)U_S= \frac{1}{2}\sum_{j=1}^d \omega_j(-\partial_{k_j}^2+k_j^2)
    \end{align*}
    with $k:=(k_1,\cdots,k_d)$. Then
    \begin{align*}
         U^*\mathfrak{h} U=\frac{1}{2}\sum_{j=1}^d \omega_j(-\partial_{x_j}^2+x_j^2) + \sum_{\beta=0,\;|\gamma|_1=1} \left<w, {\rm Im}(P_{\beta,\gamma,0}^\bot h_{e,0} P_{\gamma,\beta,0}^\bot) w\right>_{L^2_{\rm per}}.
    \end{align*}
    with the unitary operator $U:=\mathcal{F}^{-1}U_S\mathcal{F}$ on $L^2(\R^d)$. Indeed, above argument can be viewed as an application of the Stone–von Neumann theorem for quantum harmonic oscillator operator.
\end{remark}

\subsection{Construction of aperiodic crystals from periodic crystals}\label{sec:4.1}
Assumption \ref{ass:harmonic-oscillator} can be satisfied in many different setting. Here we give a simple example to illustrate when it holds. This example explains how to construct aperiodic crystals possessing localized approximate eigenfunction in $L^2(\R^d;\C^n)$ from a periodic crystal.

\medskip

Let 
\begin{align*}
    H_0=-\frac{1}{2}\Delta+V_{\rm per}(x),\qquad x\in \R^d,
\end{align*}
be a $\mathbb L$-periodic operator with $\mathbb L$-periodic potential $V_{\rm per}\in C^\infty(\R^d;\R)$. According to \eqref{op-h0}, $H_0$ can be decomposed into the family of operators $h(k)$ on $L^2_{\rm per}$ defined by
\begin{align}\label{eq:4.3}
    \Omega^*\ni k\mapsto h_0(k)=\frac{1}{2}(-i\nabla_x+k)^2+V_{\rm per}(x).
\end{align}
Let $E_j(k)$ be the eigenvalue of $h(k)$ with $E_1(k)\leq E_2(k)\leq \cdots$. 

The following assumption can be easily satisfied by many periodic operators.
\begin{assumption}\label{ass:4.4}
 We assume that for the periodic operator $H_0$, there exists a point $k_0\in \Omega^*$ and a band structure $\j\in \N^+$ such that 
\begin{itemize}
    \item $E_\j(k_0)$ is a non-degenerate eigenvalue of $h(k_0)$. Thus by perturbation theory (see e.g., \cite[Ch. XII]{Reed-Simon-1978-operator}), $k\mapsto E_\j(k)$ is smooth in a small neighborhood of $k_0$ .
    \item $k_0$ is a strictly local nondegenerate minimal or maximal point of the mapping $k\mapsto E_\j(k)$. Hence $ \nabla_k E_\j(k_0)=0$ and the Hessian matrix $ \nabla_k^2 E_\j(k_0)$ is strictly positive or strictly negative.
\end{itemize}
\end{assumption}
Then,
\begin{corollary}\label{cor:4.5}
   Assume that $H_0$ satisfies Assumption \ref{ass:4.4}. Choose $\widetilde{V}(X)\in C^\infty(\R^d;\R)$ with a strictly local nondegenerate minimum (resp. maximum) at $X_0$ when $ \nabla_k^2 E_\j(k_0)$ is strictly positive (resp. negative). Then
\begin{align*}
    H_\eps:=H_0+\widetilde{V}(\eps x)
\end{align*}
is an operator satisfying Assumption \ref{ass:harmonic-oscillator} around $(k_0,X_0)$. Consequently, Theorem \ref{th:harmonic-oscillator} holds.
\end{corollary}
\begin{remark}
 In Corollary \ref{cor:4.5}, the potential $\widetilde{V}(X)$ can be a global function such as $|X|^2$, but it may also be a local function such as $|X|^2\chi(|X|)$.
\end{remark}

\begin{proof}[Proof of Corollary \ref{cor:4.5}]
    According to \eqref{op-h}, it suffices to consider the family of operators
    \begin{align*}
        h(k,X)=h_0(k)+\widetilde{V}(X).
    \end{align*}
Obviously, $(h(k,X))_{k,X}$ satisfies Assumption \ref{ass:harmonic-oscillator}, since $\widetilde{V}(X)$ is a scalar constant for any given $X\in \R^d$.
\end{proof}

\section{Application: Landau-Schr\"odinger effective Hamiltonian}\label{sec:standard-quantum-hall}
In this section, we study $H_\eps$ relevant to the standard quantum Hall effect with Landau-Schr\"odinger effective Hamiltonian, which correspond to the case $m=2$ and $d=2$. To avoid ambiguity, we replace $x,k,X\in \R^2$ by $\bx:=(x_1,x_2)\in \R^2$, $\bk:=(k_1,k_2)\in\R^2$ and $\bX:=(X_1,X_2)\in \R^2$ in this section.

We need the following assumption.
\begin{assumption}\label{ass:quantum-hall}
Let $d=2$ and $\j\in \N^+$. We assume that the operator $H_\eps$ or equivalently the family of operators $(h(\bk,\bX))_{\bk,\bX}$ satisfies the following property:
\begin{itemize}
    \item The eigenvalue $e_0=E_{\j}(\bk_0,\bX_0)$ is a non-degenerate eigenvalue of the operator $h(\bk_0,\bX_0)$ associated with the eigenfunction $w\in L^2_{\rm per}$;
    \item In a small neighborhood of $(\bk_0,\bX_0)=(k_{0,1},k_{0,2},X_{0,1},X_{0,2})$, 
    \begin{align*}
        E_{\j}(\bk,\bX)&=e_0+ \begin{pmatrix}
            (k_1-k_{0,1})+\frac{B}{2}(X_2-X_{0,2})\\ (k_2-k_{0,2}) -\frac{B}{2}(X_1-X_{0,1})
        \end{pmatrix}^TA \begin{pmatrix}
            (k_1-k_{0,1})+\frac{B}{2}(X_2-X_{0,2})\\ (k_2-k_{0,2}) -\frac{B}{2}(X_1-X_{0,1})
        \end{pmatrix}\\
        &\quad+\cO(|\bk-\bk_0|^3+|\bX-\bX_0|^3)
    \end{align*}
    where $B\neq 0$, $A$ is a strictly positive or strictly negative real symmetric matrix.
\end{itemize}
\end{assumption}
Then,
\begin{theorem}\label{th:quantum-hall}
We assume that $H_\eps$ satisfies Assumption \ref{ass:quantum-hall}. Let $U_\eps^{(0)}$, $U_\eps^{(1)}$ and $U^{(2)}_\eps$ be defined by \eqref{eq:U0-eps}, \eqref{eq:U1-eps} and \eqref{eq:U2-epsilon} respectively, and let $w$ be defined as in Assumption \ref{ass:quantum-hall}. 

Let
$\mu_j$ be the $j$-th Landau level, indexed from $j=0$ and without counting its infinite multiplicity, of the general Landau-Schr\"odinger operator together with an energy shift
\begin{align*}
  \mathfrak{h}:=\begin{pmatrix}
            -i\partial_{x_1}+\frac{B}{2}x_2\\ -i\partial_{x_2} -\frac{B}{2}x_1
        \end{pmatrix}^TA \begin{pmatrix}
            -i\partial_{x_1}+\frac{B}{2}x_2\\ -i\partial_{x_2} -\frac{B}{2}x_1
        \end{pmatrix} + \sum_{\substack{\beta,\gamma\in \N^2\\\beta=0,\;|\gamma|_1=1}} \left<w, {\rm Im}(P_{\beta,\gamma,0}^\bot h_{e,0} P_{\gamma,\beta,0}^\bot) w\right>_{L^2_{\rm per}}.
\end{align*}
Then there exist some normalized eigenfunctions $v_j\in \cS(\R^2)$ of $\mathfrak{h}$ with eigenvalue $\mu_j$ such that, for $\eps$ small enough, the functions
\begin{align*}
    \Phi_{j,\eps}:=\Phi_\eps\Big((U^{(0)}_\eps+\sqrt{\eps} U^{(1)}_\eps+\eps U^{(2)}_\eps)(w\otimes v_j)\Big),
\end{align*}
satisfy
\begin{align*}
    \|(H_\eps-e_0-\eps \mu_j)\Phi_{j,\eps}\|_{L^2(\R^2;\C^n)}=\cO(\eps^{\frac{5}{4}})
\end{align*}
with
\begin{align*}
    \left\|\Phi_{j,\eps}\right\|_{L^2(\R^2;\C^n)}=\frac{1}{|\Omega|^{1/2}}+\cO(\sqrt{\eps}).
\end{align*}
\end{theorem}

\begin{proof}
 The proof is essentially the same as for Theorem \ref{th:harmonic-oscillator}, since $\mathfrak{h}$ has some exponentially decaying eigenfunctions $v_j$ at the Landau level $\mu_j$.
\end{proof}

\begin{remark}
    Analogously to Remark \ref{rem:harmonic-oscillator},it is straightforward to see that the effective Hamiltonian $ \mathfrak{h}$ defined in Theorem \ref{th:quantum-hall} is a Landau-Schr\"odinger operator.
\end{remark}

\subsection{Approximate eigenpairs in some periodic crystals}
Now we consider the periodic operator defined in Section \ref{sec:4.1}:
\begin{align*}
    H_0=-\frac{1}{2}\Delta_\bx+V_{\rm per}(\bx),\qquad \bx\in \R^2.
\end{align*}
Then,
\begin{corollary}\label{cor:quantumhall}
   Assume that $H_0$ satisfies Assumption \ref{ass:4.4}. Then after adding a weak perpendicular magnetic field,
\begin{align*}
    H_\eps:=\frac{1}{2}\left(-i\partial_{x_1}+\frac{B}{2}\eps x_2\right)^2+\frac{1}{2}\left(-i\partial_{x_2}-\frac{B}{2}\eps x_1\right)^2+V_{\rm per}(\bx)
\end{align*}
is an operator satisfying Assumption \ref{ass:quantum-hall} around $(\bk_0,\bX_0=0)$. Furthermore, Theorem \ref{th:quantum-hall} holds with
\begin{align*}
  \mathfrak{h}:=\begin{pmatrix}
            -i\partial_{x_1}+\frac{B}{2}x_2\\ -i\partial_{x_2} -\frac{B}{2}x_1
        \end{pmatrix}^TA \begin{pmatrix}
            -i\partial_{x_1}+\frac{B}{2}x_2\\ -i\partial_{x_2} -\frac{B}{2}x_1
        \end{pmatrix} + \sum_{\beta=0,\;|\gamma|_1=1} \left<w, {\rm Im}(P_{\beta,\gamma,0}^\bot h_{e,0} P_{\gamma,\beta,0}^\bot) w\right>_{L^2_{\rm per}}
\end{align*}
where $A=\frac{1}{2}\nabla_k^2 E_\j(k_0)$.
\end{corollary}
\begin{proof}
    According to \eqref{op-h}, it suffices to consider the family of operators
    \begin{align*}
        h(\bk,\bX)=\frac{1}{2}\left(-i\partial_{x_1}+k_1+\frac{B}{2}X_2\right)^2+\frac{1}{2}\left(-i\partial_{x_2}+k_2-\frac{B}{2}X_1\right)^2+V_{\rm per}(\bx).
    \end{align*}
Clearly, according to the definition of $h_0$ in \eqref{eq:4.3},
\begin{align*}
h(\bk,\bX)=h_0\left(k_1+\frac{B}{2}X_2,k_2-\frac{B}{2}X_1\right)  
\end{align*}
 satisfies Assumption \ref{ass:quantum-hall}. Thus the corollary follows.
    
\end{proof}

\subsection{Explanation of the energy shift}\label{sec:shift-of-energy}
We now explain the role of the energy shift:
\begin{align*}
    \sum_{\beta=0,\;|\gamma|_1=1}\left<w,{\rm Im}(P_{\beta,\gamma,0}^\bot h_{e,0} P_{\gamma,\beta,0}^\bot) w\right>_{L^2_{\rm per}}.
\end{align*}
Here we consider a massive Landau–Dirac operator and show that this energy shift is precisely the Zeeman effect in that model.

More precisely, we consider the massive Landau–Dirac operator
\begin{align*}
    H_\eps:={\pmb\sigma}\cdot (-i\nabla_x-\bA(\eps \bx))+\sigma_3
\end{align*}
where
\begin{align*}
    \bA(\bX):=\frac{B}{2}\begin{pmatrix}
    -X_2\\
    X_1
\end{pmatrix},\qquad B>0,
\end{align*}
and the Pauli matrices $\pmb \sigma:=(\sigma_1,\sigma_2)$ and $\sigma_3$ are
\begin{align*}
    \sigma_1:=\begin{pmatrix}
        0&1\\1&0
    \end{pmatrix},\qquad  \sigma_2:=\begin{pmatrix}
        0&-i\\i&0
    \end{pmatrix},\qquad  \sigma_3:=\begin{pmatrix}
        1&0\\0&-1
    \end{pmatrix}.
\end{align*}
It is well known that the spectrum of $H_\eps$ is made up of eigenvalues of infinite multiplicity with the so-called Landau-Dirac levels (see e.g., \cite{Thaller-Dirac-1992}):
\begin{align}\label{eq:eigenvalue-LD}
    \mu_{\eps,j}:=\begin{cases}
        \sqrt{2B\eps j+1},\qquad j\in \N\\
        -\sqrt{2B\eps |j|+1},\qquad -j\in \N^+.
    \end{cases}
\end{align}

For simplicity, we consider only the positive eigenvalues $\mu_{\eps,j}$ around energy level $e_0=1$ for $\eps$ small enough. Using Theorem \ref{th:quantum-hall}, we have the following.
\begin{corollary}\label{cor:massive-LD}
    Let $j\in \N$ and $e_0=1$. Then for $\eps$ small enough,
\begin{align*}
    \mu_{\eps,j}=1+ \eps\mu_{j}\left(H^{\rm S}-\frac{B}{2}\right)+\cO(\eps^{\frac{5}{4}}).
\end{align*}
where $\mu_{j}\left(H^{\rm S}-\frac{B}{2}\right)$ is the $j$-th Landau level of the operator $H^{\rm S}-\frac{B}{2}$ and
$ H^{\rm S}$ is the Landau-Schr\"odinger operator in standard quantum hall effect:
\begin{align*}
    H^{\rm S}:=   \frac{1}{2}\left(-i\nabla-\bA(x)\right)^2
\end{align*}
with eigenvalues 
\begin{align*}
     \mu_{j}^{\rm S}:=B\left(j+\frac{1}{2}\right),\qquad j\in \N.
\end{align*}
In this case, the energy shift is
\begin{align*}
   \sum_{\beta=0,\;|\gamma|_1=1}{\rm Im} \left<w, P_{\beta,\gamma,0}^\bot h_{e,0} P_{\gamma,\beta,0}^\bot w\right>_{L^2_{\rm per}}=-\frac{B}{2}.
\end{align*}
\end{corollary}
Corollary \ref{cor:massive-LD} can be easily verified: using \eqref{eq:eigenvalue-LD}, for any $j\in \N$ and $\eps$ small enough,
\begin{align*}
    \mu_{\eps,j}=1+ \eps\left( \mu_{j}^{\rm S}-\frac{B}{2}\right)+\cO(\eps^2)=1+\eps\mu_{j}\left(H^{\rm S}-\frac{B}{2}\right)+\cO(\eps^2).
\end{align*}
In this case, the energy shift $-\frac{B}{2}$ is a consequence of the Zeeman effect in the following Pauli operator:
\begin{align}\label{eq:zeeman}
   \frac{1}{2} \mathcal{D}^* \mathcal{D}-H^{\rm S}=-\frac{B}{2}
    \end{align}
    or equivalently
  \begin{align}\label{eq:zeeman'}
   H^{\rm S}-\frac{1}{2} \mathcal{D}\mathcal{D}^* =-\frac{B}{2}
    \end{align}  
where 
\begin{align*}
    \mathcal{D}=\left(-i\partial_{x_1}+\frac{B}{2}x_2\right)+i\left(-i\partial_{x_2}-\frac{B}{2}x_1\right).
\end{align*}

\begin{proof}
 Recall that
\begin{theorem}\cite[Theorem 5.9]{Hislop-Sigal-introduction-1996}\label{th:spectrum}
    Let $A$ be self-adjoint. If for some $\delta>0$ there exists some $u\in D(A)\setminus\{0\}$ such that
    \begin{align*}
        \|(A-\lambda) u\|\leq \delta \|u\|,
    \end{align*}
    then $\sigma(A)\cap [\lambda-\delta,\lambda+\delta]\neq \emptyset$, that is, $A$ has spectrum inside $[\lambda-\delta,\lambda+\delta]$.
\end{theorem}
Then it suffices to show that 
    for $\eps$ small enough,
  \begin{align*}
    \|(H_\eps-e_0-\eps \mu_j)\Phi_\eps\|_{L^2(\R^2;\C^2)}=\cO(\eps^{\frac{5}{4}})
\end{align*}
with
\[
\Phi_\eps=\Phi_\eps\Big((U^{(0)}_\eps+\sqrt{\eps} U^{(1)}_\eps+\eps U^{(2)}_\eps)(w\otimes v_j)\Big)
\]
and
\begin{align*}
   \|\Phi_\eps\|_{L^2(\R^2)}=\frac{1}{|\Omega|^{1/2}}+\cO(\sqrt{\eps})
\end{align*}
where
\begin{align*}
    w=\frac{1}{|\Omega|^{1/2}}\begin{pmatrix}
        1\\  0
    \end{pmatrix}
\end{align*}
and $(v_j,\mu_j)$ being the $j$-th eigenpair of the Landau-Schr\"odinger operator with an energy shift
\begin{align*}
     \mathfrak{h}(m=2)= \frac{1}{2}\left(-i\partial_{x_1}+\frac{1}{2} Bx_2\right)^2+\frac{1}{2}\left(-i\partial_{x_2}-\frac{1}{2} Bx_1\right)^2-\frac{1}{2}B.
\end{align*}

For the operator $H_\eps$, by \eqref{op-h}, the corresponding symbol is 
\begin{align*}
    h(\bk,\bX)={\pmb\sigma}\cdot(-i\nabla_\bx+\bk-\bA(\bX))+\sigma_3.
\end{align*}
We are particularly interested in some eigenvalues near $\pm 1$ around $\bk_0=\bX_0=0$, that is
\begin{align*}
    E^\pm(\bk,\bX)=\pm\sqrt{\big|\bk-\bA(\bX)\big|^2+1}.
\end{align*}
Then
\begin{align*}
    w^+=\frac{1}{|\Omega|^{1/2}}\begin{pmatrix}
        1\\  0
    \end{pmatrix},\qquad  w^-=\frac{1}{|\Omega|^{1/2}}\begin{pmatrix}
        0\\ 1
    \end{pmatrix}
\end{align*}
are the normalized eigenfunctions of $h(0,0)=\sigma_3$ with eigenvalues $\pm 1$.

In this case, as the eigenvalues $E^+(\bk,\bX)$ and $E^-(\bk,\bX)$ are split around the point $(\bk_0,\bX_0)$, we focus on the case $e_0:=1$, the other case is analogous. In this case, 
\begin{align*}
    w=w^+.
\end{align*}
Let
\begin{align*}
    \hat{D}(\bk,\bX):=\left(k_1+\frac{1}{2}BX_2\right)+i\left(k_2-\frac{1}{2}BX_1\right).
\end{align*}
Then the targeted eigenvalue is
\begin{align*}
    E^+(\bk,\bX)=\sqrt{\big|\bk-\bA(\bX)\big|^2+1}=e_0+\frac{1}{2}|\hat{D}(\bk,\bX)|^2+\cO(|\bk|^3+|\bX|^3),
\end{align*}
and the corresponding normalized eigenfunction of $h(\bk,\bX)$ is
\begin{align*}
    \phi^+(\bk,\bX)&=\frac{c(k,X)}{|\Omega|^{1/2}}\left(\frac{1}{2}+\frac{h(\bk,\bX)}{2|h(\bk,\bX)|}\right)\begin{pmatrix}
       1\\0
    \end{pmatrix} \\
    &=\frac{c(\bk,\bX)}{2|\Omega|^{1/2}}\begin{pmatrix}
        1\\0
    \end{pmatrix}+\frac{c(\bk,\bX)}{2|\Omega|^{1/2}}\begin{pmatrix}
        \frac{1}{\sqrt{1+|\hat{D}(\bk,\bX)|^2}}\\
           \frac{\hat{D}(k,\bX)}{{\sqrt{1+|\hat{D}(\bk,\bX)|^2}}}    \end{pmatrix}\\
           &=\frac{1}{2|\Omega|^{1/2}}\begin{pmatrix}
               2-\frac{1}{4}|\hat{D}(\bk,\bX)|^2\\\hat{D}(\bk,\bX)
           \end{pmatrix} +\cO(|\bX|^3+|\bk|^3)
\end{align*}
with the normalized factor
\begin{align*}
    c(\bk,\bX)=\frac{1}{\sqrt{\frac{1}{2}+\frac{1}{2\sqrt{1+|\hat{D}(\bk,\bX)|^2}}}}=1+\frac{1}{8}|\hat{D}(\bk,\bX)|^2+\cO(|k|^3+|X|^3).
\end{align*}
Thus Assumption \ref{ass:bandstructure-nondegenerate} is satisfied  with $m=2$, $J=1$,  the scalar function
\begin{align*}
    \vec{\alpha}(\bk,\bX)=1,
\end{align*}
and a scalar homogeneous polynomial function
\begin{align*}
     f^{\rm eff}_m(\bk,\bX)= \frac{1}{2}|\hat{D}(\bk,\bX)|^2.
\end{align*}

In this case, 
\begin{align*}
    P^+(\bk,\bX)=\left| \phi^+(\bk,\bX)\right>\left< \phi^+(\bk,\bX)\right|,
\end{align*}
and for $|\beta|_1=0$ and $|\gamma|_1=1$,
\begin{align*}
    P_{\beta,\gamma,0}^\bot w^+=-\frac{1}{2|\Omega|^{1/2}}\begin{pmatrix}
        0\\ \partial_X^\gamma\hat{D}(0,0)
    \end{pmatrix},\qquad  P_{\gamma,\beta,0}^\bot w^+=-\frac{1}{2|\Omega|^{1/2}}\begin{pmatrix}
        0\\ \partial_k^\gamma\hat{D}(0,0)
    \end{pmatrix}.
\end{align*}
A direct calculation then gives the energy shift:
\begin{align*}
    \sum_{\beta=0,\;|\gamma|_1=1} \left<w^+, {\rm Im}(P_{\beta,\gamma,0}^\bot h_{e,0} P_{\gamma,\beta,0}^\bot) w^+\right>_{L^2_{\rm per}}=-\frac{1}{2}B.
\end{align*}
Thus according to Assumption \ref{ass:m} and Theorem \ref{th:m}, the effective Hamiltonian is Landau-Schr\"odinger operator with energy shift
\begin{align*}
    \mathfrak{h}(m=2)= \frac{1}{2}\left(-i\partial_{x_1}+\frac{1}{2} Bx_2\right)^2+\frac{1}{2}\left(-i\partial_{x_2}-\frac{1}{2} Bx_1\right)^2-\frac{1}{2}B.
\end{align*}
Therefore the corollary follows from Theorem \ref{th:quantum-hall} by choosing $(v_j,\mu_j)$ as the $j$-th eigenpair of the operator $  \mathfrak{h}(m=2)$.
\end{proof}

\section{Application: Landau-Dirac effective Hamiltonian} \label{sec:quantum-spin-hall}

In this section, we consider the application of Theorem \ref{th:m} for $m=1$. More precisely, we are interested in its application to Landau-Dirac operator relevant to honeycomb materials in $d=2$ within a magnetic field. We retain the notation $\bx$, $\bk$, and $\bX$ introduced in the preceding
section.

We need the following assumption.
\begin{assumption}\label{ass:conic}
Let $d=2$ and $\j\in \N^+$. We assume that the operator $H_\eps$ or equivalently the family of operators $(h(\bk,\bX))_{\bk,\bX}$ satisfies Assumption \ref{ass:bandstructure-nondegenerate} with the following property:
\begin{itemize}
    \item The eigenvalue $e_0=E_{\j}(\bk_0,\bX_0)$ is a degenerate eigenvalue of the operator $h(\bk_0,\bX_0)$ with multiplicity $J=2$;
    \item The normalized functions $w^+,w^-$ are orthogonal, and they are  eigenfunctions of $h(\bk_0,\bX_0)$ associated with eigenvalue $e_0$. We set
    \begin{align*}
        \vec{w}=(w^+,w^-)^T;
    \end{align*}
    \item The  homogeneous matrix-valued function $f^{\rm eff}_{m}$ of degree $m=1$ satisfies 
    \begin{align*}
        f^{\rm eff}_{m}(\bk,\bX)=
            \lambda_\#\left(\left(k_1+\frac{B}{2}X_2\right)\sigma_1-\left(k_2-\frac{B}{2}X_1\right)\sigma_2\right)
    \end{align*}
    for some constant $B\neq 0$ and $\lambda_\#\neq 0$.
    
\end{itemize}
\end{assumption}

Then,
\begin{theorem}\label{th:quantum-spin-hall}
We assume that $H_\eps$ satisfies Assumption \ref{ass:conic}. Let $U_\eps^{(0)}$ and $U_\eps^{(1)}$ be defined by \eqref{eq:U0-eps} and \eqref{eq:U1-eps} respectively, $\vec w$ be defined as in Assumption \ref{ass:conic}. 

Let $\mu_j$ be the $j$-th nonnegative relativistic Landau level for $j\in \N$ or $j$-th negative relativistic Landau level for $j\in -\N^+$ of the Landau-Dirac operator
\begin{align}\label{eq:Landau-Dirac}
  \mathfrak{h}:=\lambda_\# \left(\sigma_1\left(-i\partial_{x_1}+\frac{B}{2}x_2\right)-\sigma_2\left(-i\partial_{x_2}-\frac{B}{2}x_1\right)\right),\qquad \lambda_\#\neq 0.
\end{align}
Then there exist some eigenfunctions $\vec{v}_j\in \cS(\R^2)$ of $\mathfrak{h}$ with eigenvalue $\mu_j$ such that, for $\eps$ small enough, the functions
\begin{align*}
    \Phi_{j,\eps}:=\Phi_\eps\Big(\big(U^{(0)}_\eps+\sqrt{\eps} U^{(1)}_\eps\big)(\vec{w}\otimes \vec{v}_j)\Big)
\end{align*}
satisfy
\begin{align*}
    \|(H_\eps-e_0-\sqrt\eps \mu_j)\Phi_{j,\eps}\|_{L^2(\R^2;\C^n)}=\cO(\eps^{\frac{3}{4}})
\end{align*}
with
\begin{align*}
    \left\|\Phi_{j,\eps}\right\|_{L^2(\R^2;\C^n)}=\frac{1}{|\Omega|^{1/2}}+\cO(\sqrt{\eps}).
\end{align*}
\end{theorem}
\begin{proof}
     The proof is essentially the same as for Theorem \ref{th:harmonic-oscillator}. In addition, we refer to \cite[Ch. 7.1]{Thaller-Dirac-1992} for suitable eigenfunctions of the Landau-Dirac operator which are smooth and decay exponentially. Thus Assumption \ref{ass:m} is satisfied.
\end{proof}

\subsection{Approximate eigenpairs in 2D honeycomb materials}\label{sec:6.1}
In this part, we show that Assumption \ref{ass:conic} can be fulfilled by the 2D honeycomb materials with constant magnetic field $B$. 

Before going further, we first define the honeycomb potential with $d=2$. Let
\begin{align*}
    \ba_1=a\begin{pmatrix}
        \frac{\sqrt{3}}{2}\\ \frac{1}{2}
    \end{pmatrix},\qquad \ba_2=a\begin{pmatrix}
        \frac{\sqrt{3}}{2}\\-\frac{1}{2}
    \end{pmatrix},\qquad a>0,
\end{align*}
and
\begin{align*}
    \bk_1=q\begin{pmatrix}
        \frac{1}{2}\\
        \frac{\sqrt{3}}{2}
    \end{pmatrix},\qquad \bk_2=q\begin{pmatrix}
        \frac{1}{2}\\ -\frac{\sqrt{3}}{2}
    \end{pmatrix},\qquad q=\frac{4\pi}{a\sqrt{3}}.
\end{align*}
Then the lattice and dual lattice in graphene are respectively
\begin{align*}
    \mathbb L=\ba_1\mathbb Z+\ba_2\mathbb Z,\qquad \mathbb L^*=\bk_1\mathbb Z+\bk_2\mathbb Z.
\end{align*}
The vertices of $\Omega^*$ are given by
\begin{align*}
    \bK_1:=\frac{1}{3}(\bk_1-\bk_2),\qquad \bK_2:=-\bK_1.
\end{align*}
Concerning the potential $V_{\rm per}$, we define the rotation matrix $\mathcal{R}$ by
\begin{align*}
    \mathcal{R}:=\begin{pmatrix}
        -\frac{1}{2} & \frac{\sqrt{3}}{2}\\ -\frac{\sqrt{3}}{2} &-\frac{1}{2}
    \end{pmatrix}.
\end{align*}
It rotates a vector in $\R^2$ clockwise by $\frac{2\pi}{3}$. According to \cite[Proposition 2.3]{Fefferman-Weinstein-Honeycomb-2012} a honeycomb lattice potential is defined by
\begin{align*}
    V_{\rm hon}(\bx)=V_0+\sum_{G\in \mathbb L^*}V_{\bG} \cos(\bG\cdot \bx)
\end{align*}
with
\begin{align*}
    V_\bG=V_{\mathcal{R}\bG}=V_{\mathcal{R}^2\bG}.
\end{align*}

Now we define the honeycomb Hamiltonian
\begin{align*}
    H_{\rm hon}=-\Delta+\eta V_{\rm hon}(\bx).
\end{align*}
According to \eqref{op-h0}, it suffices to consider the operator 
\[
h_{\rm hon}(\bk)=(-i\nabla_\bx+\bk)^2+\eta V_{\rm hon}(\bx)
\]
on $L^2_{\rm per}$. Let $(\phi_{{\rm hon},j}(\bk),E_{{\rm hon},j}(\bk))$ be an eigenpair of the operator $h_{\rm hon}(\bk)$ with 
\[
E_{{\rm hon},1}\leq E_{{\rm hon},2}\leq \cdots.
\]
Then,
\begin{proposition}\cite{Fefferman-Weinstein-Wavepackets-2014}\label{prop:graphene}
 Let $\bK\in \{\bK_1, \bK_2\}$, $\bG_1=\bk_1+\bk_2\in \mathbb L^*$, and $V_{\bG_1}>0$.  
 Then there exists a countable and closed set $\widetilde{\mathcal{C}}\subset \R$ such that for all $\eta\not\in \widetilde{\mathcal{C}}$, there exists $\j\in \N^+$ such that 
 \begin{itemize}
     \item $e_0=E_{{\rm hon},\j}(\bK)$ is an eigenvalue of multiplicity $J=2$;
     \item for $|k-K|$ small enough and $j=1,2$,
\begin{align*}
    E_{{\rm hon},\j+j-1}(\bk)-e_0=(-1)^j|\lambda_\#||\bk-\bK|+\cO(|\bk-\bK|^2)
\end{align*}
with some constant $\lambda_\#\neq 0$;
\item for $j=1,2$,
\begin{align}\label{eq:ass1-2-conic}
    \left\|\phi_{{\rm hon},\j+j-1}(\bk)-\vec{w}^T\vec \alpha_{{\rm hon},j}(\bk) \right\|_{L^2_{\rm per}}=\cO(|\bk-\bK|)
\end{align}
where 
\begin{align*}
    \vec{w}=(\phi_{{\rm hon},\j}(\bK), \phi_{{\rm hon},\j+1}(\bK))^T
\end{align*}
and $\alpha_j(k)$ satisfies
\begin{align*}
     h_{{\rm hon}}^{\rm eff}(\bk)\vec \alpha_{{\rm hon},j}(\bk)=(-1)^{j}|\lambda_\#||\bk-\bK| \vec\alpha_{{\rm hon},j}(\bk)
\end{align*}
with
\begin{align*}
    h_{{\rm hon}}^{\rm eff}(\bk)=\lambda_\#((k_1-K_1)\sigma_1-(k_2-K_2)\sigma_2).
\end{align*}
 \end{itemize}
\end{proposition}
\begin{proof}
The proposition can be found in \citep[Theorem 3.2]{Fefferman-Weinstein-Wavepackets-2014}, and its proof is based on \cite{Fefferman-Weinstein-Honeycomb-2012}. In particular, \eqref{eq:ass1-2-conic} follows from \cite[Eq. (3.13)]{Fefferman-Weinstein-Wavepackets-2014} where $p_\pm(\bk)$ there are precisely the functions  $\phi_{{\rm hon},\j+j-1}(\bk)$ with $j=1,2$.
\end{proof}

As a result,
\begin{corollary}\label{cor:quantumspinhall}
Let $\bk_0\in \{\bK_1,\bK_2\}$, $\bX_0=0$ and let $\eta$ be a constant chosen as in Proposition \ref{prop:graphene}. Let
    \begin{align*}
        H_\eps=(-i\nabla_\bx-\bA (\eps \bx))^2+\eta V_{\rm hon}(\bx).
    \end{align*}
    with
    \begin{align*}
        \bA(\bX)=\frac{B}{2}\begin{pmatrix}
            -X_2\\X_1
        \end{pmatrix},\qquad B\in \R\setminus\{0\}.
    \end{align*}
    Then $H_\eps$ is an operator satisfying Assumption \ref{ass:conic} around $(k_0,X_0)$. Furthermore, Theorem \ref{th:quantum-spin-hall} holds with
    \begin{align*}
  \mathfrak{h}:=\lambda_\# \left(\sigma_1\left(-i\partial_{x_1}+\frac{B}{2}x_2\right)-\sigma_2\left(-i\partial_{x_2}-\frac{B}{2}x_1\right)\right).
\end{align*}
\end{corollary}
\begin{proof}
 This is a consequence of the fact that 
    \begin{align*}
        h(\bk,\bX)= h_{\rm hon}(\bk-\bA(\bX)).
    \end{align*}
    Let
    \begin{align*}
    \vec w=\Big(\phi_{{\rm hon},\j}(\bk_0-\bA(\bX_0)),\phi_{{\rm hon},\j+1}(\bk_0-\bA(\bX_0))\Big)^T.
\end{align*}
Then $(\phi_j(\bk,\bX),E_j(\bk,\bX))=(\phi_{{\rm hon},j}(\bk-\bA(\bX)),E_{{\rm hon},j}(\bk-\bA(\bX)))$ is an eigenpair of $h(\bk,\bX)$. Thus
\begin{align*}
    \vec{\alpha}_{j}(\bk,\bX)=\vec{\alpha}_{{\rm hon},j}(\bk-\bA(\bX))
\end{align*}
and
\begin{align*}
    f^{\rm eff}_1(\bk,\bX)=\lambda_\#\left(\sigma_1\left(k_1+\frac{B}{2}X_2\right)-\sigma_2\left(k_2-\frac{B}{2}X_1\right)\right).
\end{align*}
Thus Assumption \ref{ass:conic} is satisfied with $\bk_0\in \{\bK_1,\bK_2\}$ and $\bX_0=0$. The corollary then follows from Theorem \ref{th:quantum-spin-hall}.
\end{proof}

\section{Application: almost flat-band}\label{sec:almost-flat-band}
In this section, we consider the almost flat-band property of some commensurate crystals. As explained in Introduction, these commensurate systems are used in physics to approximate incommensurate systems via the supercell approach. Here we show that under the same assumptions as in Theorem \ref{th:m}, there exists almost flat-band structure of these commensurate supercells.

We consider the Hamiltonian
\begin{align*}
    H_\eps:=T(-i\nabla_x+\bA(x,\eps x))+V(x,\eps x).
\end{align*}
Recall that $x\mapsto\bA(x,\cdot)$ and $x\mapsto V(x,\cdot)$ are $\mathbb L$-periodic with $ \mathbb L=\sum_{j=1}^d a_j \Z$. Concerning commensurate/incommensurate system, we also need periodic assumptions on the $X$ variable. Without loss of generality, we consider the following simple case.
\begin{assumption}\label{ass:periodicpotentials}We assume that
    $X\mapsto\bA(\cdot,X)$ and $X\mapsto V(\cdot,X)$ are $\mathbb L$ periodic.
\end{assumption}
In this case, if $\eps\in \R\setminus \mathbb Q$ is irrational, then $H_\eps$ is an incommensurate Hamiltonian; if $\eps\in \mathbb Q$ is rational, then we can write $\eps=\frac{p}{q}$ with $p,q\in \mathbb Z$ and $q\neq 0$, and $H_\eps$ is a commensurate Hamiltonian with $ \mathbb L_q$-periodicity, where
\begin{align*}
 \mathbb L_q  = q\mathbb L:=\{q R;\;R\in\mathbb L\}.
\end{align*}

Concerning $\eps\in  \R\setminus \mathbb Q$, we can split it as
\begin{align*}
    \eps=\widetilde{\eps}+\eps_n
\end{align*}
where 
\begin{itemize}
    \item $\widetilde{\eps}\in \mathbb Q$ is fixed such that $\eps-\widetilde{\eps}$ is small enough;
    \item $\eps_n \in \mathbb Q$ is a sequence such that $\eps_n \to \eps-\widetilde{\eps}$. 
\end{itemize}
In physics, to study $H_\eps$, the supercell approach is used: we use periodic system $H_{\widetilde{\eps}+\eps_n}$ to approach $H_\eps$.

In this section, without loss of generality , we assume $\widetilde{\eps}=0$ and study the Hamiltonian $H_{\eps_n}$. If $\widetilde{\eps}\neq 0$, we replace $H_0$ by $H_{\widetilde{\eps}}$, and regard  $\mathbb L$ as the lattice for $H_{\widetilde{\eps}}$.

\subsection{Commensurate system}
In this subsection, we recall some basic concepts for commensurate system.

Let $\eps=\frac{p}{q}\in \mathbb Q$ with $p,q\in \mathbb Z$ and $q\neq 0$. Then
\begin{align*}
    x\mapsto \bA(x,\eps x),\qquad x\mapsto V(x,\eps x)
\end{align*}
are $\mathbb L_q$-periodic. Thus $H_\eps$ is $ \mathbb L_q$-periodic. The corresponding dual lattice is
\begin{align*}
    \mathbb L^*_q:=q^{-1}\mathbb L^*,
\end{align*}
the unit cell $\Omega_q$ and first Brillouin zone $\Omega_q^*$ are
\begin{align}\label{eq:unit-cell-q}
    \Omega_q= q\Omega,\qquad \Omega_q^*=q^{-1}\Omega^*.
\end{align}

We now define the Bloch transform $\cU_q$ on $\mathbb L_q$:
\begin{align*}
&\forall u \in C^\infty_{\rm c}(\R^d;\C), \quad (\cU_q u)_k(x) = \sum_{R_q \in \mathbb{L}_q} u(x+R_q) e^{-ik \cdot (x+R_q)}
\end{align*}
Then analogous to \eqref{op-h0}, we have the decomposition
\begin{align}\label{eq:Bloch-transform-q}
    H_\eps:=\cU^{-1}_q\left( \fint_{\Omega_q^*}^\oplus h_q(k)dk \right)\cU_q
\end{align}
with
\begin{align}\label{eq:h-q}
    h_q(k):=T(-i\nabla_x+k+\bA(x,\eps x))+V(x,\eps x).
\end{align}
The operator $h_q(k)$ is acting on $ L^2_{q,\rm per}$, defined by
\begin{align*}
    L^2_{q,\rm per}:=\{f \in L^2_{\rm loc}(\mathbb{R}^d) \; |  \; \forall R_q\in\mathbb{L}_q, 
\; f(x-R_q)= f(x) \mbox{ for a.a. } x \in \R^d\}.
\end{align*}
We now study the almost-flat band property of the band structure of the mapping $k\mapsto h_q(k)$. 

\subsection{Almost flat-band}
The main result of this section is the following.
\begin{theorem}\label{th:almost-flat}
Assume that 
\begin{itemize}
    \item the potentials $\bA$ and $V$ satisfy Assumption \ref{ass:periodicpotentials};
    \item the operator $H_\eps$ satisfies  Assumption \ref{ass:bandstructure-nondegenerate};
    \item there exists an eigenpair $(\vec{v}_*,\mu_*)$ satisfying  Assumption \ref{ass:m} with $m=1$ or Assumption \ref{ass:m} with $m=2$.
\end{itemize}
Let $\eps:=\frac{p}{q}\in \mathbb Q$ be small enough with $p,q\in \N$ and $q\neq 0$. Then there exists a family of approximate Bloch eigenfunctions $\Psi_\eps(k) \in L^2_{\rm per}(\Omega_q;\C^n)$ and $\mu \in \R$ such that,
\begin{align}\label{eq:eigenfunction-flatband}
  \sup_{k\in \Omega_q^*} \| (h_q(k)-(e_0+\eps^{\frac{m}{2}}\mu))\Psi_\eps(k)\|_{L^2_{q,\rm per}}=\cO(\eps^{\frac{m}{2}+\frac{1}{4}}).
\end{align}
Furthermore there exists an almost flat band of $h_q(k)$ in the following sense:
\begin{align}\label{eq:flatband}
   \sup_{k\in \Omega^*_q} {\rm dist}\Big(e_0+\eps^{\frac{m}{2}}\mu,\sigma(h_q(k))\Big)=\cO(\eps^{\frac{m}{2}+\frac{1}{4}}).
\end{align}
\end{theorem}

\begin{proof}
Without loss of generality, we assume $H_\eps$ satisfies Assumption \ref{ass:m} with $m=1$. The case  Assumption \ref{ass:m} with $m=2$ can be proved in the same manner.

Let 
\begin{align*}
\Phi_\eps(x):=\Phi_\eps\Big((U_\eps^{(0)}+\sqrt\eps U_\eps^{(1)}\big)(\vec{w}\otimes \vec{v}_*)\Big)(x)
\end{align*}
be given as in Theorem \ref{th:m}. We first claim that for any $k\in \Omega_q^*$,
\begin{align*}
    \Psi_\eps(k,x):=(\cU_q \Phi_\eps)_k(x)=\sum_{R_q \in \mathbb{L}_q} \Phi_\eps(x+R_q) e^{-ik \cdot (x+R_q)} \in L^2_{q,\rm per}
\end{align*}
satisfies \eqref{eq:eigenfunction-flatband}. Let $\cT_{R}$ be a translation operator defined by
\begin{align*}
    \cT_{R} u(x)=u(x- R),\qquad u\in L^2(\R^d).
\end{align*}
According to the periodicity of $H_\eps$, for any $R_q\in \mathbb L_q$,
\begin{align*}
    [H_\eps,\cT_{R_q}]=0.
\end{align*}
Then by Theorem \ref{th:m},
\begin{align*}
    \|(H_\eps-e_0-\eps^{\frac{m}{2}}\mu) \cT_{R_q}\Phi_{\eps}\|_{L^2(\R^d;\C^n)}=\|\cT_{R_q} (H_\eps-e_0-\eps^{\frac{m}{2}}\mu) \Phi_{\eps}\|_{L^2(\R^d;\C^n)}=\cO(\eps^{\frac{m}{2}+\frac{1}{4}}).
\end{align*}
Next, by the definition \eqref{eq:Phi-eps} of $\Phi_\eps(\cdot)$, we have
\begin{align*}
    \Supp(\Phi_\eps)\subset B_{2\eps^{-s_1}}(\eps^{-1}X_0).
\end{align*}
Thus for any $R_q\neq R_q'$ and $R_q,R_q'\in \mathbb L_q$,
\begin{align}\label{eq:7.6}
    \Supp(\Phi_\eps(\cdot-R_q)) \cap  \Supp(\Phi_\eps(\cdot-R_q'))=\emptyset
\end{align}
since for $\eps$ small enough,
\begin{align*}
    |R_q'-R_q|\geq q\,{\rm dist}(0,\mathbb L\setminus\{0\}) = \eps^{-1}p \;{\rm dist}(0,\mathbb L\setminus\{0\})\geq   \eps^{-1} {\rm dist}(0,\mathbb L\setminus\{0\})\gg 4\eps^{-s_1}.
\end{align*}
Consequently,
\begin{equation}\label{eq:7.7}
    \begin{aligned}
\MoveEqLeft \sup_{k\in \Omega_q^*}\| (h_q(k)-(e_0+\eps^{\frac{m}{2}}\mu))\Psi_\eps(k)\|_{L^2_{q,\rm per}}^2\\
&=\sup_{k\in \Omega_q^*}\sum_{R_q\in \mathbb L_q}\left\|(h_q(k)-(e_0+\eps^{\frac{m}{2}}\mu))\Big(\Phi_\eps(\cdot+R_q)e^{-ik\cdot(x+R_q)}\Big)\right\|_{L^2_{q,\rm per}}^2\\
 &=\sup_{k\in \Omega_q^*}\sum_{R_q\in \mathbb L_q}\|(H_\eps-(e_0+\eps^{\frac{m}{2}}\mu))\Phi_\eps(\cdot+R_q)\|_{L^2_{q,\rm per}}^2\\
 &=\|(H_\eps-(e_0+\eps^{\frac{m}{2}}\mu))\Phi_\eps\|_{L^2(\R^d;\C^n)}^2=\cO(\eps^{m+\frac{1}{2}}).
\end{aligned}
\end{equation}
This proves \eqref{eq:eigenfunction-flatband}.

It remains to prove \eqref{eq:flatband}. Using Theorem \ref{th:spectrum}, it suffices to show that there exists $C>0$ such that for $\eps$ small enough,
\begin{align*}
    \inf_{k\in \Omega_q^*}\|\Psi_\eps(k)\|_{L^2_{q,\rm per}}\geq C.
\end{align*}
Arguing as for \eqref{eq:7.7} and by \eqref{eq:7.6}, this follows from Theorem \ref{th:m} (for $m=1$) or Theorem \ref{th:m} (for $m=2$) that
\begin{align*}
    \inf_{k\in \Omega^*_q}\|\Psi_\eps(k)\|_{L^2_{q,\rm per}}=\|\Phi_\eps\|_{L^2(\R^d;\C^n)}=\frac{1}{|\Omega|^{1/2}}+\cO(\sqrt{\eps}).
\end{align*}
Hence \eqref{eq:flatband}. This completes the proof.
\end{proof}

\section*{Acknowledgements} 
This project is supported by the National Key Research and Development Program of China (2025YFA1016800). The author would like to thank Guo Chuan Thiang and Domenico Monaco for insightful discussions regarding the quantum Hall effect, as well as for noting the distinctions between this effect and the present work.

\appendix

\section{Estimate on some Bloch wavefunctions}\label{sec:estimate-Bloch}

We need an estimate on the Bloch transformed functions.
\begin{lemma}\label{lem:Uf-cH}
Let $n\geq d+1$ and let $f(x,y)\in W^{n,1}(\R^d_y;L^2_{{\rm per},x})$ with $x\mapsto f(x,\cdot)$ being $\mathbb L$-periodic. Let $g(x)=f(x,x)$. Then  for any $k\in \Omega^*$,
\begin{align*}
     \|(\cU g)_k-|\Omega|^{-1} \mathcal{F}_y(f)(\cdot,k)\|_{L^2_{\rm per}}\lesssim  \int_{\R^d} \left\|(-\Delta_y)^{n/2} f(\cdot,y)\right\|_{L^2_{\rm per}} dy
\end{align*}
and
\begin{align*}
     \|(\cU g)_k\|_{L^2_{\rm per}}\lesssim \|f\|_{W^{n,1}(\R^d;L^2_{\rm per})}
\end{align*}
where $\mathcal{F}_y(f)$ is the Fourier transform of $f$, i.e.,
\begin{align*}
  \mathcal{F}_y(f)(x,\xi):= \int_{\R^d} e^{-i\xi\cdot y}f(x,y)dy.
\end{align*}
\end{lemma}
\begin{proof}
According to Poisson summation formula,
\begin{align*}
  \MoveEqLeft  (\cU g)_k=\sum_{R\in\mathbb L} f(x,x+R)e^{-ik\cdot(x+R)}=|\Omega|^{-1}\sum_{G\in \mathbb{L}^*}e^{iG\cdot x}\mathcal{F}_y(f)(x,G+k)
\end{align*}
Thus 
\begin{align*}
    \|(\cU g)_k-|\Omega|^{-1}\mathcal{F}_y(f)(\cdot,k)\|_{L^2_{\rm per}}\lesssim  \sup_{k\in \Omega^*}\sum_{G\in \mathbb L^*\setminus\{0\}}\| \mathcal{F}_y(f)(\cdot,G+k)\|_{L^2_{\rm per}}.
\end{align*}
As $|G+k|^{n} e^{-i(G+k)\cdot y}= (-\Delta_y)^{n/2}e^{-i(G+k)\cdot y}$, by integration by parts, for $G\not\in \overline{\Omega^*}$, 
\begin{align*}
    \| \mathcal{F}_y(f)(\cdot,G+k)\|_{L^2_{\rm per}}&= \frac{1}{|G+k|^n}\left\|\mathcal{F}_y\big((-\Delta_y)^{n/2}f\big)(\cdot,G+k)\right\|_{L^2_{\rm per}}\\
    &\lesssim \frac{1}{|G+k|^n}\int_{\R^d} \left\|(-\Delta_y)^{n/2}f(\cdot,y)\right\|_{L^2_{\rm per}} dy.
\end{align*}
By \eqref{eq:Omega*}, 
\begin{align*}
   \inf_{ G\in \mathbb L^*\setminus\{0\}} {\rm dist}(G, \Omega^*)>0.
\end{align*}
Thus, for $n\geq d+1$
\begin{align*}
 \MoveEqLeft  \|(\cU g)_k-|\Omega|^{-1}\mathcal{F}_y(f)(\cdot,k)\|_{L^2_{\rm per}}\\
 &\lesssim \sup_{k\in \Omega^*} \sum_{G\in \mathbb L^*\setminus\{0\}} \frac{1}{|G+k|^{n}}  \int_{\R^d} \left\|(-\Delta_y)^{n/2}f(\cdot,y)\right\|_{L^2_{\rm per}} dy\\
    &\lesssim  \int_{\R^d} \left\|(-\Delta_y)^{n/2} f(\cdot,y)\right\|_{L^2_{\rm per}} dy.
\end{align*}
Furthermore, 
\begin{align*}
    \|(\cU g)_k\|_{L^2_{\rm per}}&\lesssim  \|\mathcal{F}_y(f)(\cdot,k)\|_{L^2_{\rm per}}+\int_{\R^d} \left\|(-\Delta_y)^{n/2} f(\cdot,y)\right\|_{L^2_{\rm per}} dy\\
    &\lesssim \|f\|_{L^1(\R^d;L^2_{\rm per})}+\|(-\Delta_y)^{n/2}  f\|_{L^1(\R^d;L^2_{\rm per})}\lesssim \|f\|_{W^{n,1}(\R^d;L^2_{\rm per})}.
\end{align*}
This ends the proof.
\end{proof}

\section{Scaling of the effective Hamiltonian}\label{sec:scaling}
Concerning our effective Hamiltonian, the part of $\vec{v}_\eps$ is of the form of
\begin{align*}
    \mathcal{F}^{-1}\Op(f)\mathcal{F}
\end{align*}
with some polynomial function $f$ of the form
\begin{align*}
    (k-k_0)^\beta(X-X_0)^\gamma
\end{align*}
for some $\beta,\gamma\in \N^d$. 

In this section, we consider the scaling and translation of this operator $ \mathcal{F}^{-1}\Op(f)\mathcal{F}$. This scaling and translation rely on the homogeneity of the polynomial function.
\begin{lemma}[Scaling and translation of the effective Hamiltonian]\label{lem:scaling}
Let $(k,X)\mapsto f(k,X)$ be a polynomial function such that 
\begin{align*}
    f(k,X):=f_m^{\rm hom}(k-k_0,X-X_0)
\end{align*}
with some homogeneous function $f_m^{\rm hom}(k,X)$ w.r.t $(k,X)$ of degree $m\in\N^+$. Then,
\begin{align}\label{eq:scaling}
      \big[\mathcal{F}^{-1}\Op(f)\mathcal{F} T_\eps(\vec{v})\big](x)= \eps^{\frac m2}T_\eps\big(\mathcal{F}^{-1}{\rm Op}_1(f_m^{\rm hom})\mathcal{F} \vec{v}\big)(x)
\end{align}
where ${\rm Op}_1(f_m^{\rm hom})$ is the Weyl quantization of the symbol $f_m^{\rm hom}$ with $\eps=1$ and $T_\eps$ is defined by \eqref{eq:scaling-translation}.
\end{lemma}
\begin{proof}
For any $k_0\in\R^d$,
\begin{align}\label{eq:fourier-v_epsilon-k0}
  \mathcal{F}(T_\eps(\vec{v}))\left(k\right)=\eps^{-\frac{d }{4}} e^{-i\frac{k\cdot X_0}{\eps}} \mathcal{F}(\vec{v})\left(\frac{k-k_0}{\sqrt\eps}\right).
\end{align}
Using this identity,
\begin{align*}
  \MoveEqLeft  \Op(f)\mathcal{F}(T_\eps(\vec v))(k)\\
  &= \frac{1}{(2\pi \eps)^d}\int_{\R^d\times \R^d}e^{-i\frac{(k-k')\cdot X}{\eps}}f^{\rm hom}_m\left(\frac{k+k'}{2}-k_0, X-X_0\right) \mathcal{F}(T_\eps(\vec v))(k')dk'dX\\
  &=\frac{\eps^{-\frac{d}{4}}}{(2\pi \eps)^d } \int_{\R^d\times \R^d}e^{-i\frac{(k-k')\cdot X +k'\cdot X_0}{\eps}}f_m^{\rm hom}\left(\frac{k+k'}{2}-k_0, X-X_0\right) \mathcal{F}(\vec{v})\left(\frac{k'-k_0}{\sqrt\eps}\right)dk'dX\\
   &=\frac{\eps^{-\frac{d}{4}}}{(2\pi \eps)^d }e^{-i\frac{k\cdot X_0}{\eps}}\int_{\R^d\times \R^d}e^{-i\frac{((k-k_0)-k')\cdot X }{\eps}}f^{\rm hom}_m\left(\frac{(k-k_0)+k'}{2}, X\right) \mathcal{F}(\vec{v})\left(\frac{k'}{\sqrt\eps}\right)dk'dX
\end{align*}
Thus by the homogeneity of $f_m^{\rm hom}$,
\begin{align*}
     \MoveEqLeft  \Op(f)\mathcal{F}(T_\eps(\vec v))(\sqrt\eps k+k_0)\\
     &=\frac{\eps^{-\frac{d}{4}}}{(2\pi \eps)^d }e^{-i\frac{(\sqrt{\eps}k+k_0)\cdot X_0}{\eps}}\int_{\R^d\times \R^d}e^{-i\frac{(\sqrt\eps k-k')\cdot X }{\eps}}f^{\rm hom}_m\left(\frac{\sqrt{\eps} k+k'}{2}, X\right) \mathcal{F}(\vec{v})\left(\frac{k'}{\sqrt{\eps}}\right)dk'dX\\
     &=\frac{\eps^{-\frac{d}{4}}}{(2\pi )^d }e^{-i\frac{(\sqrt{\eps}k+k_0)\cdot X_0}{\eps}}\int_{\R^d\times \R^d}e^{-i\frac{(\sqrt{\eps}k-\sqrt{\eps} k')\cdot \sqrt{\eps}X }{\eps}}f^{\rm hom}_m\left(\frac{\sqrt{\eps}k+\sqrt{\eps} k'}{2}, \sqrt{\eps}X\right) \mathcal{F}(\vec{v})\left(k'\right)dk'dX\\
      &= \eps^{\frac{m}{2}} \frac{\eps^{-\frac{d}{4}}}{(2\pi )^d }e^{-i\frac{(\sqrt{\eps}k+k_0)\cdot X_0}{\eps}}\int_{\R^d\times \R^d}e^{-i(k-k')\cdot X }f^{\rm hom}_m\left(\frac{k+ k'}{2}, X\right) \mathcal{F}(\vec{v})\left(k'\right)dk'dX\\
      &=\eps^{\frac{m}{2}}\eps^{-\frac{d}{4}} e^{-i\frac{(\sqrt{\eps}k+k_0)\cdot X_0}{\eps}} {\rm Op}_1(f^{\rm hom}_m)\mathcal{F}(\vec{v})(k).
\end{align*}
Thus,
\begin{align*}
    \Op(f)\mathcal{F}(T_\eps(\vec v))( k)=\eps^{\frac{m}{2}}\eps^{-\frac{d}{4}}e^{-i\frac{k \cdot X_0}{\eps}} {\rm Op}_1(f^{\rm hom}_m)\mathcal{F}(\vec{v})\left(\frac{k-k_0}{\sqrt{\eps}}\right).
\end{align*}
Then analogous to \eqref{eq:fourier-v_epsilon-k0}, we get \eqref{eq:scaling}.

\end{proof}

\medskip
\begin{refcontext}[sorting=nyt]
\printbibliography[heading=bibintoc, title={Bibliography}]
\end{refcontext}

\end{document}